\documentclass[%
 aip,jmp,
 amsmath,amssymb,
 reprint,%
]{revtex4-1}

\usepackage{graphicx}
\usepackage{tikz}
\usetikzlibrary{decorations,shapes,arrows,matrix,positioning}

\usepackage{dcolumn}
\usepackage{bm}
\usepackage{amsthm}
\usepackage[utf8]{inputenc}
\usepackage[T1]{fontenc}
\usepackage{mathptmx}
\usepackage{tocloft}
\usepackage{bbold}
\usetikzlibrary{backgrounds}
\usetikzlibrary{patterns}
\usetikzlibrary{arrows.meta}
\usetikzlibrary{trees,snakes}
\tikzstyle{bag} = [align=center]
\usetikzlibrary{decorations.pathmorphing}
\usepackage{xcolor}
\usepackage{fancyhdr}
\usepackage{amssymb}
\usepackage{enumitem}
\usepackage{xcolor}
\usepackage{slashed}
\usepackage{amsmath}
\usepackage[export]{adjustbox}
\usepackage{hyperref}
\usepackage{subcaption}

\newtheorem{theorem}{Theorem}
\newtheorem{proposition}{Proposition}
\newtheorem{corollary}{ Corollary}
\newtheorem{lemma}{Lemma}
\theoremstyle{definition}
\newtheorem{definition}{Definition}
\newtheorem{remark}{Remarks}

\begin{document}

\preprint{AIP/123-QED}

\title[Renormalization of $\phi_4^4$ theory on the half-space 
$\mathbb{R}^+ \times\mathbb{R}^3$ with flow equations II]
{The Surface counter-terms of the $\phi_4^4$ theory 
on the half space $\mathbb{R}^+ \times\mathbb{R}^3$}

\author{Majdouline Borji}
\email{majdouline.borji@polytechnique.edu}
\author{Christoph Kopper}%
\email{christoph.kopper@polytechnique.edu}
\affiliation{Centre de Physique Théorique CPHT, CNRS, UMR 7644, Institut Polytechnique de Paris, 91128 Palaiseau, France}%

\date{\today}

\begin{abstract}
In a previous work, we established perturbative renormalizability to all orders of the massive $\phi^4_4$-theory on a half-space also called the semi-infinite massive $\phi^4_4$-theory. Five counter-terms which are functions depending on the position in the space, were needed to make the theory finite. The aim of the present paper is to prove that these counter-terms are position independent (i.e. constants) for a particular choice of renormalization conditions. We investigate this problem by decomposing the correlation functions into a bulk part, which is defined as the $\phi^4_4$ theory on the full space $\mathbb{R}^4$ with an interaction supported on the half-space, plus a remainder which we call "the surface part". We analyse this surface part and establish perturbatively that the $\phi^4_4$ theory in $\mathbb{R}^+\times\mathbb{R}^3$ is made finite by adding the bulk counter-terms and two additional  counter-terms to the bare interaction in the case of Robin and Neumann boundary conditions. These surface counter-terms are position independent  and are proportional to $\int_S \phi^2$ and $\int_S \phi\partial_n\phi$. For Dirichlet boundary conditions, we prove that no surface counter-terms are needed and the bulk counter-terms are sufficient to renormalize the connected amputated (Dirichlet) Schwinger functions. A key technical novelty as compared to our previous work is a proof that the power counting of the surface part of the correlation functions is better by one power than their bulk counterparts.
\end{abstract}

\maketitle

\section{Introduction}
\indent Renormalization group methods in the presence of boundaries have been developed in theoretical physics in the context of the semi-infinite scalar field model in \cite{DiehlD1,DiehlD2}. Besides the physical significance of this model as a prototype for systems with spatial inhomogeneities, it serves to illustrate the characteristic new features of the RG when translation invariance is broken by the presence of surfaces. Recently, a rigorous proof of perturbative renormalizability of this model has been established \cite{BorjiKopper2} using the Polchinski flow equations. The loss of translation invariance implies that the relevant parameters are not constants but are rather functions, which can possibly depend on the position $x\in V$, the domain on which the system is defined. A related problem is the renormalizability of theories which break translation invariance not necessarily by the presence of boundaries, but in a different manner such as the geometry of the considered space or the support of the coupling. As an example, we mention the work on renormalization of finite temperature field theory \cite{temperature}. This theory breaks down the space-time $O(4)$-symmetry to the $\mathbb{Z}_2\times O(3)$-symmetry. It was proved to be temperature independently renormalizable to all orders in perturbation theory, in the sense that renormalization conditions can be fixed order by order such that the counter-terms are temperature independent. There exists also a work by Kopper and Müller \cite{KopperMuller} which rigorously establishes the renormalizability of the scalar field theory on a Riemanian manifold $\mathcal{M}$ without a boundary. The breaking of translation invariance implies that the counter-terms are functions which depend on the position in $\mathcal{M}$. However, it was proven that for a particular choice of renormalization conditions, these counter-terms can be chosen as constants. The common feature of these two examples is that in a first step they are renormalized by position dependent counter-terms (or temperature dependent counter-terms in the first example), but still this dependence on the position can be reabsorbed in the renormalization conditions such that the counter-terms are constants. In a previous work \cite{BorjiKopper2}, we proved the renormalizability of the semi-infinite model. We considered BPHZ renormalization conditions and found that the semi-infinite model is renormalized by adding five position dependent counter-terms to the bare interaction given by 
\begin{multline}
    L^{\Lambda_0,\Lambda_0}(\phi)=\frac{\lambda}{4!}\int_{\mathbb{R}^+}dz~\int_{\mathbb{R}^3}d^3x~\phi^4(z,x)
+\frac{1}{2}\int_{\mathbb{R}^+}dz~\int_{\mathbb{R}^3}d^3x~ \biggl(a^{\Lambda_0}(z)\phi^2(z,x)
-b^{\Lambda_0}(z)\phi(z,x)\Delta_x \phi(z,x)\\
-d^{\Lambda_0}(z)\phi(z,x)\partial_z^2\phi(z,x) 
+s^{\Lambda_0}(z)\phi(z,x)(\partial_z\phi)(z,x)
+\frac{2}{4!}c^{\Lambda_0}(z)\phi^4(z,x)\biggr)
\ .\nonumber
\end{multline}
In the present work, we prove that there exists a particular choice of renormalization conditions for which the counter-terms appearing in the bare interaction are position independent. From a theoretical physics point of view, we aim to establish renormalizability in the restrictive sense that all counter-terms are of the same form as the interactions included in the original Hamiltonian.\\
\indent Our strategy is based on the following ideas: all possible propagators in the mixed position momentum space can be decomposed into a sum of two terms. The first term is the propagator of the translationally invariant theory $C_b$ and the second one is the part that breaks translation invariance $C_s$. Inserting this decomposition in the Feynman graph expansion, we obtain graphs involving exclusively $C_b$ (i.e. bulk graphs), and others involving $C_{s}$ or $C_{s}$ and $C_b$ (i.e. surface graphs). The bulk graphs are given by Feynman integrals which are identical to those of the corresponding translationally invariant theory up to the restriction $z\geq 0$ on $z$-integrations. This implies that these graphs can be renormalized using the same counter-terms as for $\phi^4_4$ in $\mathbb{R}^4$ with an interaction supported on the half-space. The remaining surface graphs of our semi-infinite model which do also involve $C_{s}$ can be renormalized by adding position independent surface counter-terms. This means that the semi-infinite correlation "functions" can be decomposed into a bulk part, plus a remainder which we call the "surface part". One of the important results of this paper is that the surface part admits a power counting which is
dimensionally better by one scaling dimension as compared to the bulk counterpart. This modified scaling dimension appears in
Theorem \ref{theoremReno}.
\\\indent A  motivation to our approach is related to the critical behaviour of the semi-infinite scalar field model which was studied extensively in \cite{DiehlD,DiehlD1,DiehlD2}. Methods based on the renormalization group proved that for the semi-infinite model, the behaviour on the surface differs considerably from the bulk, in the sense that the critical exponents of this model can not be fully expressed in terms of the bulk critical exponents. For Robin and Neumann boundary conditions, the critical exponents are given in terms of the bulk critical exponents and two independent surface critical exponents. In the case of Dirichlet boundary conditions, there is only one additional surface critical exponent in addition to the usual bulk critical exponents. This implies that a given bulk universality class splits into different surface universality classes  which results in a rich bulk-surface phase diagram. Each of the possible boundary conditions can be associated to a phase of the semi-infinite Ising model.
The theory of bulk critical phenomena suggests that the number of independent critical exponents
should follow directly from the number of independent renormalization functions (i.e. $Z$ functions). This correspondence between the critical exponents and counter-terms suggests that some of the counter-terms are the same as those which renormalize the translationally invariant theory, while the remaining ones are new counter-terms which result from the presence of the boundary and can be associated to the independent surface critical exponents.\\
\indent In this work, this correspondence is made explicit but only partially and this for the following two reasons: the first reason concerns the surface critical exponent of the ordinary transition (i.e. Dirichlet b.c.). For Robin and Neumann boundary conditions, we establish that two surface counter-terms are needed to make the semi-infinite model finite which correspond to the two surface critical exponents. For Dirichlet boundary conditions, the theory of critical phenomena implies that a single surface exponent which follows from the anomalous dimension of the derivative $\partial_n\phi(x,0)$, characterizes the ordinary transition.
Following \cite{DiehlD1} the counter-term corresponding to the Dirichlet surface critical exponent can be retrieved by renormalizing the theory with the insertion of the operator $\partial_n\phi(x,0)$. In this work, we do not consider this insertion and we focus our study on the renormalization of the (non-inserted) connected amputated Schwinger (CAS) distributions of the semi-infinite model. We find that no surface counter-term is needed to renormalize the semi-infinite model with Dirichlet boundary conditions. This follows from the fact that the renormalized non-amputated connected Schwinger Dirichlet $n$-point functions with a point on the surface vanish. The second point which is missing regarding the correspondence mentioned above, concerns the bulk critical exponents. The latter are associated to the counter-terms which renormalize the translationally invariant scalar field theory. However, in the context of this work, the bulk counter-terms are defined as those needed to renormalize the $\phi^4_4$ theory in $\mathbb{R}^4$ with an interaction restricted to the half-space. In a future publication, we hope to show that this theory is renormalized by adding the translationally invariant $\phi^4_4$ theory counter-terms and two surface counter-terms which are independent of the half-space boundary conditions. This will make the correspondence fully explicit. \\
\indent Our technique of proof is based on constructing a solution to the flow equations of the semi-infinite model derived in \cite{BorjiKopper2} such that the bare interaction has the following form (in the case of Robin boundary conditions)
\begin{multline}\label{surface}
    L^{\Lambda_0,\Lambda_0}_{R}(\phi)=\int_{\mathbb{R}^+}dz\int_{\mathbb{R}^3} d^3x~\left(\frac{\lambda}{4!}~\phi^4(z,x)
+\frac{1}{2} a^{\Lambda_0}_B(z)\phi^2(z,x)-\frac{1}{2} s^{\Lambda_0}_B(z)\phi(z,x)\partial_z\phi(z,x)\right.\\\left.
-\frac{1}{2}b^{\Lambda_0}_B(z)\phi(z,x)\Delta_x \phi(z,x)
-\frac{1}{2}d^{\Lambda_0}_B(z)\phi(z,x)\partial_z^2 \phi(z,x)
+\frac{1}{4!}c^{\Lambda_0}_B(z)\phi^4(z,x)
\right)\\+\int_{\mathbb{R}^3} d^3x \left(\frac{1}{2}s^{\Lambda_0}_R+c~e^{\Lambda_0}_R\right)\phi^2(0,x),
\end{multline}
where $c$ denotes the Robin parameter associated to Robin boundary conditions.\\
\indent The exposition is organized as follows. In Section 2, we review the basic setting and recall some basic properties of the Robin, Neumann, Dirichlet and bulk heat kernels. Section 3 is devoted to define the scalar field theory in $\mathbb{R}^4$ with a quartic self-interaction restricted to the half-space $\mathbb{R}^+\times\mathbb{R}^3$. The CAS of this theory obey the standard flow equations of the $\phi^4_4$-theory in $\mathbb{R}^4$, with the exception that the $z$ , $z'$ integrations appearing on the RHS of the flow equations are restricted to $\mathbb{R}^+$ instead of the full space $\mathbb{R}$. The bare interaction corresponding to this theory reads
\begin{multline*}
     L^{\Lambda_0,\Lambda_0}_B(\phi)=\int_{\mathbb{R}^+}dz\int_{\mathbb{R}^3} d^3x~\left(\frac{\lambda}{4!}~\phi^4(z,x)
+\frac{1}{2} a^{\Lambda_0}_B(z)\phi^2(z,x)-\frac{1}{2} s^{\Lambda_0}_B(z)\phi(z,x)\partial_z\phi(z,x)\right.\\\left.
-\frac{1}{2}b^{\Lambda_0}_B(z)\phi(z,x)\Delta_x \phi(z,x)
-\frac{1}{2}d^{\Lambda_0}_B(z)\phi(z,x)\partial_z^2 \phi(z,x)
+\frac{1}{4!}c^{\Lambda_0}_B(z)\phi^4(z,x)
\right),
\end{multline*}
where $a^{\Lambda_0}(z)$, $s^{\Lambda_0}(z)$, $b^{\Lambda_0}(z)$, $d^{\Lambda_0}(z)$ and $c^{\Lambda_0}(z)$ are the bulk counter-terms which can depend (smoothly) on $z$ since the interaction breaks translation invariance. In Section 4, we introduce the trees and forest structures together with their corresponding weight factors which we need later in stating Theorem \ref{theoremReno}.
In Section 5, we construct the surface correlation distributions $\mathcal{S}_{l,n;\star}^{\Lambda,\Lambda_0}$ associated to the boundary condition $\star$. Section 6 is the central part of this paper. We present Theorem \ref{theoremReno} which contains the power counting for the connected amputated Schwinger distributions (CAS) $\mathcal{S}_{l,n;\star}^{\Lambda,\Lambda_0}$ as well as their boundedness w.r.t. to $\Lambda_0$. Then, Proposition \ref{Prop44} proves that the Dirichlet surface correlation distributions can be viewed as the limit of Robin surface correlation distributions when the Robin parameter $c$ is taken to infinity. Theorem \ref{theoremReno} together with Proposition \ref{Prop1} imply Corollary \ref{Cor1} which states that the Dirichlet surface correlation distributions when folded with Dirichlet heat kernels are irrelevant. In Section 7, we explain how the minimal form (\ref{surface}) of the bare interaction is deduced from Theorem \ref{theoremReno}. First order calculations in perturbation theory \cite{Albu1,Albu2} suggest that the amputated theory is renormalized differently w.r.t. the non-amputated one in the sense that the tadpole needs more counter-terms depending on whether one of its external points is on the surface or not. We explain this in more detail to all orders of perturbation theory in Section 8.  In the Appendices, we collected technical lemmas which we use in the proof of Theorem \ref{theoremReno}. 

\section{The heat kernels and the propagators}
\subsection{Some notations and the heat kernels}
In the sequel, we will be using the following notations
\begin{align*}
    ~\int_z:=\int_0^{\infty}dz,~~\vec{p}_n&:=\left(p_1,\cdots,p_n\right),~~(\vec{z}_n,\vec{p}_n):=\left((z_1,p_1),\cdots,(z_n,p_n)\right),\\~~z_{i,j}&=\left(z_i,\cdots,z_j\right),
    ~~\vec{p}_{i,j}=\left(p_i,\cdots,p_j\right),~~~1\leq i\leq j\leq n,\\
    \|\vec{p}_n\|&:=\max_{1\leq i\leq n}|p_i|.
\end{align*}
We will also use the mixed position-momentum space representation which consists in taking the partial Fourier transform w.r.t. the variable $x\in\mathbb{R}^3$. We recall that in this representation, the Dirichlet, Neumann and Robin propagators simply read 
\begin{equation}
C_D(p;z,z')=\frac{1}{2\sqrt{p^2+m^2}}\left[e^{-\sqrt{p^2+m^2}\left|z-z'\right|}-e^{-\sqrt{p^2+m^2}\left|z+z'\right|}\right]
\ ,
\end{equation}
\begin{equation}
C_N(p;z,z')=\frac{1}{2\sqrt{p^2+m^2}}\left[e^{-\sqrt{p^2+m^2}\left|z-z'\right|}+e^{-\sqrt{p^2+m^2}\left|z+z'\right|}\right]\ ,
\end{equation}
\begin{equation}\label{Robpro}
C_R(p;z,z')=\frac{1}{2\sqrt{p^2+m^2}}\left[e^{-\sqrt{p^2+m^2}\left|z-z'\right|}+\frac{\sqrt{p^2+m^2}-c}{\sqrt{p^2+m^2}+c}
\ e^{-\sqrt{p^2+m^2}\left|z+z'\right|}\right],~~~c>0~.
\end{equation}
Note that the Dirichlet boundary condition corresponds to $c\rightarrow +\infty$ and the Neumann boundary condition to $c=0$.
For $\star \in\left\{D,R,N\right\}$, the propagator $C_{\star}$ can also be written as
\begin{equation}
    C_{\star}\left(p;z,z'\right)=\int_0^{\infty}d\lambda~e^{-\lambda(p^2+m^2)}~p_{\star}\left(\lambda;z,z'\right),
\end{equation}
where the Robin, Neumann and Dirichlet heat kernels read 
\begin{align}p_D\left(\frac{1}{\Lambda^2};z,z'\right)&=p_B\left(\frac{1}{\Lambda^2};z,z'\right)-p_B\left(\frac{1}{\Lambda^2};z,-z'\right)\label{pD},\\
p_N\left(\frac{1}{\Lambda^2};z,z'\right)&=p_B\left(\frac{1}{\Lambda^2};z,z'\right)+p_B\left(\frac{1}{\Lambda^2};z,-z'\right),\label{pN}\\
    p_R\left(\frac{1}{\Lambda^2};z,z'\right)&=p_B\left(\frac{1}{\Lambda^2};z,z'\right)+p_B\left(\frac{1}{\Lambda^2};z,-z'\right)-2\int_0^{\infty}dw~e^{-w}p_B\left(\frac{1}{\Lambda^2};z,-\frac{w}{c}-z'\right),\label{pR}
\end{align}
and the bulk heat kernel $p_B$ is given by
 \begin{equation}
     p_B\left(\tau;z,z'\right)=\frac{1}{\sqrt{2\pi\tau}}e^{-\frac{(z-z')^2}{2\tau}},~~~~~\tau>0~.
 \end{equation}
 It verifies the following basic properties:
\begin{itemize}
\item (The bulk semi-group property) For $z_1$ and $z_2$ in $\mathbb{R}$ 
    \begin{equation}\label{rr+}
        \int_{\mathbb{R}}du~ p_B(\tau_1;z_1,u)
~p_B(\tau_2;u,z_2)=p_B(\tau_1+\tau_2;z_1,z_2)~.
    \end{equation}
    \item (The $\star$ semi-group property) For $z_1$ and $z_2$ in $\mathbb{R}^+$ and $\star\in\left\{D,N,R\right\}$, we have
    \begin{equation}\label{10'}
        \int_{\mathbb{R}^+}du~ p_{\star}(\tau_1;z_1,u)
~p_{\star}(\tau_2;u,z_2)=p_{\star}(\tau_1+\tau_2;z_1,z_2)~.
    \end{equation}
\item (Completeness) For $z_1$ in $\mathbb{R}$, we have 
    \begin{equation}\label{rr*}
        \int_{\mathbb{R}}du~ p_B(\tau_1;z_1,u)=1~.
    \end{equation}
     \item For $z_1$ and $z_2$ in $\mathbb{R}^+$, we have 
    \begin{equation}\label{rr++}
        \int_{\mathbb{R}}du~ p_B(\tau_1;z_1,u)
~p_B(\tau_2;u,z_2)\leq 2 \int_{\mathbb{R}^+}du~ p_B(\tau_1;z_1,u)
~p_B(\tau_2;u,z_2)~.
    \end{equation}
\item For $\delta\geq 0$, $\tau_{\delta}=(1+\delta)\tau$ and $z_1,z_2\in \mathbb{R}^+$, we have 
\begin{eqnarray}\label{pbdelta}
     p_B\left(\tau;z_1,z_2\right)\leq \sqrt{1+\delta}~p_B\left(\tau_{\delta};z_1,z_2\right)
\end{eqnarray}
and for $\delta'>\delta$
\begin{equation}\label{in1}
     \indent |z_1-z_2|^r\ 
p_B\left(\tau_{\delta};z_1,z_2\right)
\leq C_{\delta,\delta'}\  {\tau^{\frac{r}{2}}}p_B\left(\tau_{\delta'};z_1,z_2
\right),
\end{equation}
where 
\begin{equation}\label{Cdelt}
C_{\delta,\delta'}=\sqrt{\frac{1+\delta'}{1+\delta}} \, \| \,x^r e^{-\frac{x^2}{2}\frac{\delta'-\delta}{(1+\delta)(1+\delta')}}\|_{\infty}\leq O(1)~|\delta-\delta'|^{-\frac{r}{2}}~.
\end{equation}
\item For $z,~z'\in\mathbb{R}^+$, $\tau>0$ and $c\geq 0$, we have 
\begin{equation}\label{bulkBou}
p_B\left(\tau;z,-z'\right)\leq p_B\left(\tau;z,z'\right),~~~~ \int_{\mathbb{R}^+}dw~e^{-w}p_B\left(\tau;z,-z'-\frac{w}{c}\right)\leq p_B\left(\tau;z,z'\right). 
\end{equation}
\subsection{The regularized propagators}
We denote by $\star$ the type of boundary conditions considered. For $0\leq \Lambda\leq \Lambda_0$, we define the regularized flowing propagator associated to the boundary condition $\star\in\left\{D,N,R\right\}$ as follows:
\begin{equation}
{C}_{\star}^{\Lambda,\Lambda_0}(p;z,z'):=\int_{\frac{1}{\Lambda_0^2}}^{\frac{1}{\Lambda^2}}d\lambda~p_{\star}\left(\lambda;z,z'\right)e^{-\lambda\left(p^2+m^2\right)}.
\end{equation}
This can also be rewritten as
\begin{equation}\label{decom}
    C_{\star}^{\Lambda,\Lambda_0}\left(p;z,z'\right)=C_{B}^{\Lambda,\Lambda_0}\left(p;z,z'\right)+C_{S,\star}^{\Lambda,\Lambda_0}\left(p;z,z'\right),
\end{equation}
where 
\begin{equation}\label{BulkPro}
{C}_{B}^{\Lambda,\Lambda_0}(p;z,z'):=\int_{\frac{1}{\Lambda_0^2}}^{\frac{1}{\Lambda^2}}d\lambda~p_{B}\left(\lambda;z,z'\right)e^{-\lambda\left(p^2+m^2\right)}
\end{equation}
and 
\begin{equation}
{C}_{S,\star}^{\Lambda,\Lambda_0}(p;z,z'):=\int_{\frac{1}{\Lambda_0^2}}^{\frac{1}{\Lambda^2}}d\lambda~p_{S,\star}\left(\lambda;z,z'\right)e^{-\lambda\left(p^2+m^2\right)},
\end{equation}
with the surface heat kernel $p_{S,\star}$ defined as $p_{\star}-p_B$. In the case of Robin boundary conditions, the surface Robin heat kernel is given by
\begin{equation}\label{43}
p_{S,R}\left(\frac{1}{\Lambda^2};z,z'\right):=p_B\left(\frac{1}{\Lambda^2};z,-z'\right)
-2\int_0^{\infty}dw~ e^{-w}p_B\left(\frac{1}{\Lambda^2};z,-\frac{w}{c}-z'\right).
\end{equation}
Note that the Robin heat kernel and $p_{S,R}$ are uniformly bounded w.r.t. the Robin parameter $c$. Namely, we have using (\ref{bulkBou})
\begin{equation}\label{SurfHeatB}
    p_R\left(\tau;z,z'\right)\leq 4~p_B\left(\tau;z,z'\right),~~~~ p_{S,R}\left(\tau;z,z'\right)\leq 3~p_B\left(\tau;z,-z'\right),
\end{equation}
for all $z,z'\geq 0$, $\tau>0$ and $c\geq 0$. Similarly, we also have
\begin{equation}
    p_D\left(\tau;z,z'\right)\leq p_N\left(\tau;z,z'\right)\leq 2~p_B\left(\tau;z,z'\right) . 
\end{equation} 
In the sequel, we denote the derivative of the flowing propagators w.r.t. $\Lambda$ by
    \begin{equation}\label{p3}
     \dot{C}_{\bullet}^{\Lambda}(p;z,z')=\frac{\partial}{\partial\Lambda}{C}_{\bullet}^{\Lambda,\Lambda_0}(p;z,z')=\dot{C}^{\Lambda}(p)\ p_{\bullet}\left(\frac{1}{\Lambda^2};z,z'\right)~,
\end{equation}
where $\dot{C}^{\Lambda}(p)=-\frac{2}{\Lambda^3}e^{-\frac{p^2+m^2}{\Lambda^2}}$ and $\bullet \in\left\{\star,\left\{S,\star\right\},B\right\}$ with $\star \in\left\{D,N,R\right\}$. \\Given a polynomial $\mathcal{P}$ and $w\in\mathbb{N}^3$, we have the following estimate on the $3$-dimensional covariance 
\begin{equation}\label{cov}
    \left|\partial^w \Dot{C}^{\Lambda}\left(p\right)~\mathcal{P}\left(\frac{p}{\Lambda}\right)\right|\leq \left(\Lambda+m\right)^{-3-|w|} \tilde{\mathcal{P}}\left(\frac{|p|}{\Lambda+m}\right),
\end{equation}
where $\tilde{\mathcal{P}}$ is a polynomial with positive coefficients. We refer to (\ref{notation}) for the multi-index $w$ notation.
\end{itemize}
\section{The bulk theory on the half-space $\mathbb{R}^+\times\mathbb{R}^3$}\label{bulkspace}
\subsection{The Action and the Flow Equations}
\noindent We consider the theory of a real scalar field $\phi$ with mass $m$ on the four dimensional Euclidean space-time $\mathbb{R}^4$ within the framework of functional integration. The point of departure to define this theory is to write the associated regularized path integral which is uniquely defined by the corresponding gaussian measure. The regularized flowing propagator is given by (\ref{BulkPro}). Note that for $\Lambda\rightarrow 0$ and $\Lambda_0\rightarrow \infty$ we recover the unregularized propagator.
For finite $\Lambda_0$ and in finite volume the positivity and the regularity properties of $C^{\Lambda,\Lambda_0}_B$ permit to define the theory rigorously from the functional integral 
\begin{align}\label{fl}
 e^{-\frac{1}{\hbar}\left(L^{\Lambda,\Lambda_0}_B(\phi)+I^{\Lambda,\Lambda_0}\right)}:&=
\int d\mu^{\Lambda,\Lambda_0}_B(\Phi)\  
e^{-\frac{1}{\hbar}L^{\Lambda_0,\Lambda_0}_B(\Phi+\phi)}\ ,\\
    L^{\Lambda,\Lambda_0}_B(0)&=0\ ,\nonumber
\end{align}
where the factors of $\hbar$ have been introduced with regard to a systematic loop expansion considered
later. Here, the Gaussian measure $d\mu_B^{\Lambda,\Lambda_0}$ is of mean zero and covariance $\hbar C_B^{\Lambda,\Lambda_0}$. The test function $\phi$ is in the support of $d\mu_B^{\Lambda,\Lambda_0}$ which in particular implies that it is in $\mathcal{C}^{\infty}\left(\mathbb{R}^4\right)$. This regularity stems from the UV-regularization determined by the cutoff $\Lambda_0$ which is imperative to have a well-defined functional integral. The normalization factor $e^{-\frac{1}{\hbar}I^{\Lambda,\Lambda_0}}$ is due to vacuum contributions. It diverges in infinite volume so that we can take the infinite volume limit only when it has been eliminated \cite{Keller}. We do not make the finite volume explicit here since it plays no role in the sequel.\\
 \indent The functional $L^{\Lambda_0,\Lambda_0}_B(\phi)$ is the bare interaction of a renormalizable theory including counter-terms, viewed as a formal power series in $\hbar$. It contains the tree order interaction and the related counter-terms. The interaction is supported only on the half-space $\mathbb{R}^+\times\mathbb{R}^3$ which implies that translation invariance is broken in the $z$-direction (the semi-line). This implies that the counter-terms may be $z$-dependent. In general, the constraints on the bare action result from the symmetry properties of the theory which are imposed, on its field content and on the form of the propagator. It is therefore natural to consider the general bare interaction
\begin{eqnarray}\label{bare}
     L^{\Lambda_0,\Lambda_0}_B(\phi)=\frac{\lambda}{4!}\int_z\int_{\mathbb{R}^3}d^3x\phi^4(z,x)
+\frac{1}{2}\int_z\int_{\mathbb{R}^3}d^3x \biggl(a_B^{\Lambda_0}(z)\phi^2(z,x)
-b_B^{\Lambda_0}(z)\phi(z,x)\Delta_x \phi(z,x)\\
-d_B^{\Lambda_0}(z)\phi(z,x)\partial_z^2\phi(z,x) 
+s_B^{\Lambda_0}(z)\phi(z,x)(\partial_z\phi)(z,x)
+\frac{2}{4!}c_B^{\Lambda_0}(z)\phi^4(z,x)\biggr)
\ .\nonumber
\end{eqnarray}
Here we supposed the theory to be symmetric under $\phi\rightarrow-\phi\,$, 
and we included in (\ref{bare}) only relevant terms in the sense of the renormalization group. The functions $a^{\Lambda_0}_B(z)$, $b^{\Lambda_0}_B(z)$, $c^{\Lambda_0}_B(z)$, $d^{\Lambda_0}_B(z)$ and $s^{\Lambda_0}_B(z)$ are supposed to be smooth. \\
The flow equation (FE) is obtained from (\ref{fl}) on differentiating w.r.t. $\Lambda$. For the steps of the computation, we refer the reader to \cite{Keller,Muller,BorjiKopper}. It is a differential equation for the functional $L^{\Lambda,\Lambda_0}_B$:\\
\begin{equation}\label{floEq}
  \partial_{\Lambda}(L^{\Lambda,\Lambda_0}_B+I^{\Lambda,\Lambda_0})=\frac{\hbar}{2}\langle \frac{\delta}{\delta \phi},\dot{C}_B^{\Lambda}\,\frac{\delta}{\delta \phi}\rangle L^{\Lambda,\Lambda_0}_B-\frac{1}{2}\langle\frac{\delta}{\delta \phi}L^{\Lambda,\Lambda_0}_B,\dot{C}^{\Lambda}_B\,\frac{\delta}{\delta \phi}L^{\Lambda,\Lambda_0}_B\rangle\ .
\end{equation}
By $\langle,\rangle$ we denote the standard inner product in 
$L^2(\mathbb{R}^+\times\mathbb{R}^3)$.\\
We expand the functional 
$L^{\Lambda,\Lambda_0}_B(\phi)$ in a formal power series w.r.t. $\hbar$,
\begin{equation*}
    L^{\Lambda,\Lambda_0}_B(\phi)=\sum_{l=0}^{\infty}\hbar^l L^{\Lambda,\Lambda_0}_{l,B}(\phi)\ .
\end{equation*}
Corresponding expansions for $a^{\Lambda_0}_B(z),~b^{\Lambda_0}_B(z)$..., are $a^{\Lambda_0}_B(z)=\sum_{l=1}^{\infty}\hbar^l a^{\Lambda_0}_{l,B}(z) $, etc. From $L^{\Lambda,\Lambda_0}_{l,B}(\phi)$ we obtain the CAS distributions of loop order $l$ as 
\begin{equation*}
    \mathcal{D}^{\Lambda,\Lambda_0}_{l,n}\left((z_1,x_1),\cdots,(z_n,x_n)\right)
:=\delta_{\phi(z_1,x_1)}\cdots \delta_{\phi(z_n,x_n)}L^{\Lambda,\Lambda_0}_{l,B}|_{\phi=0}~,
\end{equation*}
where we used the notation $\delta_{\phi(z,x)}=\delta/\delta \phi(z,x)\,$. \\
In the $pz$-representation, we set for $r$, $r_1$ and $r_2\in \mathbb{N}^*$
\begin{equation}\label{13*}
\mathcal{D}_{l,n;r}^{\Lambda,\Lambda_0;(i)}\left(z_1;\vec{p}_n;\Phi_{n}\right):=\int_0^{\infty}dz_2\cdots dz_n~(z_1-z_i)^r\mathcal{D}_{l,n}^{\Lambda,\Lambda_0}
\left((z_1,p_1),\cdots,(z_n,p_n)\right)\phi_2(z_2)\cdots\phi_n(z_n),\ 
\end{equation}
\begin{multline}\label{13**}
\mathcal{D}_{l,n;r_1,r_2}^{\Lambda,\Lambda_0;(i,j)}\left(z_1;\vec{p}_n;\Phi_{n}\right):=\int_0^{\infty}dz_2\cdots dz_n~(z_1-z_i)^{r_1}(z_1-z_j)^{r_2}\mathcal{D}_{l,n}^{\Lambda,\Lambda_0}
\left((z_1,p_1),\cdots,(z_n,p_n)\right)\\\times\phi_2(z_2)\cdots\phi_n(z_n),~~~~r_1+r_2=r, 
\end{multline}
and for $r=0$
\begin{eqnarray}\label{13***}
\mathcal{D}_{l,n}^{\Lambda,\Lambda_0}\left(z_1;\vec{p}_n;\Phi_{n}\right):=\int_0^{\infty}dz_2\cdots dz_n~\mathcal{D}_{l,n}^{\Lambda,\Lambda_0}
\left((z_1,p_1),\cdots,(z_n,p_n)\right)\phi_2(z_2)\cdots\phi_n(z_n)\ .
\end{eqnarray}
Here we denote 
 $$\Phi_{n}(z_2,\cdots,z_n):=\prod_{i=2}^n\phi_i(z_i)$$
and 
$$\delta^{(3)}(p_1+\cdots+p_n)
\mathcal{D}_{l,n}^{\Lambda,\Lambda_0}\left((z_1,p_1),\cdots,(z_n,p_n)\right)
=(2\pi)^{3(n-1)}\frac{\delta^n}{\delta \phi(z_1,p_1)
\cdots\delta \phi(z_n,p_n)}L^{\Lambda,\Lambda_0}_{l,B}(\phi)|_{\phi\equiv0}~.$$
The distribution $\delta^{(3)}(p_1+\cdots+p_n)$ appears because of translation
invariance in the ${x}$ directions.
The FE for the CAS distributions derived from (\ref{floEq}) are
\begin{multline}\label{FED}
    \partial_{\Lambda}\partial^w\mathcal{D}_{l,n}^{\Lambda,\Lambda_0}\left((z_1,p_1),\cdots,(z_n,p_n)\right)\\
=\frac{1}{2}\int_{\mathbb{R}^+}dz~\int_{\mathbb{R}^+}dz'~\int_k 
\partial^w\mathcal{D}_{l-1,n+2}^{\Lambda,\Lambda_0}
\left((z_1,p_1),\cdots,(z_n,p_n),(z,k),(z',-k)\right)
\dot{C}^{\Lambda}_B(k;z,z')\\
 \indent -\frac{1}{2}\int_{\mathbb{R}^+}dz~\int_{\mathbb{R}^+}dz'~ \sum_{l_1,l_2}'\sum_{n_1,n_2}'\sum_{w_i}c_{w_i} \left[\partial^{w_1}\mathcal{D}_{l_1,n_1+1}^{\Lambda,\Lambda_0}((z_1,p_1),\cdots,(z_{n_1}p_{n_1}),(z,p))\partial^{w_3}\dot{C}^{\Lambda}_B(p;z,z')\right.\\\left.\times\ 
\partial^{w_2}\mathcal{D}_{l_2,n_2+1}^{\Lambda,\Lambda_0}((z',-p),\cdots,(z_{n},p_{n}))\right]_{rsym},\\
    p=-p_1-\cdots-p_{n_1}=p_{n_1+1}+\cdots+p_n\ .
\end{multline}
The number $|w|$ of momentum derivatives, is characterized by a multi-index. We use the shorthand notation
\begin{equation}\label{notation}
    \partial^{w}:=\prod^{n}_{i=1}\prod_{\mu=0}^3\left(\frac{\partial}{\partial p_{i,\mu}}\right)^{w_{i,\mu}}~ \mathrm{with}~ w=(w_{1,0},\cdots,w_{n,3}),~
|w|=\sum_{i,\mu}{w_{i,\mu}},~w_{i,\mu}\in \mathbb{N}^*\ .
\end{equation}{}
The symbol "rsym" means summation over those permutations of the momenta $(z_1,p_1)$, $\cdots$ ,$(z_n,p_n)$, which do not leave invariant the (unordered) subsets $\left((z_1,p_1),\cdots,(z_{n_1},p_{n_1})\right)$ and \newline $\left((z_{n_1+1},p_{n_1+1})\right.$,$\left.\cdots,(z_n,p_n)\right)$, and therefore, produce mutually different pairs of (unordered) image subsets, and the primes restrict the summations to $n_1+n_2=n$, $l_1+l_2=l$, $w_1+w_2+w_3=w$, respectively. The combinatorial factor $c_{\left \{w_i\right \}}=w!(w_1!w_2!w_3!)^{-1}$ stems from Leibniz's rule. 
\subsection{Test functions and boundary conditions}
The $n$-point correlation "functions" $\mathcal{D}_{l,n}^{\Lambda,\Lambda_0}$ when considered in the $pz$-representation are tempered distributions which belong for fixed $\Vec{p}_n$ to the space $\mathcal{S}'\left(\mathbb{R}^{+n}\right)$ w.r.t. the semi-norms $$\prod_{i=1}^n \mathcal{N}_{2}\left(\phi_i\right)\ ,$$ 
where 
$\mathcal{N}_{2}(\phi):=\sup_{0\leq \alpha,\beta\leq2}\left\|(1+z^{\beta})
\partial_z^{\alpha}\phi(z)\right\|_{\infty}$ 
and $\partial_z\phi|_{z=0}=\lim_{z\rightarrow 0^+}\partial_z\phi\,$. 
For additional informations on the topological construction of $\mathcal{S}'\left(\mathbb{R}^+\right)$, we refer the reader to \cite{Pott}. Our method of proof relies on inductive bounds deduced from the flow equations (\ref{FED}). The induction restricts our choice of the test functions. To proceed inductively we cannot admit any arbitrary test function in $\mathcal{S}\left(\mathbb{R}^{+n}\right)$. Let us give the set of test functions we will be using in the sequel. \\For $2\leq s \leq n$, we define 
 \begin{center}
 $\tau:=\inf \tau_{2,s}$,~~ where $\tau_{2,s}=\left(\tau_2,\cdots,\tau_s\right)$ with $\tau_i>0$~,
 \end{center}
 and similarly $z_{2,s}=\left(z_2,\cdots,z_s\right)$. Given $\left(y_2,\cdots,y_s\right)\in \mathbb{R}^{s-1}$, we define
\begin{equation}
    \phi_{\tau_{2,s},y_{2,s}}(z_{2,s}):=\prod_{i=2}^sp_B(\tau_i;z_i,y_i)
\prod_{i=s+1}^n \chi^{+}(z_i)\ ,
\end{equation}
where $\chi^{+}$ is the characteristic function of the semi-line $\mathbb{R}^+$. This definition can be generalized by choosing any other subset of $s-1$ coordinates among $z_2,\cdots,z_n\, $. The characteristic functions $\chi^+$ are introduced in order to be able to extract the relevant terms in the sense of the renormalization group from the full $n$-point distributions and to get inductive control of the local counter terms. To go further one could either prove (in a more functional analysis type of approach) that our test functions are dense  in the set of
smooth rapidly decaying functions on $\mathbb{R}^+$ w.r.t. a suitable norm and that $\mathcal{D}_{l,n}^{\Lambda,\Lambda_0}\left(z_1;\vec{p}_n;\Phi_{n}\right)$ are continuous w.r.t. this semi-norm. To extract the relevant terms contained in 
\begin{equation}
\mathcal{D}_{l,2}^{\Lambda,\Lambda_0}\left(z_1;\vec{0};\phi_2\right):=\int_{z_2}\mathcal{D}_{l,2}^{\Lambda,\Lambda_0}\left((z_1,0),(z_2,0)\right)\phi_2(z_2)
\end{equation}
and 
\begin{equation}
\mathcal{D}_{l,4}^{\Lambda,\Lambda_0}\left(z_1;\vec{0};\Phi_4\right):=\int_{z_{2,4}}\mathcal{D}_{l,4}^{\Lambda,\Lambda_0}\left((z_1,0),\cdots,(z_4,0)\right)\prod_{i=2}^4 \phi_i(z_i)\ ,
 \end{equation}
we use a Taylor expansion of the test functions $\phi_2$ and $\Phi_4$, which gives
\begin{align}
     \mathcal{D}_{l,2}^{\Lambda,\Lambda_0}\left(z_1;0,0;\phi_2\right)
    &=a_{l,B}^{\Lambda,\Lambda_0}(z_1)\phi_2(z_1)-s_{l,B}^{\Lambda,\Lambda_0}(z_1)(\partial_{z_1}\phi_2)(z_1)\nonumber\\&~~~~~~-d_{l,B}^{\Lambda,\Lambda_0}(z_1) (\partial_{z_1}^2\phi_2)(z_1)+l_{l,2;B}^{\Lambda,\Lambda_0}(z_1; \phi_2)\ ,\label{f1}\\
   \left(\partial_{p^2}\mathcal{D}_{l,2}^{\Lambda,\Lambda_0}\right)\left(z_1;0,0;\phi_2\right)&=b_{l,B}^{\Lambda,\Lambda_0}(z_1)\phi_2(z_1)+\left(\partial_{p^2}l_{l,2}^{\Lambda,\Lambda_0}\right)(z_1; \phi_2)\ ,\label{f2}\\
    \mathcal{D}_{l,4}^{\Lambda,\Lambda_0}\left(z_1;0,\cdots,0;\Phi_4\right)&=c_{l,B}^{\Lambda,\Lambda_0}(z_1)\phi_2(z_1)\phi_3(z_1)\phi_4(z_1)+l_{l,4,B}^{\Lambda,\Lambda_0}\left(z_1;\Phi_4\right)\ ,\label{f3}
\end{align}
where $\Phi_4(z_2,z_3,z_4)=\prod_{i=2}^4\phi_i(z_i)\,$.\\
Then the relevant terms appear as 
\begin{align}\label{rel}
     a_{l,B}^{\Lambda,\Lambda_0}(z_1)&=\int_0^{\infty} dz_2 \ \mathcal{D}_{l,2}^{\Lambda,\Lambda_0}\left((z_1,0),(z_2,0)\right),\\
    s_{l,B}^{\Lambda,\Lambda_0}(z_1)&=\int_0^{\infty} dz_2 \, (z_1-z_2)\mathcal{D}_{l,2}^{\Lambda,\Lambda_0}\left((z_1,0),(z_2,0)\right),\\
    d_{l,B}^{\Lambda,\Lambda_0}(z_1)&=-\frac{1}{2}\int_0^{\infty} dz_2 \, (z_1-z_2)^2\mathcal{D}_{l,2}^{\Lambda,\Lambda_0}\left((z_1,0),(z_2,0)\right),\\
    b_{l,B}^{\Lambda,\Lambda_0}(z_1)&=\int_0^{\infty} dz_2\ \partial_{p^2}\left(\mathcal{D}_{l,2}^{\Lambda,\Lambda_0}\left((z_1,p),(z_2,-p)\right)\right)_{|_{p=0}},\\
     c_{l,B}^{\Lambda,\Lambda_0}(z_1)&=\int_0^{\infty} dz_2dz_3dz_4 \ \mathcal{D}_{l,4}^{\Lambda,\Lambda_0}\left((z_1,0),\cdots,(z_4,0)\right),\label{lolo}
\end{align}
and the remainders $l_{l,2,B}^{\Lambda,\Lambda_0}\left(z_1; \phi_2\right)$, $\left(\partial_{p^2}l_{l,2,B}^{\Lambda,\Lambda_0}\right)\left(z_1;\phi_2\right)$ and $l_{l,4,B}^{\Lambda,\Lambda_0}\left(z_1;\Phi_4\right)$ can be written as
\begin{equation}
 l_{l,2,B}^{\Lambda,\Lambda_0}(z_1;\phi_2)=\int_{0}^{\infty}dz_2\int_0^1 dt \frac{(1-t)^2}{2!}\partial_t^3{\phi_2}\left(t z_2+(1-t)z_1\right)\mathcal{D}_{l,2}^{\Lambda,\Lambda_0}((z_1;0),(z_2;0))\ ,
\end{equation}
\begin{equation*}
 \left(\partial_{p^2}l_{l,2,B}^{\Lambda,\Lambda_0}\right)(z_1; \phi_2)
=\int_{0}^{\infty}\!\! dz_2\int_0^1 dt \partial_t{\phi_2}\left(t z_2+(1-t)z_1\right)\partial_{p^2}\left(\mathcal{D}_{l,2}^{\Lambda,\Lambda_0}\left((z_1,p),(z_2,-p)\right)\right)_{|_{p=0}}\, 
\end{equation*}
and
\begin{multline}\label{loup}
l_{l,4,B}^{\Lambda,\Lambda_0}(z_1;\Phi_4)\\
=\int_{0}^{\infty}dz_2dz_3dz_4~\mathcal{D}_{l,4}^{\Lambda,\Lambda_0}((z_1,0),\cdots,(z_4,0))\left[\int_0^1 dt~ {\partial_t\phi_2}\left(t z_2+(1-t)z_1\right)\phi_3(z_3)\phi_4(z_4)\right.\\\left.
+\phi_2(z_1)\int_0^1 dt ~{\partial_t\phi_3}\left(tz_3+(1-t)z_1\right)\phi_4(z_4)+\phi_2(z_1)\phi_3(z_1)\int_0^1 dt ~{\partial_t\phi_4}\left(tz_4+(1-t)z_1\right)\right]~.
\end{multline}
\textbf{Boundary conditions at $\Lambda=\Lambda_0$:}\\
The bare interaction implies that at $\Lambda=\Lambda_0$
\begin{align}\label{Dl1}
     \mathcal{D}_{l,2}^{\Lambda_0,\Lambda_0}
\left((z_1,p),(z_2,-p)\right)&=\left(a^{\Lambda_0}_{l;B}(z_1)+b^{\Lambda_0}_{l;B}(z_1)p^2-s^{\Lambda_0}_{l;B}(z_1)\partial_{z_1}-d^{\Lambda_0}_{l;B}(z_1)\partial_{z_1}^2\right)\delta(z_1-z_2)~,\nonumber\\
\mathcal{D}_{l,4}^{\Lambda_0,\Lambda_0}\left((z_1,p_1),\cdots,(z_4,p_4)\right)&=\left(\lambda \delta_{l,0}+c_{l;B}^{\Lambda_0}(z_1)(1-\delta_{l,0})\right)\prod_{i=2}^4 \delta(z_1-z_i)~.\nonumber\\
\mathcal{D}_{l,n}^{\Lambda_0,\Lambda_0}
\left((z_1,p_1),\cdots,(z_n,p_n)\right)&=0~,~~~~\forall n \geq 5.
\end{align}
\textbf{Renormalization conditions at $\Lambda=0$ (BPHZ renormalization conditions):}\\
The renormalization conditions are fixed at $\Lambda=0$ by imposing for all $z_1\geq 0$
\begin{eqnarray}\label{Dl3}
     \indent a_{l,B}^{0,\Lambda_0}(z_1)\equiv0,~~s_{l,B}^{0,\Lambda_0}(z_1)\equiv0,~~
    d_{l,B}^{0,\Lambda_0}(z_1)\equiv0,~~b_{l,B}^{0,\Lambda_0}(z_1)\equiv0,~~
     c_{l,B}^{0,\Lambda_0}(z_1)\equiv0\ .
\end{eqnarray}
These will be adopted in the following.
Note that the boundary conditions are invariant under $O(3)$-symmetry.\\
We will need the following result, which we do not prove. The proof can be performed following the same steps as in the proof of Theorem 1 in \cite{BorjiKopper2}:
\begin{proposition}\label{Prop1}
For $0\leq\Lambda\leq \Lambda_0<\infty$, $1\leq s\leq n$, $2\leq i \leq n$ and $0\leq r\leq 3$, we consider test functions of the form $\phi_{\tau_{2,s},y_{2,s}}(z_{2,s})$, which are also denoted in shorthand as $\phi_{\tau_{2,s},y_{2,s}}$.\\
Adopting (\ref{Dl1})-(\ref{Dl3}) we claim
\begin{multline}\label{c1b}
    \left| \partial^w \mathcal{D}_{l,n;r}^{\Lambda,\Lambda_0;(i)}(z_1;\vec{p}_n;\phi_{\tau_{2,s},y_{2,s}})\right|\\ \indent~~~~~~\leq \left(\Lambda+m\right)^{4-n-|w|-r}\mathcal{P}_1\left(\log \frac{\Lambda+m}{m}\right)\mathcal{P}_2\left(\frac{\left\|\vec{p}_n\right\|}{\Lambda+m}\right) \mathcal{Q}_1\left(\frac{\tau^{-\frac{1}{2}}}{\Lambda+m}\right)\mathcal{F}^{\Lambda}_{s,l;\delta}(\tau_{2,s})\ ,
\end{multline}
\begin{multline}\label{c1bb}
    \left| \partial^w \mathcal{D}_{l,n;r_1,r_2}^{\Lambda,\Lambda_0;(i,j)}(z_1;\vec{p}_n;\phi_{\tau_{2,s},y_{2,s}})\right|\\ \indent~~~~~~\leq \left(\Lambda+m\right)^{4-n-|w|-r_1-r_2}\mathcal{P}'_1\left(\log \frac{\Lambda+m}{m}\right)\mathcal{P}'_2\left(\frac{\left\|\vec{p}_n\right\|}{\Lambda+m}\right) \mathcal{Q}'_1\left(\frac{\tau^{-\frac{1}{2}}}{\Lambda+m}\right)\mathcal{F}^{\Lambda}_{s,l;\delta}(\tau_{2,s})\ ,
\end{multline}
and $\partial^w \mathcal{D}_{l,n;r}^{\Lambda,\Lambda_0;(i)}(z_1;\vec{p}_n;\phi_{\tau_{2,s},y_{2,s}})\,$ is continuous w.r.t. $z_1\,$. Here, $\mathcal{P}_i$ and $\mathcal{Q}_i$ denote polynomials with non-negative coefficients which depend on $l,n,|w|,r,$ but not on $\left \{ p_i \right \}$, $\Lambda$, $\Lambda_0$ and $z_1$. The polynomials $\mathcal{Q}_i$ are reduced to a constant if $s=1$, and for $l=0$ all polynomials $\mathcal{P}_i$ and $\mathcal{Q}_i$ reduce to constants.
\end{proposition}
The weight factors $\mathcal{F}_{s,l;\delta}^{\Lambda}\left(\tau_{2,s}\right)$ are defined in Section \ref{TFW}.
\section{The Surface correlation distributions}
\subsection{The semi-infinite theory}
In this subsection, we recall the flow equations of the semi-infinite massive scalar field model presented in \cite{BorjiKopper2}:
\begin{multline}\label{FEL}
    \partial_{\Lambda}\partial^w\mathcal{L}_{l,n;\star}^{\Lambda,\Lambda_0}\left((z_1,p_1),\cdots,(z_n,p_n)\right)\\
=\frac{1}{2}\int_z \int_{z'}\int_k 
\partial^w\mathcal{L}_{l-1,n+2;\star}^{\Lambda,\Lambda_0}
\left((z_1,p_1),\cdots,(z_n,p_n),(z,k),(z',-k)\right)
\dot{C}^{\Lambda}_{\star}(k;z,z')\\
    \indent -\frac{1}{2}\int_z \int_{z'} \sum_{l_1,l_2}'\sum_{n_1,n_2}'\sum_{w_i}c_{w_i} \left[\partial^{w_1}\mathcal{L}_{l_1,n_1+1;\star}^{\Lambda,\Lambda_0}((z_1,p_1),\cdots,(z_{n_1}p_{n_1}),(z,p))\partial^{w_3}\dot{C}^{\Lambda}_{\star}(p;z,z')\right.\\\left.\times\ 
\partial^{w_2}\mathcal{L}_{l_2,n_2+1;\star}^{\Lambda,\Lambda_0}((z',-p),\cdots,(z_{n},p_{n}))\right]_{rsym},\\
    p=-p_1-\cdots-p_{n_1}=p_{n_1+1}+\cdots+p_n\ .
\end{multline}
where $C_{\star}^{\Lambda,\Lambda_0}\left(p;z,z'\right)$ is defined in (\ref{decom}) and $\mathcal{L}_{l,n;\star}^{\Lambda,\Lambda_0}\left((z_1,p_1),\cdots,(z_n,p_n)\right)$ denote the semi-infinite correlation distributions at loop order $l$ and with $n$ external points. The $\star$ index refers to the type of considered boundary conditions, namely Dirichlet, Neumann and Robin. In \cite{BorjiKopper2}, we imposed the following mixed boundary conditions:
\begin{itemize}
    \item \underline{\textbf{At $\Lambda=\Lambda_0$:}}
\begin{align}\label{sib}
     \mathcal{L}_{l,2;\star}^{\Lambda_0,\Lambda_0}
\left((z_1,p),(z_2,-p)\right)&=\left(a^{\Lambda_0}_{l;\star}(z_1)+b^{\Lambda_0}_{l;\star}(z_1)p^2-s^{\Lambda_0}_{l;\star}(z_1)\partial_{z_1}-d^{\Lambda_0}_{l;\star}(z_1)\partial_{z_1}^2\right)\delta(z_1-z_2)~,\nonumber\\
\mathcal{L}_{l,4;\star}^{\Lambda_0,\Lambda_0}\left((z_1,p_1),\cdots,(z_4,p_4)\right)&=\lambda \delta_{l,0}+c_{l;\star}^{\Lambda_0}(z_1)(1-\delta_{l,0})\prod_{i=2}^4 \delta(z_1-z_i)~.\nonumber\\
\mathcal{L}_{l,n;\star}^{\Lambda_0,\Lambda_0}
\left((z_1,p_1),\cdots,(z_n,p_n)\right)&=0~,~~~~\forall n \geq 5.
\end{align}
\item \underline{\textbf{At $\Lambda=0$:}} We impose BPHZ type renormalization conditions. Namely, for all $z_1\geq 0$ we set 
\begin{equation}\label{renoc}
    \indent a_{l;\star}^{0,\Lambda_0}(z_1)\equiv0,~~s_{l;\star}^{0,\Lambda_0}(z_1)\equiv0,~~
    d_{l;\star}^{0,\Lambda_0}(z_1)\equiv0,~~b_{l;\star}^{0,\Lambda_0}(z_1)\equiv0,~~
     c_{l;\star}^{0,\Lambda_0}(z_1)\equiv0\ ,
\end{equation}
where 
\begin{align}\label{rel}
     a_{l;\star}^{\Lambda,\Lambda_0}(z_1)=\int_0^{\infty} dz_2 \ \mathcal{L}_{l,2;\star}^{\Lambda,\Lambda_0}\left((z_1,0),(z_2,0)\right),\\
    s_{l;\star}^{\Lambda,\Lambda_0}(z_1)=\int_0^{\infty} dz_2 \, (z_1-z_2)\mathcal{L}_{l,2;\star}^{\Lambda,\Lambda_0}\left((z_1,0),(z_2,0)\right),\\
    d_{l;\star}^{\Lambda,\Lambda_0}(z_1)=-\frac{1}{2}\int_0^{\infty} dz_2 \, (z_1-z_2)^2\mathcal{L}_{l,2;\star}^{\Lambda,\Lambda_0}\left((z_1,0),(z_2,0)\right),\\
    b_{l;\star}^{\Lambda,\Lambda_0}(z_1)=\int_0^{\infty} dz_2\ \partial_{p^2}\left(\mathcal{L}_{l,2;\star}^{\Lambda,\Lambda_0}\left((z_1,p),(z_2,-p)\right)\right)_{|_{p=0}},\\
     c_{l;\star}^{\Lambda,\Lambda_0}(z_1)=\int_0^{\infty} dz_2dz_3dz_4 \ \mathcal{L}_{l,4;\star}^{\Lambda,\Lambda_0}\left((z_1,0),\cdots,(z_4,0)\right)\label{lolo}.
\end{align}
\end{itemize}
This yielded five position dependent counter-terms which appear in the bare interaction of our semi-infinite model
\begin{multline}\label{bare1}
    L^{\Lambda_0,\Lambda_0}_{\star}(\phi)=\frac{\lambda}{4!}\int_{\mathbb{R}^+}dz\int_{\mathbb{R}^3}d^3x~\phi^4(z,x)
+\frac{1}{2}\int_{\mathbb{R}^+}dz\int_{\mathbb{R}^3}d^3x~ \biggl(a^{\Lambda_0}_{\star}(z)\phi^2(z,x)
-b^{\Lambda_0}_{\star}(z)\phi(z,x)\Delta_x \phi(z,x)\\
-d^{\Lambda_0}_{\star}(z)\phi(z,x)\partial_z^2\phi(z,x) 
+s^{\Lambda_0}_{\star}(z)\phi(z,x)(\partial_z\phi)(z,x)
+\frac{2}{4!}c^{\Lambda_0}_{\star}(z)\phi^4(z,x)\biggr).
\end{multline}
In our previous work \cite{BorjiKopper2}, we saw that imposing constant renormalization conditions w.r.t. the position $z$ at the scale $\Lambda=0$ is at the expense of obtaining position dependent counter-terms. In this work, we would like to reverse the process in the sense that we aim to obtain position independent counter-terms by transferring the position dependent part to the renormalization conditions. Our strategy is based on extracting the surface counter-terms from the semi-infinite counter-terms by separating the bulk and the surface effects. Concretely, we proceed by subtracting the bulk correlation distributions defined in Section \ref{bulkspace} from the semi-infinite correlation distributions and study the behaviour of the difference to which we refer as the surface correlation distributions. Namely, we write
\begin{equation}\label{S1ln}
    \mathcal{S}_{l,n;\star}^{\Lambda,\Lambda_0}\left((\vec{z}_n,\vec{p}_n)\right):= \mathcal{L}_{l,n;\star}^{\Lambda,\Lambda_0}\left((\vec{z}_n,\vec{p}_n)\right)-\mathcal{D}_{l,n}^{\Lambda,\Lambda_0}\left((\vec{z}_n,\vec{p}_n)\right).
\end{equation}
The definition (\ref{S1ln}) allows to write the FEs verified by $\mathcal{S}_{l,n;\star}^{\Lambda,\Lambda_0}$, which we give explicitly in the next subsection. 
\subsection{The surface correlation distributions}
Before getting to the mathematical definition of the surface correlation distributions, let us give a brief motivation of our approach based on a diagrammatic approach to the renormalization problem of the semi-infinite scalar field model. The propagator associated to the b.c. $\star$ can be decomposed into a sum of the two contributions given in (\ref{decom}), where $C_B^{\Lambda,\Lambda_0}$ is the regularized bulk propagator which is responsible for the singularities arising from coalescing of points and $C_{S,\star}^{\Lambda,\Lambda_0}$ is the part which is responsible of singularities arising when a point approaches the surface. Therefore, an arbitrary Feynman diagram of the semi-infinite model can be written as the sum of a diagram which contains only bulk internal lines consisting of propagators $C_B^{\Lambda,\Lambda_0}$ only, and other diagrams which contain at least one surface internal line given by the propagator $C_{S,\star}^{\Lambda,\Lambda_0}$. Renormalizing the massive semi-infinite model then amounts to renormalizing the diagrams with only bulk internal lines and those with at least a surface internal line. This approach has the advantage to disentangle the surface divergences from the bulk divergences. From the renormalization group point of view, we proceed similarly by writing (\ref{S1ln}) with $\mathcal{D}_{l,n}^{\Lambda,\Lambda_0}$ (resp. $\mathcal{S}_{l,n;\star}^{\Lambda,\Lambda_0}$) consisting of all connected amputated diagrams with $n$ external legs and $l$ loops involving exclusively  $C_B^{\Lambda,\Lambda_0}$ (resp. $C_B^{\Lambda,\Lambda_0}$ and at least one $C_{S,\star}^{\Lambda,\Lambda_0}$). 
Using the flow equations (\ref{FED}) and (\ref{FEL}), we obtain the flow equations verified by the surface correlation distributions $\mathcal{S}_{l,n;\star}^{\Lambda,\Lambda_0}\left((z_1,p_1),\cdots,(z_n,p_n)\right)$ 
\begin{align}\label{FEB}
 \partial&_{\Lambda}\partial^w\mathcal{S}_{l,n;\star}^{\Lambda,\Lambda_0}\left((z_1,p_1),\cdots,(z_n,p_n)\right)\nonumber\\
&=\frac{1}{2}\int_{z,z'}\int_k 
\partial^w\mathcal{S}_{l-1,n+2;\star}^{\Lambda,\Lambda_0}
\left((z_1,p_1),\cdots,(z_n,p_n),(z,k),(z',-k)\right)
\dot{C}^{\Lambda}_{\star}(k;z,z')\nonumber\\
&+\frac{1}{2}\int_{z,z'}\int_k 
\partial^w\mathcal{D}_{l-1,n+2}^{\Lambda,\Lambda_0}
\left((z_1,p_1),\cdots,(z_n,p_n),(z,k),(z',-k)\right)
\dot{C}^{\Lambda}_{S,\star}(k;z,z')\nonumber\\
   ~ &-\frac{1}{2}\int_{z,z'} \sum_{l_1,l_2}'\sum_{n_1,n_2}'\sum_{w_i}c_{w_i}\nonumber\\ &\left[\partial^{w_1}\mathcal{S}_{l_1,n_1+1;\star}^{\Lambda,\Lambda_0}((z_1,p_1),\cdots,(z_{n_1},p_{n_1}),(z,p))\partial^{w_3}\dot{C}^{\Lambda}_{\star}(p;z,z')\ 
\partial^{w_2}\mathcal{S}_{l_2,n_2+1;\star}^{\Lambda,\Lambda_0}((z',-p),\cdots,(z_{n},p_{n}))\right.\nonumber\\
&+\partial^{w_1}\mathcal{D}_{l_1,n_1+1}^{\Lambda,\Lambda_0}((z_1,p_1),\cdots,(z_{n_1},p_{n_1}),(z,p))\partial^{w_3}\dot{C}^{\Lambda}_{\star}(p;z,z')\ 
\partial^{w_2}\mathcal{S}_{l_2,n_2+1;\star}^{\Lambda,\Lambda_0}((z',-p),\cdots,(z_{n},p_{n}))\nonumber\\
&+\partial^{w_1}\mathcal{S}_{l_1,n_1+1;\star}^{\Lambda,\Lambda_0}((z_1,p_1),\cdots,(z_{n_1},p_{n_1}),(z,p))\partial^{w_3}\dot{C}^{\Lambda}_{\star}(p;z,z')
\partial^{w_2}\mathcal{D}_{l_2,n_2+1}^{\Lambda,\Lambda_0}((z',-p),\cdots,(z_{n},p_{n}))\nonumber\\
&\left.+\partial^{w_1}\mathcal{D}_{l_1,n_1+1}^{\Lambda,\Lambda_0}((z_1,p_1),\cdots,(z_{n_1},p_{n_1}),(z,p))\partial^{w_3}\dot{C}^{\Lambda}_{S,\star}(p;z,z')
\partial^{w_2}\mathcal{D}_{l_2,n_2+1}^{\Lambda,\Lambda_0}((z',-p),\cdots,(z_{n},p_{n}))\right]_{rsym},\nonumber\\
   &~~~~~~ p=-p_1-\cdots-p_{n_1}=p_{n_1+1}+\cdots+p_n\ .
\end{align}
For the tree order $l=0$ we have
\begin{equation}\label{tree}
    \mathcal{S}_{0,4;\star}^{\Lambda,\Lambda_0}\left((z_1,p_1),\cdots,(z_4,p_4)\right)=0~.
\end{equation}
The existence of $\mathcal{S}_{l,n}^{\Lambda,\Lambda_0}\left((z_1,p_1),\cdots,(z_n,p_n)\right)$ is ensured by (\ref{tree}) and by the flow equations (\ref{FEB}) through induction in $n+2l$ and in $l$ for fixed $n+2l$.
\subsection{Boundary and renormalization conditions}
\noindent For $\phi_1$ and $\phi_2$ in $\mathcal{S}(\mathbb{R}^+)$, the relevant terms are contained in 
\begin{equation}\label{jiji}
    \mathcal{S}_{l,2;\star}^{\Lambda,\Lambda_0}\left(0,0\right):=\int_{z_1,z_2}\mathcal{S}_{l,2;\star}^{\Lambda,\Lambda_0}\left((z_1,0),(z_2,0)\right)\phi_1(z_1)\phi_2(z_2).
\end{equation}
They are extracted from (\ref{jiji}) by performing a Taylor expansion of the test functions $\phi_1$ and $\phi_2$ around $0$ which gives \footnote{$(\partial_n\phi)(0)=\lim_{z\rightarrow 0}(\partial_z\phi)(z)$}
\begin{multline}\label{f0}
   \mathcal{S}_{l,2;\star}^{\Lambda,\Lambda_0}\left(0,0\right)
    =s_{l;\star}^{\Lambda,\Lambda_0}\phi_1(0)\phi_2(0)+e_{l;;\star}^{\Lambda,\Lambda_0}\phi_1(0)(\partial_{n}\phi_2)(0)+h_{l;\star}^{\Lambda_0}\phi_2(0)(\partial_{n}\phi_1)(0)\\+l_{l,2;\star}^{\Lambda,\Lambda_0}(\phi_1,\phi_2)\ .
\end{multline}
Then the relevant terms $s_{l;\star}^{\Lambda,\Lambda_0},~e_{l;\star}^{\Lambda,\Lambda_0}$ and $h_{l;\star}^{\Lambda,\Lambda_0}$ are obtained as 
\begin{multline}\label{eands}
    s_{l;\star}^{\Lambda,\Lambda_0}:=\int_{z_1,z_2}\mathcal{S}_{l,2}^{\Lambda,\Lambda_0}\left((z_1,0),(z_2,0)\right),~~e_{l;\star}^{\Lambda,\Lambda_0}:=\int_{z_1,z_2}z_2~\mathcal{S}_{l,2}^{\Lambda,\Lambda_0}\left((z_1,0),(z_2,0)\right),\\
    h_{l;\star}^{\Lambda,\Lambda_0}:=\int_{z_1,z_2}z_1~\mathcal{S}_{l,2}^{\Lambda,\Lambda_0}\left((z_1,0),(z_2,0)\right).
\end{multline}
Bose symmetry implies that 
\begin{equation}\label{eandh}
  \int_{z_1,z_2}z_2~\mathcal{S}_{l,2;\star}^{\Lambda,\Lambda_0}\left((z_1,0),(z_2,0)\right)=\int_{z_1,z_2}z_1~\mathcal{S}_{l,2;\star}^{\Lambda,\Lambda_0}\left((z_1,0),(z_2,0)\right),
\end{equation}
so that the counter-terms $e^{\Lambda,\Lambda_0}_{l;\star}$ and $h^{\Lambda,\Lambda_0}_{l;\star}$ are equal to all orders of perturbation theory. The remainder $l_{l,2;\star}^{\Lambda,\Lambda_0}(\phi_1,\phi_2)$ has the form 
\begin{align}\label{restI}
    l_{l,2;\star}^{\Lambda,\Lambda_0}&(\phi_1,\phi_2)=\left(\int_{z_1,z_2}z_1z_2~\mathcal{S}_{l,2;\star}^{\Lambda,\Lambda_0}\left((z_1,0),(z_2,0)\right)\right)\left(\partial_n\phi_1\right)(0)\left(\partial_n\phi_2\right)(0)\nonumber\\
    &+\phi_1(0)\int_{z_1,z_2}\mathcal{S}_{l,2;\star}^{\Lambda,\Lambda_0}\left((z_1,0),(z_2,0)\right)\int_0^1dt ~(1-t)\left(\partial_t^2\phi_2\right)(tz_2)~\nonumber\\
    &+\phi_2(0)\int_{z_1,z_2}\mathcal{S}_{l,2;\star}^{\Lambda,\Lambda_0}\left((z_1,0),(z_2,0)\right)\int_0^1dt ~(1-t)\left(\partial_t^2\phi_1\right)(tz_1)\nonumber\\
    &+(\partial_n\phi_1)(0)\int_{z_1,z_2}z_1~\mathcal{S}_{l,2;\star}^{\Lambda,\Lambda_0}\left((z_1,0),(z_2,0)\right)\int_0^1dt ~(1-t)\left(\partial_t^2\phi_2\right)(tz_2)\nonumber\\
    &+(\partial_n\phi_2)(0)\int_{z_1,z_2}z_2~\mathcal{S}_{l,2;\star}^{\Lambda,\Lambda_0}\left((z_1,0),(z_2,0)\right)\int_0^1dt ~(1-t)\left(\partial_t^2\phi_1\right)(tz_1)\nonumber\\
&+\int_{z_1,z_2}\mathcal{S}_{l,2;\star}^{\Lambda,\Lambda_0}\left((z_1,0),(z_2,0)\right)\left(\int_0^1dt~(1-t)\left(\partial_t^2\phi_1\right)(tz_1)\right)\times\left(\int_0^1dt'~(1-t')\left(\partial_{t'}^2\phi_2\right)(t'z_2)\right).
\end{align}
In the sequel, we use the following notations. For $\star\in\left\{R,N,D\right\}$, we write
\begin{equation*}
    \partial^w\mathcal{S}_{l,n;\star;r_1,r_2}^{\Lambda_0,\Lambda_0}
\left(\vec{p}_n;\phi_{\tau_{1,s},y_{1,s}}\right):=\int_{z_1,\cdots,z_n}z_1^{r_1}z_2^{r_2}\partial^w\mathcal{S}_{l,n;\star}^{\Lambda,\Lambda_0}\left((z_1,p_1),\cdots,(z_n,p_n)\right)\prod_{i=1}^s p_B\left(\tau_i;z_i,y_i\right),
\end{equation*}
\begin{equation*}
    \partial^w\mathcal{S}_{l,n;\star;r_1,r_2}^{\Lambda_0,\Lambda_0}
\left(\vec{p}_n;\phi_{\tau_{1,s},y_{1,s}}^{\star}\right)\\:=\int_{z_1,\cdots,z_n}z_1^{r_1}z_2^{r_2}\partial^w\mathcal{S}_{l,n;\star}^{\Lambda,\Lambda_0}\left((z_1,p_1),\cdots,(z_n,p_n)\right)\prod_{i=1}^s p_{\star}\left(\tau_i;z_i,y_i\right).
\end{equation*}
The boundary conditions imposed on $\mathcal{S}_{l,n;\star}^{\Lambda,\Lambda_0}$ are the following:\\
\begin{itemize}
    \item At $\Lambda=\Lambda_0$, we impose for $\star\in\left\{R,N\right\}$
\begin{align}
\mathcal{S}_{l,2;\star}^{\Lambda_0,\Lambda_0}
\left((z_1,p),(z_2,-p)\right)&=s_{l;\star}^{\Lambda_0,\Lambda_0}\delta_{z_1}\delta_{z_2}+e_{l;\star}^{\Lambda_0,\Lambda_0}\left(\delta_{z_1}\delta'_{z_2}+\delta'_{z_1}\delta_{z_2}\right)~,~~~\forall l\geq 1~,\label{BCDTS}\\
\mathcal{S}_{0,2;\star}^{\Lambda_0,\Lambda_0}
\left((z_1,p),(z_2,-p)\right)&=0~,\nonumber\\
\mathcal{S}_{l,n;\star}^{\Lambda_0,\Lambda_0}
\left((\vec{z}_n,\vec{p}_n)\right)&=0~,~~~\forall n\geq4,~\forall l\geq 0~.\label{BCDTS2}
\end{align}
\item  At $\Lambda=0$, we fix the renormalization conditions for $\star\in\left\{R,N\right\}$ as
\begin{equation}\label{renoS}
    s_{l;\star}^{0,\Lambda_0}=0,~~e_{l;\star}^{0,\Lambda_0}=0~.
\end{equation}
\item For Dirichlet boundary conditions we impose 
\begin{align}
\mathcal{S}_{l,n;D}^{\Lambda_0,\Lambda_0}
\left((\vec{z}_n,\vec{p}_n)\right)&=0,~~\forall n\geq2,~\forall l\geq 0~.\label{BCDDS}
\end{align}

\end{itemize}
\begin{remark}
\begin{itemize}
\item[-] The boundary conditions (\ref{BCDTS})-(\ref{renoS}) together with the flow equations (\ref{FEB}) and the tree order (\ref{tree}) define uniquely the surface correlation distributions   $$\mathcal{S}_{l,n;\star}^{\Lambda,\Lambda_0}\left((z_1,p_1),\cdots,(z_n,p_n)\right),~~~~\star\in\left\{D,R,N\right\}.$$
This can be verified inductively by taking the difference of two solutions of the flow equations which obey the same boundary conditions (\ref{BCDTS})-(\ref{renoS}) and by proving to all orders of perturbation theory that this difference vanishes. 
\item[-] We would like to emphasize w.r.t. (\ref{S1ln}) that we do not require any a priori knowledge on the semi-infinite correlation distributions $\mathcal{L}_{l,n;\star}^{\Lambda,\Lambda_0}$ to give a meaning to $\mathcal{S}_{l,n;\star}^{\Lambda,\Lambda_0}$. The flow equations (\ref{FEB}) together with the bulk correlation distributions defined in Section \ref{bulkspace}, the tree order (\ref{tree}) and the boundary conditions (\ref{BCDTS})-(\ref{renoS}) are sufficient to define uniquely the surface correlation distributions. The relation (\ref{S1ln}) implies the flow equations to be verified by $\mathcal{S}_{l,n;\star}^{\Lambda,\Lambda_0}$ such that the sum 
$$\mathcal{D}_{l,n}^{\Lambda,\Lambda_0}+\mathcal{S}_{l,n;\star}^{\Lambda,\Lambda_0}$$
is a solution to the FEs (\ref{FEL}).
\end{itemize}
\end{remark}
\section{Trees, Forests and Weight factors}\label{TFW}
The bounds on the surface and bulk correlation distributions are specified in terms of weighted trees and forests, which we define in the following, and for which we also derive some properties that will be important later. Our trees basically represent tree level Feynman graphs. However, we stress that
this analogy must not be taken too literally; the trees and the incidence number of vertices are independent of the
detailed form of the $n$-point interactions in the theory, the loop order controls the number of vertices of incidence number $2$ of the trees and forests via a bound, but there is no one-to-one correspondence between the loop order and the number of these vertices. A class of these tree/forest structures was already introduced in \cite{BorjiKopper2}, mainly the rooted trees. The novelty of this work lies in understanding the surface effects, which require the introduction of surface type trees which have an external point on the surface. The appearance of forests, which are a collection of surface trees is motivated by the structure of the FEs (\ref{FEB}), in particular the linear and quadratic terms  consisting only of bulk correlation distributions and the surface part of the propagator. These terms do not require any knowledge on the inductive bound and are bounded using (\ref{SurfHeatB}) and (\ref{c1b}), by a product of weight factors associated to a collection of surface trees.\\
First, we start with some notations that we will use in the sequel:
\begin{itemize}
    \item For $s\geq 1$, we denote by $\sigma_s$ the set $\left\{1,\cdots,s\right\}$ and for $i\leq j$ we denote by $\sigma_{i:j}$ the set $\left\{i,\cdots,j\right\}$.
    \item Let $\mathcal{P}_s$ be the set of all the partitions of $\sigma_s$. For a partition $\Pi\in\mathcal{P}_s$, we write $\Pi=\left(\pi_i\right)_{1\leq i\leq l_{\Pi}}$ with $\pi_i$ denoting an element of the partition $\Pi$ and $l_{\Pi}$ the cardinality of $\Pi$. 
    \item For $\Pi\in\mathcal{P}_s$ such that $r\in\pi_i$, we define
    \begin{equation}\label{tigre}
         \pi_i^r:=\pi_i\setminus\left\{r\right\}~,~~~~~\Pi^r:=\left(\cup_{j=1,~i\neq j}^r\pi_j\right) \cup \pi_i^r.
    \end{equation}
    \item Given $\Pi\in\mathcal{P}_{s+2}$ such that $\left\{s+1,s+2\right\}\in\pi_i$, we define the reduced sub-partition 
    \begin{equation}\label{tigre2}
        \pi_i^{s+1,s+2}:=\pi_i\setminus\left\{s+1,s+2\right\}.
    \end{equation}
    \item We denote by $\mathcal{P}^1_s$ the set of partitions which contain at least one sub-partition of length $1$ (i.e. $\exists \pi_i\subset \Pi,~|\pi_i|=1$) and $\mathcal{P}^{1;c}_s$ its complementary set.
    \item We denote by $\Tilde{\mathcal{P}}_{2;s}$ the set of partitions of length $2$ of the set $\sigma_s$. Note that $\Tilde{\mathcal{P}}_{2;s}$ is a subset of $\mathcal{P}_{s}$.
    
\end{itemize}
\subsection{Bulk trees, surface trees and forests}
\begin{itemize}
\item For $s\geq 2$, we denote by $\mathcal{T}^s$ the set of all trees that have a root vertex and $s-1$ external vertices. 
For a tree $T^s \in \mathcal{T}^s$ we will call $z_1\in \mathbb{R}^+$ its root vertex, and $Y=\left\{y_2,\cdots,y_s\right\}$ the set of points in $\mathbb{R}^{s-1}$ to be identified with its external vertices. Likewise we call $Z=\left\{z_2,\cdots,z_{r+1}\right\}$ the set of internal vertices of $T^s$ where $z_i\in\mathbb{R}^+$. For $r=0$, the set $Z$ is empty.  
\item We call $c_1=c(z_1)$ the incidence number of the root vertex, that is the number of lines of the tree that have the root vertex as an edge.
The external vertices have incidence number $1\,$ and the internal vertices have incidence number $>1\,$.
 We call a line $p$ an external line of the tree if one of its edges is in 
$Y\,$. The set of external lines is denoted $\mathcal{J}\,$. 
The remaining lines are called internal lines of the tree and 
are denoted by $\mathcal{I}\,$.
\item By $T^s_l$ we denote a tree $T^s \in \mathcal{T}^s$ satisfying $v_2+\delta_{c_1,1}\leq 3l-2+s/2$ for $l\geq1$ and satisfying $v_2=0$ for $l=0$, where $v_n$ is the number of vertices having incidence number $n$. Then $\mathcal{T}^{s}_l$ denotes the set of all trees $T^{s}_l$. We indicate the external vertices and internal vertices of the tree by writing $T^{s}_l(z_1,y_{2,s},\vec{z})$ with $y_{2,s}=(y_2,\cdots,y_s)$ and $\vec{z}=(z_2,\cdots,z_{r+1})$.
\item For $s\geq 1$, we define the set of bulk trees $\hat{\mathcal{T}}^s_l$ as the set of all trees with $s$ external vertices, no root vertex \footnote{All the vertices of incidence number one are external.} and which satisfy $v_2\leq 3l-2+\frac{s}{2}$ for all $l\geq1$, or $v_2=0$ for $l=0$.
 \item Let $s\geq 1$. For $Y_{\sigma_s}:=\left(y_1,\cdots,y_s\right)\in \mathbb{R}^s$, we define the set of surface trees $\mathcal{T}^{s,0}$ to be the set consisting of all trees of $s+1$ external vertices $\left\{y_1,\cdots,y_s,0\right\}$. In the sequel, we refer to the external vertex $0$ as the surface external vertex to distinguish it from the other external vertices.
 \item By $T^{s,0}_l$ we denote a surface tree $T^{s,0} \in \mathcal{T}^{s,0}$ satisfying $v_2\leq 3l-2+s/2$ for $l\geq1$ and satisfying $v_2=0$ for $l=0$. Then $\mathcal{T}^{s,0}_l$ denotes the set of all surface trees $T^{s,0}_l$. For a tree $T^{s,0}_l \in \mathcal{T}^{s,0}_l$, the set $\left\{y_1,\cdots,y_s,0\right\}$ of points in $\mathbb{R}$ is identified with its external vertices, and $\Vec{z}=\left(z_1,\cdots,z_{r}\right)$ such that $r\geq 1$ with the set of its internal vertices. We indicate the external vertices and the internal vertices of the tree by writing $T^{s,0}_l(Y_{\sigma_s},0,\vec{z})$. Note that this definition implies for all $l'\leq l$, $\mathcal{T}^{s,0}_{l'}\subset\mathcal{T}^{s,0}_l$.
\item Given a partition $\Pi\in\mathcal{P}_s$ and trees $T^{s_{\pi_i},0}_{l}\left(Y_{\pi_i},0,\vec{z}_{\pi_i}\right)\in\mathcal{T}^{s_{\pi_i},0}_l$, we define the forest $W^s_l(\Pi)$ as follows,
\begin{equation*}
W^s_l(\Pi)=\cup_{i=1}^{l_{\pi}}\left\{T^{s_{\pi_i},0}_{l}\left(Y_{\pi_i},0,\vec{z}_{\pi_i}\right)\right\}~~\mathrm{where}~~s_{\pi_i}:=|\pi_i|,~~\sum_{i=1}^{l_{\Pi}}s_{\pi_i}=s~~\mathrm{and}~~Y_{\pi_j}=\cup_{i\in\pi_j}\left\{y_i\right\}.
\end{equation*}
Here $|\pi_i|$ is the cardinality of the set $\pi_i$. We write shortly $T^{s_{\pi_i},0}_{l}\equiv T^{s_{\pi_i},0}_{l}\left(Y_{\pi_i},0,\vec{z}_{\pi_i}\right)$.  
Then the set of all forests $W^s_l(\Pi)$ denoted by $\mathcal{W}^s_l(\Pi)$, is defined as:
\begin{equation}\label{defs}
\mathcal{W}^s_l(\Pi):=\cup_{i=1}^{l_{\Pi}}\mathcal{T}^{s_{\pi_i},0}_l~.
\end{equation}
Note that for the trivial partition $\Pi_0=\sigma_s$, the length of the partition is equal to one. Therefore, the set $W^s_l(\sigma_s)$ reduces to surface trees $\mathcal{T}^{s,0}_l$. We write 
\begin{equation}\label{trivfor}
\mathcal{W}^s_l(\sigma_s)=\mathcal{T}^{s,0}_l~.
\end{equation}
This implies that each tree $T^{s_{\pi_i},0}_{l}\left(Y_{\pi_i},0,\vec{z}_{\pi_i}\right)$ can be identified with a forest in $\mathcal{W}^{s_{\pi_i}}_l(\sigma_{s_{\pi_i}})$, where $\sigma_{s_{\pi_i}}:=\cup_{k\in\pi_i}\left\{k\right\}$.
\item We define the global set of forests $\mathcal{W}^s_l$ by 
\begin{equation*}
\mathcal{W}^s_l:=\cup_{\Pi\in\mathcal{P}_s}\mathcal{W}^s_l(\Pi).
\end{equation*}
To illustrate these concepts, we give some examples of trees and forests for $s=3$ and $l=2$. The set of partitions is in this case 
\begin{equation}
    \mathcal{P}_3=\left\{\cup_{i=1}^3\left\{i\right\},\left\{1\right\}\cup\left\{2,3\right\},\left\{2\right\}\cup\left\{1,3\right\},\left\{3\right\}\cup\left\{1,2\right\},\sigma_3\right\}.
\end{equation}
\begin{itemize}
\item For the trivial partition $\Pi_0=\sigma_3$, the partition length $l_{\Pi}$ is equal to one and therefore the elements of the set $\mathcal{W}^3_2(\Pi_0)$ are the trees $T^{3,0}_2\in\mathcal{T}^{3,0}$ such that $v_2\leq 5$. For $v_2=3$, Figure 1 is an example of a surface tree in $\mathcal{W}^3_2(\Pi_0)$.
\begin{center}
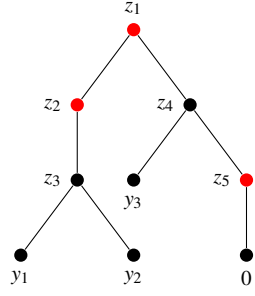
\begin{figure}
\begin{tikzpicture}[font=\footnotesize]
  \tikzset{
    level 1/.style={level distance=10mm,sibling angle=120},
    level 2/.style={level distance=10mm,sibling angle=60},
    level 3/.style={level distance=10mm,sibling angle=30},
    level 4/.style={level distance=10mm,sibling angle=30},
  }

  \node[circle,fill=red,inner sep=0pt,minimum size=5pt,label=above:{$z_1$}]{}
    child{node[circle,fill=red,inner sep=0pt,minimum size=5pt,label=left:{$z_2$}]{}
    child{node(l1)[circle,fill=black,inner sep=0pt,minimum size=5pt,label=left:{$z_3$}]{}
    child{node[circle,fill=black,inner sep=0pt,minimum size=5pt,label=below:{$y_1$}]{}edge from parent node[left]{}}
    child{node[circle,fill=black,inner sep=0pt,minimum size=5pt,label=below:{$y_2$}]{}edge from parent node[right]{}}
        edge from parent node[left]{}
     }} 
    child{node[circle,fill=black,inner sep=0pt,minimum size=5pt,label=left:{$z_4$}]{}
    child{node[circle,fill=black,inner sep=0pt,minimum size=5pt,label=below:{$y_3$}]{}edge from parent node[right]{}}
    child{node[circle,fill=red,inner sep=0pt,minimum size=5pt,label=left:{$z_5$}]{}
    child{node[circle,fill=black,inner sep=0pt,minimum size=5pt,label=below:{$0$}]{}edge from parent node[]{}}
        edge from parent node[left]{}
     }} 
  ;
\end{tikzpicture}
\caption{Example of a forest $W^3_2(\Pi_0)$ with $v_2=3$ and $\vec{z}=(z_1,\cdots,z_5)$. The red color is used for the internal vertices of incidence number $2$.}
\end{figure}
\end{center}
\item For the partition $\Pi_1=\cup_{i=1}^3 \left\{i\right\}$, an element of $\mathcal{W}^3_2(\Pi_1)$ (i.e. the set of forests of the partition $\Pi_1$) is given by the forest in Figure 2.
\begin{figure}[h!]
\centering
\begin{tikzpicture}
\node[circle,fill=black,inner sep=0pt,minimum size=5pt,label=above:{$z_3$}] (z3) at (0,0) {};
\node[circle,fill=black,inner sep=0pt,minimum size=5pt,label=above:{$z_1$}] (z1) at (-2,0) {};
\node[circle,inner sep=0pt,minimum size=5pt,fill=black,label=above:{$z_5$}] (z6) at (2,0) {};
\node[circle,fill,inner sep=0pt,minimum size=5pt,label=left:{$z_2$}] (z2) at (-2.5,-1) {};
\node[circle,fill=black,inner sep=0pt,minimum size=5pt,label=below:{$0$}] (01) at (-1.5,-1) {};
\node[circle,fill=black,inner sep=0pt,minimum size=5pt,label=below:{$y_2$}] (y2) at (-0.5,-1) {};
\node[circle,fill=black,inner sep=0pt,minimum size=5pt,label=below:{$y_1$}] (y1) at (-2.5,-2) {};
\node[circle,fill=black,inner sep=0pt,minimum size=5pt,label=left:{$z_4$}] (z5) at (0.5,-1) {};
\node[circle,fill=black,inner sep=0pt,minimum size=5pt,label=below:{$0$}] (02) at (0.5,-2) {};
\node[circle,fill=black,inner sep=0pt,minimum size=5pt,label=below:{$y_3$}] (y3) at (1.5,-1) {};
\node[circle,fill=black,inner sep=0pt,minimum size=5pt,label=below:{$0$}] (03) at (2.5,-1) {};
\draw (z1) -- (01) ;
\draw (z1) -- (z2) ;
\draw (z2) -- (y1) ;

\draw (z3) -- (z5) ;
\draw (z3) -- (y2) ;
\draw (z5) -- (02) ;

\draw (z6) -- (y3) ;
\draw (z6) -- (03) ;
\end{tikzpicture}
\caption{Example of a forest $W^3_2(\Pi_1)$ with $v_2=5$ and $\vec{z}=(z_1,\cdots,z_5)$~.}
\end{figure}
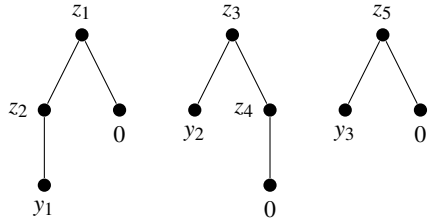\\
This forest is composed of three trees. Each tree has two external vertices. The external vertex $y_i$ has an index which belongs to the sub-partition $\left\{i\right\}$. Note that the total number of vertices of incidence number $2$ does not exceed $5$ (in this case it is equal to $5$). Note also that all the internal vertices of a surface tree with only two external vertices are of incidence number $2$.
\item For the partition $\Pi_2=\left\{1,2\right\}\cup\left\{3\right\}$, Figure 3 is an example of a forest in $\mathcal{W}^3_2(\Pi_2)$ with a total number of vertices of incidence number $2$ equal to $4$. 
\begin{figure}[h!]
\centering
\begin{tikzpicture}
\node[circle,fill=black,inner sep=0pt,minimum size=5pt,label=above:{$z_4$}] (z4) at (0,0) {};
\node[circle,fill=black,inner sep=0pt,minimum size=5pt,label=above:{$z_1$}] (z1) at (-3,0) {};
\node[circle,inner sep=0pt,minimum size=5pt,fill=black,label=above:{$z_6$}] (z6) at (0.5,-1) {};
\node[circle,fill,inner sep=0pt,minimum size=5pt,label=left:{$z_2$}] (z2) at (-3,-1) {};
\node[circle,fill=black,inner sep=0pt,minimum size=5pt,label=below:{$0$}] (01) at (-2,-1) {};
\node[circle,fill=black,inner sep=0pt,minimum size=5pt,label=below:{$y_2$}] (y2) at (-3,-2) {};
\node[circle,fill=black,inner sep=0pt,minimum size=5pt,label=below:{$y_1$}] (y1) at (-4,-1) {};
\node[circle,fill=black,inner sep=0pt,minimum size=5pt,label=left:{$z_5$}] (z5) at (-0.5,-1) {};
\node[circle,fill=black,inner sep=0pt,minimum size=5pt,label=below:{$0$}] (02) at (0.5,-2) {};
\node[circle,fill=black,inner sep=0pt,minimum size=5pt,label=below:{$y_3$}] (y3) at (-0.5,-2) {};

\draw (z1) -- (01) ;
\draw (z1) -- (z2) ;
\draw (z2) -- (y2) ;
\draw (z1) -- (y1) ;

\draw (z4) -- (z6) ;
\draw (z6) -- (02) ;
\draw (z4) -- (z5) ;
\draw (z5) -- (y3) ;
\end{tikzpicture}
\caption{Example of a forest $W^3_2(\Pi_2)\in \mathcal{W}^3_2(\Pi_2)$ with $v_2=4$~.}
\end{figure}
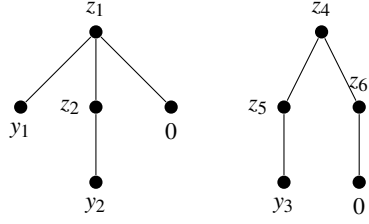
\end{itemize}
Similar examples for the forest $W^3_2(\Pi_3)$ (respectively $W^3_2(\Pi_4)$) for $\Pi_3=\left\{1,3\right\}\cup\left\{2\right\}$ (resp. $\Pi_4=\left\{2,3\right\}\cup\left\{1\right\}$) can be constructed by replacing in Figure $2$ the vertices $\left\{y_1,y_2\right\}$ by  $\left\{y_1,y_3\right\}$ and the vertex $y_3$ by $y_2$ (resp. $\left\{y_1,y_2\right\}$ by  $\left\{y_2,y_3\right\}$ and $y_3$ by $y_1$).\\
An example of a forest in the global set of forests $\mathcal{W}^3_2$ is $\cup_{i=0}^4 W^3_2(\Pi_i)$.
\end{itemize}
\subsection{Some operations on Forests and Trees}
\subsubsection{Reduction}
Let $W^{s+2}_{l-1}(\Pi)$ be a forest in $\mathcal{W}^{s+2}_{l-1}(\Pi)$. In this part, we define and explain the process of reducing the forest $W^{s+2}_{l-1}(\Pi)$ to a forest in $\mathcal{W}^s_l$. 
\begin{definition}(Reduced partition)
Let $s\geq1$ and $\Pi$ be in $\mathcal{P}_{s+2}$. 
We denote by $\pi_i$ and $\pi_j$ the sub-partitions of $\Pi$ such that ${s+1}\in \pi_i$ and ${s+2}\in\pi_j$. The reduced partition $\Pi^{s+1,s+2}$ is defined as follows,
$$\Pi^{s+1,s+2}
= \left\{
    \begin{array}{ll}
        \left\{\bigcup_{k=1,k\notin\left\{i,j\right\}}^{l_{\Pi}}\pi_k\right\}\cup \pi_i^{s+1}\cup\pi_j^{s+2} & \mbox{if } i\neq j \\
        \left\{\bigcup_{k=1,k\neq i}\pi_k\right\}\cup \pi_i^{s+1,s+2} &  \mbox{otherwise~,}
    \end{array}
\right.
$$
where we used the notation (\ref{tigre2}). 
\end{definition}
\begin{proposition}(Reduction process)
Let $s\geq 1$.    For $\Pi\in \mathcal{P}_{s+2}$, we define $C_{y_{s+1},y_{s+2}}$ to be the operator which acts on a forest $W^{s+2}_{l-1}(\Pi)\in\mathcal{W}^{s+2}_{l-1}(\Pi)$ by removing the two external legs attached to $y_{s+1}$ and $y_{s+2}$. If this operation produces an internal vertex of incidence number one, it is removed until an internal vertex of incidence number $c(z)\geq 2$ is reached. We have  
    \begin{equation}\label{cutop}
        C_{y_{s+1},y_{s+2}}W^{s+2}_{l-1}(\Pi)\in \mathcal{W}^{s}_{l}(\Pi^{s+1,s+2})~.
    \end{equation}
\end{proposition}
\begin{proof}
 The set $\mathcal{P}_{s+2}$ can be separated into two subsets $\tilde{\mathcal{P}}_{s+2}$ and $\tilde{\mathcal{P}}_{s+2}^c$ defined as follows:
\begin{itemize}
    \item $\tilde{\mathcal{P}}_{s+2}$ is defined as a subset of $\mathcal{P}_{s+2}$ which contains all the partitions $\Pi$ for which there exists $\pi_i\in\Pi$ such that $\left\{s+1,~s+2\right\}\in\pi_i$.
    \item  $\tilde{\mathcal{P}}_{s+2}^c$ is the complementary set of $\tilde{\mathcal{P}}_{s+2}$~.
    \end{itemize}
    Diagrammatically, the global set of forests $\mathcal{W}^{s+2}_{l-1}$ is partitioned into two subsets: the subset of forests for which $y_{s+1}$ and $y_{s+2}$ both belong to the same surface tree and the subset of forests in which $y_{s+1}$ and $y_{s+2}$ belong to different surface trees.
    \begin{figure}[h!]
\centering
\begin{tikzpicture}
\node[circle,fill=black,inner sep=0pt,minimum size=5pt,label=above:{$z_2$}] (z2) at (1,0) {};
\node[circle,fill=black,inner sep=0pt,minimum size=5pt,label=above:{$z_1$}] (z1) at (-3,0) {};
\node[circle,fill=black,inner sep=0pt,minimum size=5pt,label=below:{$0$}] (02) at (1.5,-1) {};
\node[circle,fill=black,inner sep=0pt,minimum size=5pt,label=below:{$y_1$}] (y1) at (0.5,-1) {};
\node[circle,fill=black,inner sep=0pt,minimum size=5pt,label=below:{$y_2$}] (y2) at (-5,-1) {};
\node[circle,fill=black,inner sep=0pt,minimum size=5pt,label=below:{$y_{s+1}$}] (ys1) at (-3.5,-1) {};
\node[circle,fill=black,inner sep=0pt,minimum size=5pt,label=below:{$0$}] (01) at (-2,-1) {};
\node[circle,fill=black,inner sep=0pt,minimum size=5pt,label=below:{$y_{s+2}$}] (ys2) at (-2.5,-1) {};

\draw (z1) -- (01) ;
\draw (z1) -- (ys2) ;
\draw (z1) -- (ys1) ;
\draw (z1) -- (y2) ;
\path (y2) -- (ys1) node [midway] {$\cdots$};
\draw (z2) -- (y1) ;
\draw (z2) -- (02) ;
\end{tikzpicture}
\caption{Example of a forest $W^{s+2}_{l-1}(\Pi)$ where $\Pi\in \tilde{\mathcal{P}}_{s+2}$ and $l_{\Pi}=2$.}
\end{figure}
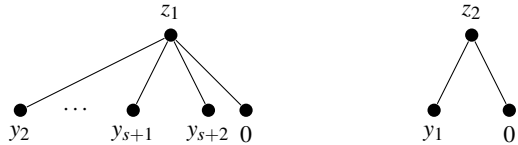
\\
 \begin{figure}[h!]
\centering
\begin{tikzpicture}
\node[circle,fill=black,inner sep=0pt,minimum size=5pt,label=above:{$z_2$}] (z2) at (1,0) {};
\node[circle,fill=black,inner sep=0pt,minimum size=5pt,label=above:{$z_1$}] (z1) at (-1,0) {};
\node[circle,fill=black,inner sep=0pt,minimum size=5pt,label=above:{$z_3$}] (z3) at (3,0) {};
\node[circle,fill=black,inner sep=0pt,minimum size=5pt,label=below:{$0$}] (02) at (1.5,-1) {};
\node[circle,fill=black,inner sep=0pt,minimum size=5pt,label=below:{$y_1$}] (y1) at (3,-1) {};
\node[circle,fill=black,inner sep=0pt,minimum size=5pt,label=below:{$y_{s+1}$}] (ys1) at (-1.5,-1) {};
\node[circle,fill=black,inner sep=0pt,minimum size=5pt,label=below:{$0$}] (01) at (-0.5,-1) {};
\node[circle,fill=black,inner sep=0pt,minimum size=5pt,label=below:{$0$}] (03) at (2.5,-1) {};
\node[circle,fill=black,inner sep=0pt,minimum size=5pt,label=below:{$y_{s}$}] (ys) at (4,-1) {};
\node[circle,fill=black,inner sep=0pt,minimum size=5pt,label=below:{$y_{s+2}$}] (ys2) at (0.5,-1) {};

\draw (z1) -- (01) ;
\draw (z1) -- (ys1) ;
\draw (z2) -- (ys2) ;
\draw (z2) -- (02) ;
\path (y1) -- (ys) node [midway] {$\cdots$};
\draw (z3) -- (y1) ;
\draw (z3) -- (03) ;
\draw (z3) -- (ys) ;
\end{tikzpicture}
\caption{Example of a forest $W^{s+2}_{l-1}(\Pi^c)$ where $\Pi^c\in \tilde{\mathcal{P}}_{s+2}^c$ and $l_{\Pi}=3$.}
\end{figure}\\
The proof of the statement (\ref{cutop}) follows directly from establishing that
\begin{equation}\label{cutop1}
    \forall \Pi\in \tilde{\mathcal{P}}_{s+2}:~C_{y_{s+1},y_{s+2}}W^{s+2}_{l-1}(\Pi)\in \mathcal{W}^s_l(\Pi^{s+1,s+2})
\end{equation}
and 
\begin{equation}\label{B9}
   \forall \Pi\in \tilde{\mathcal{P}}_{s+2}^c:~C_{y_{s+1},y_{s+2}}W^{s+2}_{l-1}(\Pi)\in \mathcal{W}^s_l(\Pi^{s+1,s+2})~.
\end{equation}
\begin{itemize}
    \item First, we prove (\ref{cutop1}). 
Given a partition $\Pi$ in $\tilde{\mathcal{P}}_{s+2}$, there exists a sub-partition $\pi_i\in\Pi$ such that $\left\{s+1,s+2\right\}\in\pi_i$. Therefore, we can write in slightly abusive notation
    \begin{equation}
        C_{y_{s+1},y_{s+2}}W^{s+2}_{l-1}(\Pi)=\bigcup_{k=1,k\neq i}^{l_{\pi}}\left\{T^{s_k,0}_{l-1}\left(Y_{\pi_k},0,\vec{z}_{\pi_k}\right)\right\}\bigcup \left\{C_{y_{s+1},y_{s+2}}T^{s_{\pi_i},0}_{l-1}\left(Y_{\pi_i},0,\vec{z}_{\pi_i}\right)\right\},
    \end{equation}
    where the tree $T^{s_{\pi_i},0}_{l-1}\left(Y_{\pi_i},0,\vec{z}_{\pi_i}\right)$ can be identified with a forest in $\mathcal{W}^{s_{\pi_i}}_{l-1}\left(\sigma_{s_{\pi_i}}\right)$.
    Deducing (\ref{cutop1}) amounts to prove for $s_{\pi_i}>2$
    \begin{equation}
    C_{y_{s+1},y_{s+2}}T^{s_{\pi_i},0}_{l-1}\left(Y_{\pi_i},0,\vec{z}_{\pi_i}\right)\in \mathcal{T}^{s_{\pi_i}-2,0}_l~.
    \end{equation}
    For $s_{\pi_i}=2$
    \begin{equation}\label{trivial}
    C_{y_{s+1},y_{s+2}}T^{2,0}_{l-1}\left(y_{s+1},y_{s+2},0,\vec{z}_{\pi_i}\right)=\varnothing~,
    \end{equation}
   and we have 
    \begin{equation}
        C_{y_{s+1},y_{s+2}}W^{s+2}_{l-1}(\Pi)=\bigcup_{k=1,k\neq i}^{l_{\pi}}\left\{T^{s_k,0}_{l-1}\left(Y_{\pi_k},0,\vec{z}_{\pi_k}\right)\right\},
    \end{equation}
    which is clearly in $W^s_l\left(\Pi^s\right)$. To treat the case $s_{\pi_i}>2$, the discussion is simplified by considering the case of the trivial partition $\Pi=\sigma_{s+2}$ s.t. $s\geq 1$. In this case, the set of forest $\mathcal{W}^{s+2}_{l-1}(\sigma_{s+2})$ is given by all the surface trees $T^{s+2,0}_{l-1}$ (see (\ref{trivfor})). Let $J_i$ (resp. $J_j$) be the external line which attaches the internal vertex $z_i$ (resp. $z_j$) to the external vertex $y_{s+1}$ (resp. $y_{s+2}$). The operator $C_{y_{s+1},y_{s+2}}$ removes the external legs $J_i$ and $J_j$ from the forest $W^{s+2}_{l-1}(\sigma_{s+2})$, and if one of the internal vertices $z_i$ and $z_j$ becomes of incidence number one, it is removed and the process continues until an internal vertex $z$ of incidence number $c(z)\geq 2$ is reached. This implies that $v'_2$ (i.e. the number of vertices of incidence number $2$ of the new forest $C_{y_{s+1},y_{s+2}}W^{s+2}_{l-1}\left(\sigma_{s+2}\right)$) is at most $v_2+2$, with $v_2$ the number of vertices of incidence number $2$ of $W^{s+2}_{l-1}(\sigma_{s+2})$. Therefore,
$$v'_2\leq v_2+2\leq 3(l-1)-2+\frac{s+2}{2}+2\leq 3l-2+\frac{s}{2}~.$$
The last point to verify is that the reduction process converges for $s\geq 1$ in the sense that we have
\begin{equation}\label{RRR}
C_{y_{s+1},y_{s+2}}W^{s+2}_{l-1}\left(\sigma_{s+2}\right)\neq \varnothing~.
\end{equation} 
In order to obtain (\ref{RRR}), we need to prove that there exists at least one internal vertex $\tilde{z}$ such that $c(\tilde{z})\geq 2$. If $W^{s+2}_{l-1}(\sigma_{s+2})$ has at least one internal vertex such that $c(z)\geq4$, then eventually (\ref{RRR}) holds. If all the internal vertices are of incidence number less than or equal to $3$, then since $s\geq 1$, the tree $W^{s+2}_{l-1}(\sigma_{s+2})$ has a number of external vertices greater than or equal to $4$ (taking into account the surface external vertex $0$ as well). This implies that it has at least two internal vertices $z$ and $z'$ such that $c(z)=c(z')=3$ which leads directly to (\ref{RRR}). 
This proves (\ref{cutop1}).
\item Now, we prove (\ref{B9}). Take $\Pi\in\tilde{\mathcal{P}}_{s+2}^c$, there exist $\pi_i,~\pi_j\in\Pi$ such that $i\neq j$, $\left\{s+1\right\}\in\pi_i$ and $\left\{s+2\right\}\in\pi_j$. Therefore, we can write 
    \begin{multline}
        C_{y_{s+1},y_{s+2}}W^{s+2}_{l-1}(\Pi)=\bigcup_{k=1,k\neq i,k\neq j}^{l_{\Pi}}\left\{T^{s_{\pi_k},0}_{l-1}\left(Y_{\pi_k},0,\vec{z}_{\pi_k}\right)\right\}\cup \left\{C_{y_{s+1}}T^{s_{\pi_i},0}_{l-1}\left(Y_{\pi_i},0,\vec{z}_{\pi_i}\right)\right\}\\
        \cup \left\{C_{y_{s+2}}T^{s_{\pi_j},0}_{l-1}\left(Y_{\pi_j},0,\vec{z}_{\pi_j}\right)\right\},
    \end{multline}
where the operator $C_{y_{s+1}}$ acts on the tree $T^{s_{\pi_i},0}_{l-1}$ by removing the external leg to which $y_{s+1}$ is attached and by removing all the internal vertices which through this process become of incidence number one. Following the same steps of the discussion above, we deduce that $v'_{2,i}$ the number of vertices of incidence number $2$ of the tree $C_{y_{s+1}}T^{s_{\pi_i},0}_{l}$ is at most \footnote{$v_{2,i}$ is the number of vertices of incidence number $2$ of $T^{s_{\pi_i}+1,0}_{l}$.} $v_{2,i}+1$ and it obeys
$$v'_{2,i}\leq v_{2,i}+1\leq 3(l-1)-2+\frac{s_{\pi_i}}{2}\leq  3l-2+\frac{s_{\pi_i}-1}{2}~.$$
Here again, we need to verify that the reduction process of the forest $W^{s+2}_{l-1}$ converges in the sense of (\ref{RRR}). If $|\pi_i|=|\pi_j|=1$, then we have 
\begin{equation}
        C_{y_{s+1},y_{s+2}}W^{s+2}_{l-1}(\Pi)=\bigcup_{k=1,k\neq i,k\neq j}^{l_{\pi}}\left\{T^{s_{\pi_k},0}_{l-1}\left(Y_{\pi_k},0,\vec{z}_{\pi_k}\right)\right\}~.
    \end{equation}
If $|\pi_i|\geq 2$, we have  
$$C_{y_{s+1}}T^{s_{\pi_i},0}_{l-1}\left(Y_{\pi_i},0,\vec{z}_{\pi_i}\right)\neq \varnothing~.$$
This holds since the tree $T^{s_{\pi_i},0}_{l-1}\left(Y_{\pi_i},0,\vec{z}_{\pi_i}\right)$ has at least three external vertices which implies that there exists at least one internal vertex such that $c(z)\geq 3$, and removing at most one external leg at each step of the reduction process implies that the incidence number of $z$ is strictly greater than $1$ at the end of the process.
\end{itemize}
\end{proof}
\subsubsection{Fusion}
\noindent In this part, we define and explain the merging process of a bulk tree with a forest.
\begin{proposition}\label{BulkMerge}
For $s\geq 2$ and $l\geq 0$, we consider the partition $(\Tilde{\pi}_1,\Tilde{\pi}_2)$ in $\Tilde{\mathcal{P}}_{2;s}$ such that $|\Tilde{\pi}_i|=s_i$ and $s_1+s_2=s$. Given a partition $\Pi$ of the set $\Tilde{\pi}_2\cup \left\{s_2+1\right\}$,  we define the $a$-merging operator $M^{a}_{y_{s_1+1},y_{s_2+1}}$ acting on the forest $W^{s_2+1}_{l_2}(\Pi)$ and the bulk tree $\hat{T}^{s_1+1}_{l_1}\left(Y_{\tilde{\pi}_1},y_{s_1+1};\vec{z}\right)$ at the external vertices $y_{s_1+1}$ and $y_{s_2+1}$ following the steps below:
\begin{itemize}
\item[(a)] Let $J_{s_1+1}=(z,y_{s_1+1})$ and $J_{s_2+1}=(z',y_{s_2+1})$ be the external legs which attach respectively $y_{s_1+1}$ to the internal vertex $z\in\hat{T}^{s_1+1}_{l_1} $ and $y_{s_2+1}$ to the internal vertex $z'\in W^{s_2+1}_{l_2}(\Pi)$. In the first step of the merging process $J_{s_1+1}$ and $J_{s_2+1}$ are removed.
\item [(b)] A new internal line $(z,z')$ is added.
\end{itemize}
Similarly, we define the $b$-merging operator $M^{b}_{y_{s_1+1},y_{s_2+1}}$ acting on $W^{s_2+1}_{l_2}(\Pi)$ and $\hat{T}^{s_1+1}_{l_1}\left(Y_{\tilde{\pi}_1},y_{s_1+1},\vec{z}\right)$ following the same steps above except for adding an internal vertex of incidence number $2$, which replaces the internal line $(z,z')$ in step (b) by the two internal lines  $(z,u)$ and $(u,z')$. 
Then we claim 
\begin{equation}\label{Ma}
    M^{i}_{y_{s_1+1},y_{s_2+1}}\left(\hat{T}^{s_1+1}_{l_1}\left(Y_{\Tilde{\pi}_1},y_{s_1+1};\vec{z}\right),W^{s_2+1}_{l_2}(\Pi)\right)\in \mathcal{W}^s_l(\Pi')~,~~~i\in\left\{a,b\right\}
\end{equation}
where\footnote{we used the notation (\ref{tigre}).} $\Pi':=\Tilde{\pi}_1\bigcup\Pi^{s_2+1}$  and $l:=l_1+l_2$.\\
\end{proposition}
\begin{proof}
Let $\pi_{i}$ be the sub-partition of $\Pi$ such that $s_2+1\in \pi_i$. The merging operators (a) and (b) act only on the tree $T^{s_{\pi_i},0}_{l_2}$ since all trees corresponding to the remaining sub-partitions do not have external vertices on which the merging operators act. Therefore, without loss of generality, we simplify the discussion by considering the case of a partition $\Pi$ of length one. \\
The first and second step of the two merging processes create a tree with $s+1$ external vertices given by the set $$\left\{Y_{\Tilde{\pi}_1}\right\}\cup\left\{Y_{\Tilde{\pi}_2},0\right\}.$$
The only difference between the two cases is related to the set of internal vertices, which in case of (a) is given by the union of the internal vertices of the bulk tree $\hat{T}^{s_1+1}_{l_1}\left(Y_{\tilde{\pi}_1},y_{s_1+1},\vec{z}\right)$ and the forest $W^{s_2+1}_{l_2}(\Pi)$. For (b), a new vertex of incidence number $2$ is added, which implies
\begin{equation}\label{v2vr}
    v_{2,a}=v_{2,1}+v_{2,2},~~~~v_{2,b}=v_{2,1}+v_{2,2}+1,
\end{equation}
where $v_{2,i}$ denotes the number of vertices of incidence number $2$ of the surface tree obtained through the merging process (i). Therefore, we obtain 
\begin{equation}
    v_{2,i}\leq  3(l_1+l_2)-4+\frac{s_1+s_2+2}{2}+1=3l-2+s/2~,~~i\in\left\{a,b\right\}.
\end{equation}
This concludes that the obtained surface trees obtained through the merging processes (a) and (b) are indeed in $\mathcal{W}^s_l(\sigma_s)$.
\end{proof}  

\subsection{Weight factors}
\subsubsection{The bulk weight factors}
Let $0<\delta<1$. Given a set $\tau_{2,s}:=\left\{\tau_2,\cdots,\tau_s\right\}$ with $\tau:=\inf_{2\leq i\leq s} \tau_{i}$, a set of 
external vertices  
$y_{2,s}=\left\{y_2,\cdots,y_s\right\} \in \mathbb{R}^{s-1}\,$ 
and a set of internal vertices $\vec{z}=(z_2,\cdots,z_{r+1})
 \in (\mathbb{R}^+)^{r} ,$ and attributing positive parameters $\Lambda_{\mathcal{I}}=\left\{\Lambda_I|I\in \mathcal{I}\right\}$ to the internal lines, the weight factor $\mathcal{F}_{\delta}\left(\Lambda_{\mathcal{I}},\tau_{2,s};T^s_l(z_1,y_{2,s},\vec{z})\right)$ of a tree $T^s_l(z_1,y_{2,s},\vec{z})$ at scales $\Lambda_I$ is defined as a product of heat kernels associated with the internal and external lines of the tree. We set
\begin{equation}\label{treeStr}
\mathcal{F}_{\delta}\left(\Lambda_{\mathcal{I}},\tau_{2,s};T^s_l(z_1,y_{2,s},\vec{z})\right):=\prod_{I\in \mathcal{I}}p_B\left(\frac{1+\delta}{\Lambda_I^2};I\right)\prod_{J\in \mathcal{J}}p_B(\tau_{J,\delta};J)~,
\end{equation}
where $\tau_{J,\delta}$ denotes the entry $\tau_{i,\delta}$ in $\tau$ carrying the index of the external coordinate $y_i$ in which the external line $J$ ends, and $\tau_{i,\delta}:=(1+\delta) \tau_i$. For $I=\left\{a,b\right\}$ the notation $p_B(\frac{1+\delta}{\Lambda_I^2};I)$ stands for $p_B(\frac{1+\delta}{\Lambda_I^2};a,b)$. We also define the integrated weight factor 
\begin{equation}\label{tree0}
\mathcal{F}_{\delta}\left(\Lambda,\tau_{2,s};T^s_l;z_1,y_{2,s}\right):=\sup_{\Lambda\leq\Lambda_I\leq\Lambda_0}\int_{\vec {z}}\mathcal{F}_{\delta}\left(\Lambda_{\mathcal{I}},\tau_{2,s};T^s_l(z_1,y_{2,s},\vec{z})\right).
\end{equation}
It depends on $\Lambda_0$, but note that its limit for $\Lambda_0\rightarrow \infty$ exists, and that typically the $\sup$ is expected to be taken for the minimal values of $\Lambda$ admitted. Therefore we suppress the dependence on $\Lambda_0$ in the notation. The definitions (\ref{treeStr})-(\ref{tree0}) can be generalized to a bulk tree $\hat{T}^s_l$.
Finally we introduce the global weight factor $\mathcal{F}\left(\Lambda,\tau_{2,s},z_1,y_{2,s}\right)$, which is defined through
\begin{equation}\label{15}
     \mathcal{F}_{s,l;\delta}\left(\Lambda,\tau_{2,s},z_1,y_{2,s}\right)
:=\sum_{T^s_l \in \mathcal{T}_l^s}\mathcal{F}_{\delta}
\left(\Lambda,\tau_{2,s};T^s_l;z_1,y_{2,s}\right)\ .
\end{equation}
Similarly, we define the global bulk weight factor 
\begin{equation}
     \hat{\mathcal{F}}_{s,l;\delta}\left(\Lambda,\tau_{2,s},y_{2,s}\right)
:=\sum_{\hat{T}^s_l \in \hat{\mathcal{T}}_l^s}\mathcal{F}_{\delta}
\left(\Lambda,\tau_{2,s};\hat{T}^s_l;y_{2,s}\right)\ .
\end{equation}
If this does not lead to ambiguity we write shortly 
\begin{equation}\label{1996}
    \mathcal{F}_{s,l;\delta}^{\Lambda}\left(\tau_{2,s}\right)\equiv 
\mathcal{F}_{s,l;\delta}\left(\Lambda,\tau_{2,s},z_1,y_{2,s}\right)\ .
\end{equation}
For $s=1$ we set $\mathcal{F}_{1,l;\delta}^{\Lambda}\equiv 1$. 
\subsubsection{The surface weight factors}
\begin{itemize}
\item In the sequel, we will use the following notations:
\begin{multline}
    \tau_{\pi_i}:=\left\{\tau_k|~k\in\pi_i\right\},~~~Y_{\pi_i}:=\left(y_k\right)_{k\in\pi_i},~~~\tau_{\pi_i,\delta}:=\left\{(1+\delta)\tau_k|~k\in\pi_i\right\},~~~\tau:=\inf_{1\leq i\leq s} \tau_i.
\end{multline}
\item Let $0<\delta<1$ and ${\tau}_{1,s}:=\left\{\tau_1,\cdots,\tau_s\right\}$ such that $\tau>0$ and let
$Y_{\sigma_s}\in \mathbb{R}^{s}\,$ be the set of the external vertices. Given a partition $\Pi\in \mathcal{P}_s$, let $\vec{z}_{\Pi}=\left(\vec{z}_{\pi_1},\cdots,\vec{z}_{\pi_{l_{\Pi}}}\right)\in (\mathbb{R}^+)^{p}$, where each vector $\vec{z}_{\pi_i}$ consists of the internal vertices of the tree $T^{s_{\pi_i},0}_l$ in the forest $W^s_l(\Pi)$. We denote by $\mathcal{I}=\cup_{i=1}^{l_{\Pi}} \mathcal{I}_k$ the set of the internal lines of the trees of $W^s_l(\Pi)$ and by $\mathcal{J}=\cup_{k=1}^{l_{\Pi}}\mathcal{J}_k$ the set of the external lines which link an internal vertex to an external vertex belonging to the set $Y_{\sigma_s}$. Each set $\mathcal{I}_k$ (resp. $\mathcal{J}_k$) denotes the internal lines (resp. the external lines) of the tree $T^{s_{\pi_k},0}_l$. We also use the notation $\mathcal{J}_k^0=\left\{J^0_k|1\leq k \leq l_{\Pi}\right\}$ to denote the set of surface external lines which link an internal vertex to $0$.
 \item Attributing positive parameters $\Lambda_{\mathcal{I}}=\left\{\Lambda_I|I\in \mathcal{I}\right\}$ to the internal lines and $\tilde{\Lambda}=\left\{\tilde{\Lambda}_k|k\in \mathcal{J}_k^0\right\}$ to the surface external lines, the weight factor $\mathcal{F}^0_{\delta}\left(\Lambda_{\mathcal{I}},\tilde{\Lambda};{\tau}_{1,s};W^s_l(\Pi);\vec{z}_{\Pi};Y_{\sigma_s}\right)$ of the forest  $W^s_l(\Pi)$ at scales $\Lambda_I$ and $\tilde{\Lambda}_k$ is defined as the product of heat kernels associated to the internal and external lines of each tree of the forest. For a sub-partition $\pi_k\in\Pi$, we define the weight factor of the tree $T^{s_{\pi_k},0}_l$ as follows:
\begin{multline}\label{treeStr'}
\mathcal{F}^0_{\delta}\left(\Lambda_{\mathcal{I}_k},\tilde{\Lambda}_{k};\tau_{\pi_k};T^{s_{\pi_k},0}_l;\vec{z}_{{\pi_k}};Y_{{\pi_k}}\right)\\:=\prod_{I\in \mathcal{I}_k}p_B\left(\frac{1+\delta}{\Lambda_I^2};I\right)\prod_{J\in \mathcal{J}_k}p_B((1+\delta)\tau_{J};J)~p_B\left(\frac{1+\delta}{\tilde{\Lambda}_k^2};J_0^k\right)~,
\end{multline}
where we used the same notations as in (\ref{treeStr}) and $J_0^k$ denotes the line which links an internal vertex to the external vertex $0$ with an attributed positive parameter $\tilde{\Lambda}_k$. The weight factor of the forest $W^s_l(\Pi)$ is defined for $\Pi\in\mathcal{P}^{1,c}_s$ as follows:
\begin{equation}\label{tree1}
\mathcal{F}^0_{\delta}\left(\Lambda_{\mathcal{I}},\tilde{\Lambda};{\tau}_{1,s};W^s_l(\Pi);\vec{z}_{\Pi};Y_{\sigma_s}\right):=\prod_{\pi_k\in\Pi}\mathcal{F}^0_{\delta}\left(\Lambda_{\mathcal{I}_k},\tilde{\Lambda}_{k};\tau_{\pi_k};T^{s_{\pi_k},0}_l;\vec{z}_{{\pi_k}};Y_{{\pi_k}}\right).
\end{equation}
For $\Pi\in\mathcal{P}^1_s$, it is given by 
\begin{multline}\label{tree2}
\mathcal{F}^0_{\delta}\left(\Lambda_{\mathcal{I}},\tilde{\Lambda};{\tau}_{1,s};W^s_l(\Pi);\vec{z}_{\Pi};Y_{\sigma_s}\right):=\prod_{\pi_k}\mathcal{F}^0_{\delta}\left(\Lambda_{\mathcal{I}_k},\tilde{\Lambda}_{k};\tau_{{\pi}_k};T^{s_{\pi_k},0}_l;\vec{z}_{{\pi_k}};Y_{{\pi_k}}\right)~\\
\times \prod_{\tilde{\pi}_k}\mathcal{F}^0_{\delta}\left(\Lambda_{\mathcal{I}_k},\tilde{\Lambda}_{k};2{\tau}_{\tilde{\pi}_k};T^{s_{\tilde{\pi}_k},0}_l;\vec{z}_{\tilde{\pi}_k};Y_{\tilde{\pi}_k}\right),
\end{multline}
where the product $\prod_{\tilde{\pi}_k}$ runs over all sub-partitions in $\Pi$ of length equal to $1$.
\item We also define the integrated surface weight factor 
\begin{equation}\label{treeStr2}
\mathcal{F}^0_{\delta}\left(\Lambda,{\tau}_{1,s};W^{s}_l(\Pi);Y_{\sigma_s}\right):=\sup_{\Lambda\leq\Lambda_I,\tilde{\Lambda}_k\leq\Lambda_0}\int_{\vec {z}_{\Pi}}\mathcal{F}^0_{\delta}\left(\Lambda_{\mathcal{I}},\tilde{\Lambda};{\tau}_{1,s};W^s_l(\Pi);\vec{z}_{\Pi};Y_{\sigma_s}\right)~,
\end{equation}
where $\int_{\vec{z}_{\Pi}}:=\prod_{i=1}^{p}\int_0^{\infty}dz_i$. The weight factor associated to a global forest $W^s_l$ is defined as
\begin{equation}\label{Sur101}
\mathcal{F}^{0}_{\delta}\left(\Lambda,{\tau}_{1,s};W^s_l;Y_{\sigma_s}\right):=\sum_{\Pi\in\mathcal{P}_s}\mathcal{F}^0_{\delta}\left(\Lambda,{\tau}_{1,s};W^s_l(\Pi);Y_{\sigma_s}\right).
\end{equation}
\item We define the global surface weight factor as follows,
\begin{equation}\label{treeStr3}
\mathcal{F}_{s,l;\delta}^{0}\left(\Lambda,{\tau}_{1,s};Y_{\sigma_s}\right):=\sum_{W^s_l \in \mathcal{W}^s_l}\mathcal{F}^{0}_{\delta}\left(\Lambda,{\tau}_{1,s};W^s_l;Y_{\sigma_s}\right).
\end{equation}
If it does not lead to ambiguity we write shortly 
\begin{equation}\label{treeStr4}
    \mathcal{F}^{\Lambda,0}_{s,l;\delta}\left({\tau}_{1,s}\right)\equiv \mathcal{F}_{s,l;\delta}^{0}\left(\Lambda,{\tau}_{1,s};Y_{\sigma_s}\right).
\end{equation}
For $s=0$ we set $\mathcal{F}^{\Lambda,0}_{0,l;\delta}\equiv 1$.
\end{itemize}
\begin{remark}
    \begin{itemize}
        \item The definitions (\ref{treeStr2}) and (\ref{tree0}) imply for $ 0\leq \Lambda'\leq \Lambda$ 
\begin{equation}\label{lambWeight}
    \mathcal{F}^{\Lambda,0}_{s,l;\delta}\left({\tau}_{1,s}\right)\leq \mathcal{F}^{\Lambda',0}_{s,l;\delta}\left({\tau}_{1,s}\right),~~~~~\mathcal{F}^{\Lambda}_{s,l;\delta}\left({\tau}_{1,s}\right)\leq \mathcal{F}^{\Lambda'}_{s,l;\delta}\left({\tau}_{1,s}\right).
\end{equation}
\item Combining the bound (\ref{in1}) together with the definitions (\ref{treeStr}) and (\ref{treeStr'}), the following bounds hold for all $0<\delta< \delta'$ and $0\leq \Lambda\leq \Lambda_0$
\begin{equation}\label{deltaWeight}
    \mathcal{F}^{\Lambda,0}_{s,l;\delta}\left({\tau}_{1,s}\right)\leq O(1)~\mathcal{F}^{\Lambda,0}_{s,l;\delta'}\left({\tau}_{1,s}\right),~~~~~\mathcal{F}^{\Lambda}_{s,l;\delta}\left({\tau}_{1,s}\right)\leq O(1)~\mathcal{F}^{\Lambda}_{s,l;\delta'}\left({\tau}_{1,s}\right).
\end{equation}
The constant $O(1)$ is explicitly given by 
$$\sup_{(\mathcal{I},\mathcal{J})\in T^s_l,~T^s_l\in\mathcal{T}^s_l}C_{\delta,\delta'}^{|\mathcal{I}|+|\mathcal{J}|},$$
where $\mathcal{I}$ and $\mathcal{J}$ are respectively the set of internal and external lines of the tree $T^s_l$ and $|\cdot|$ denotes their cardinality. The constant $C_{\delta,\delta'}$ is given by (\ref{Cdelt}) for $r=0$. 
    \end{itemize}
\end{remark}

\subsection{Inequalities}
For the proof in Sec \ref{ResProof}, we need to bound the tree/forest weight factors for reduced forests and for merged trees and forests. 
\begin{lemma}\label{reduction}
(Reduction)
    Let $\tau,~\delta>0$~, $0\leq \Lambda\leq \Lambda_0$ and $Y_{\sigma_s}\in\mathbb{R}^s$, we have 
        \begin{equation}\label{redBou}
            \int_{\mathbb{R}}du~\mathcal{F}^{0}_{s+2,l-1;\delta}\left(\Lambda;\tau_{1,s},\frac{1}{2\Lambda^2},\frac{1}{2\Lambda^2};Y_{\sigma_s},u,u\right)\leq~O(1)~\Lambda~\mathcal{F}^{0}_{s,l;\delta}\left(\Lambda;\tau_{1,s};Y_{\sigma_s}\right)~,
        \end{equation}
        where the constant $O(1)$ depends on $s$ and $l$.
\end{lemma}
\begin{proof}
 Let us recall the definition of the surface weight factor, which in this case is given by 
\begin{multline}
    \mathcal{F}_{s+2,l-1;\delta}^{0}\left(\Lambda;\tau_{1,s},\frac{1}{2\Lambda^2},\frac{1}{2\Lambda^2};Y_{\sigma_{s+2}}\right)\\=\sum_{W^{s+2}_{l-1} \in \mathcal{W}^{s+2}_{l-1}}\sum_{\Pi\in\mathcal{P}_{s+2}}\mathcal{F}^0_{\delta}\left(\Lambda;\tau_{1,s},\frac{1}{2\Lambda^2},\frac{1}{2\Lambda^2};W^{s+2}_{l-1}(\Pi);Y_{\sigma_{s+2}}\right),
\end{multline}
where $y_{s+1}:=u$ and $y_{s+2}=u$. The weight factor $$\mathcal{F}^0_{\delta}\left(\Lambda;\tau_{1,s},\frac{1}{2\Lambda^2},\frac{1}{2\Lambda^2};W^{s+2}_{l-1}(\Pi);Y_{\sigma_{s+2}}\right)$$ is given by (\ref{tree1})-(\ref{tree2}). 
Let $(z_i,u)$ and $(z_j,u)$ be the external lines which attach respectively the internal vertices $z_i$ and $z_j$ to the external vertices $y_{s+1}$ and $y_{s+2}$. Using (\ref{rr+}), we obtain 
    \begin{multline}\label{semigroup}
    \int_{\mathbb{R}}du~p_B\left(\frac{\alpha_1(1+\delta)}{2\Lambda^2};z_i,u\right)p_B\left(\frac{\alpha_2(1+\delta)}{2\Lambda^2};z_j,u\right)
    =p_B\left(\frac{(\alpha_1+\alpha_2)(1+\delta)}{2\Lambda^2};z_i,z_j\right)\leq \Lambda~.
    \end{multline}
    We recall that a tree of two external vertices (including the surface external vertex $0$) corresponds to a sub-partition of length $1$ and the surface weight factor associated to these trees differs from a surface tree of three or more external vertices by a factor $2$ multiplying the parameter $\tau_i$ of the corresponding external vertex, as it appears in (\ref{tree2}). Therefore, the constants $\alpha_1$ and $\alpha_2$ take either the value $2$ or $1$ depending on whether the two external vertices at $u$ belong to a surface tree of only two external vertices or more.
    The bound (\ref{semigroup}) removes the external legs $(z_i,u)$ and $(z_j,u)$ from the forest $W^{s+2}_{l-1}(\Pi)$ by bounding their contribution in the surface weight factor by $\Lambda$. Furthermore, the property
\begin{equation}\label{inty}
    \int_0^{\infty}dz~p_B\left(\frac{1+\delta}{\Lambda_i^2};z,z'\right)\leq 1~
\end{equation}
implies that all internal vertices which after removing $(z_i,u)$ and $(z_j,u)$, their incidence number is equal to one are removed. These two steps correspond to reducing the forest $W^{s+2}_{l-1}(\Pi)$ at the external vertices $(u,u)$. Therefore, we have
\begin{multline}\label{summy}
   \int_{\mathbb{R}}du~\mathcal{F}^0_{\delta}\left(\Lambda;\tau_{1,s},\frac{1}{2\Lambda^2},\frac{1}{2\Lambda^2};W^{s+2}_{l-1}(\Pi);Y_{\sigma_{s}},u,u\right)
   \leq \Lambda~\mathcal{F}^0_{\delta}\left(\Lambda;\tau_{1,s};C_{u,u}W^{s+2}_{l-1}(\Pi);Y_{\sigma_s}\right)~.
\end{multline}
Proposition 1 gives that $C_{u,u}W^{s+2}_{l-1}(\Pi)\in \mathcal{W}^s_l\left(\Pi^{s+1,s+2}\right)$, where $\Pi^{s+1,s+2}\in\mathcal{P}_s$ is the reduced partition obtained from $\Pi$. Hence, we obtain
\begin{multline}\label{summy2}
    \sum_{W^{s+2}_{l-1} \in \mathcal{W}^{s+2}_{l-1}}\sum_{\Pi\in\mathcal{P}_{s+2}}\mathcal{F}^0_{\delta}\left(\Lambda,\tau_{1,s};C_{u,u}W^{s+2}_{l-1}(\Pi);Y_{\sigma_s}\right)\\\leq O(1)\sum_{W^{s}_{l} \in \mathcal{W}^{s}_{l}}\sum_{\Pi\in\mathcal{P}_{s}}\mathcal{F}^0_{\delta}\left(\Lambda,\tau_{1,s};W^{s}_{l}(\Pi);Y_{\sigma_s}\right).
\end{multline}
The constant $O(1)$ takes into account that the reduction operator is not a one-to-one map, in the sense that the same forest can be obtained by reducing different forests, which implies that some weight factors $\mathcal{F}^0_{\delta}\left(\Lambda,\tau_{1,s};W^{s}_{l}(\Pi);Y_{\sigma_s}\right)$ are possibly summed more than once in (\ref{summy}). Combining (\ref{summy}) and (\ref{summy2}) gives the final bound (\ref{redBou}).
\end{proof}
\begin{lemma}\label{FFfusion}
(Forest-Forest Fusion)
  Let $\delta,~\delta'>0$ and $1\leq l_1,~l_2\leq l-1$ such that $l_1+l_2=l$. Given $\left(\Tilde{\pi}_1,\Tilde{\pi}_2\right)\in\tilde{\mathcal{P}}_{2;s}$, we have
        \begin{multline}\label{106}
            \int_{\mathbb{R}}du~\mathcal{F}^0_{s_1+1,l_1;\delta}\left(\Lambda;\tau_{\tilde{\pi}_{1}},\frac{1}{2\Lambda^2};Y_{\tilde{\pi}_1},u\right)\mathcal{F}^0_{s_2+1,l_2;\delta'}\left(\Lambda;\tau_{\tilde{\pi}_{2}},\frac{1}{2\Lambda^2};Y_{\tilde{\pi}_2},u\right)\\
            \leq ~\Lambda~\mathcal{F}^0_{s,l;\delta''}\left(\Lambda;\tau_{1,s};Y_{\sigma_s}\right),
        \end{multline}
        where $s_i:=|\pi_i|$ and $\delta''=\max\left(\delta,\delta'\right)$.
\end{lemma}
\begin{proof}
 Without loss of generality, we consider the ordered sub-partitions $\Tilde{\pi}_1:=\sigma_{1,s_1}$ and $\Tilde{\pi}_2:=\sigma_{s_1+1,s}$. To establish (\ref{106}), it is sufficient to bound  
\begin{multline}\label{WFT1}
    \int_{\mathbb{R}}du~\mathcal{F}^0_{s_1+1,l_1;\delta}\left(\Lambda;\tau_{1,s_1},\frac{1}{2\Lambda^2};W^{s_1+1}_{l_1}(\Pi_1);Y_{\sigma_{s_1}},u\right)\\\times\mathcal{F}^0_{s_2+1,l_2;\delta'}\left(\Lambda;\tau_{s_1+1,s},\frac{1}{2\Lambda^2};W^{s_2+1}_{l_2}(\Pi_2);Y_{\sigma_{s_1+1:s}},u\right)
\end{multline}
where $\Pi_1\in\Tilde{\mathcal{P}}_{s_1+1}$ and $\Pi_2\in\Tilde{\mathcal{P}}_{s_2+1}$. The sets $\Tilde{\mathcal{P}}_{s_1+1}$ and $\Tilde{\mathcal{P}}_{s_2+1}$ denote respectively the set of partitions of $\sigma_{s_1}\cup \left\{s+1\right\}$ and $\sigma_{s_1+1:s}\cup \left\{s+2\right\}$. Using (\ref{deltaWeight}), we can bound (\ref{WFT1}) by 
\begin{multline}\label{WFT}
    \int_{\mathbb{R}}du~\mathcal{F}^0_{s_1+1,l_1;\delta''}\left(\Lambda;\tau_{1,s_1},\frac{1}{2\Lambda^2};W^{s_1+1}_{l_1}(\Pi_1);Y_{\sigma_{s_1}},u\right)\\\times\mathcal{F}^0_{s_2+1,l_2;\delta''}\left(\Lambda;\tau_{s_1+1,s},\frac{1}{2\Lambda^2};W^{s_2+1}_{l_2}(\Pi_2);Y_{\sigma_{s_1+1:s}},u\right),
\end{multline}
where $\delta''=\max\left(\delta,\delta'\right)$.\\
Let $\pi_i$ and $\pi_j$ be respectively the sub-partitions in $\Pi_1$ and $\Pi_2$ such that $\left\{s+1\right\}\in\pi_i$ and $\left\{s+2\right\}\in\pi_j$. We denote by $z_i$ and $z_j$ the internal vertices in the sub-surface trees $T^{s_{\pi_i},0}_{l_1}\left(Y_{\pi_i},u,0\right)$ and $T^{s_{{\pi}_j},0}_{l_2}\left(Y_{\pi_j},u,0\right)$ in the forests $W^{s_1+1}_{l_1}(\Pi_1)$ and $W^{s_2+1}_{l_2}(\Pi_2)$, which are attached to $u$. As we mentioned previously, the bound (\ref{semigroup}) amputates the external legs $(z_i,u)$ and $(z_j,u)$ and bounds their contribution in (\ref{WFT}) by $\Lambda$. Furthermore, (\ref{inty}) implies that all internal vertices of incidence number $1$ are removed. The amputation can possibly create in each tree at most one internal vertex of incidence number $2$. Denoting by $T^{s_{\pi_i}-1,0}_{l_1}\left(Y_{\pi_i},0\right)$ the surface tree obtained by amputating the external leg $(z,u)$ from $T^{s_{\pi_i},0}_{l_1}\left(Y_{\pi_i},u,0\right)$, we deduce 
$$\tilde{v}_{2,i}^{s_1}\leq {v}_{2,i}^{s_1}+1\leq 3l_1-1+\frac{s_{\pi_i}-1}{2}+\frac{1}{2},$$
where $\tilde{v}_{2,i}^{s_1}$ and ${v}_{2,i}^{s_1}$ denote respectively the number of vertices of incidence number $2$ of the surface trees  $T^{s_{\pi_i}-1,0}_{l_1}\left(Y_{\pi_i},0\right)$ and $T^{s_{\pi_i},0}_{l_1}\left(Y_{\pi_i},u,0\right)$.
Since $1\leq l_1\leq l-1$, we obtain 
\begin{equation}\label{v2si}
\tilde{v}_{2,i}^{s_1}\leq 3l-3-1+\frac{s_{\pi_i}-1}{2}+\frac{1}{2}\leq 3l-2+\frac{s_{\pi_i}-1}{2}.
\end{equation}
Proceeding similarly with $T^{s_{\pi_j},0}_{l_2}\left(Y_{\pi_j},u,0\right)$, we deduce that the number of vertices of the amputated tree obeys 
\begin{equation}\label{v2sj}
\tilde{v}_{2,j}^{s_2}\leq 3l-2+\frac{s_{\pi_j}-1}{2}.
\end{equation}
From (\ref{v2si}) and (\ref{v2sj}), we deduce that  $T^{s_{\pi_i}-1,0}_{l_i}\left(Y_{\pi_i},0\right)\in  \mathcal{T}^{s_{\pi_i}-1,0}_{l}$. Therefore, we obtain that
\begin{multline}
            \int_{\mathbb{R}}du~\mathcal{F}^0_{s_1+1,l_1;\delta''}\left(\Lambda;\tau_{\pi_1},\frac{1}{2\Lambda^2};W^{s_1+1}_{l_1}(\Pi_1);Y_{\pi_1},u\right)\\\times\mathcal{F}^0_{s_2+1,l_2;\delta''}\left(\Lambda;\tau_{\pi_2},\frac{1}{2\Lambda^2};W^{s_2+1}_{l_2}(\Pi_2);Y_{\pi_2},u\right)
        \end{multline}
is bounded by 
\begin{equation}\label{siropp}
     \Lambda~\mathcal{F}^{0}_{\delta''}\left(\Lambda;\tau_{1,s_1};W^{s_1}_{l_1}(\Pi_1^{s_1+1});Y_{\sigma_{s_1}}\right)~\mathcal{F}^0_{\delta''}
\left(\Lambda;\tau_{s_1+1,s};W^{s_2}_{l_2}(\Pi_2^{s_2+1});Y_{\sigma_{s_1+1:s}}\right),
 \end{equation}
 where we used the notation (\ref{tigre}).
Note that $\Pi_1^{s+1}\cup\Pi_2^{s+2}\in\mathcal{P}_s$ together with (\ref{siropp}) gives the integrated surface weight factor of the forest $W^{s_1}_{l_1}(\Pi_1^{s_1+1})\cup W^{s_2}_{l_2}(\Pi_2^{s+1})$.
Hence (\ref{siropp}) is bounded by 
\begin{equation*}
  \Lambda~\mathcal{F}^0_{\delta''}
\left(\Lambda;\tau_{1,s};W^{s}_{l}(\Pi);Y_{\sigma_s}\right),
\end{equation*}
where $\Pi=\Pi_1^{s+1}\cup\Pi_2^{s+2}$ belongs to $\mathcal{P}_s$ and we deduce 
\begin{multline}\label{grr}
    \int_{\mathbb{R}}du~\mathcal{F}_{s_1+1,l_1;\delta}^{\Lambda;0}\left(\tau_{1,s_1},\frac{1}{2\Lambda^2};Y_{\sigma_{s_1}},u\right)~\mathcal{F}_{s_2+1,l_2;\delta'}^{\Lambda;0}\left(\tau_{s_{1}+1,s},\frac{1}{2\Lambda^2};Y_{\sigma_{s_1+1,s}},u\right)\\
    \leq ~\Lambda~\mathcal{F}_{s,l;\delta''}^{\Lambda,0}
\left({\tau}_{1,s}\right).
\end{multline}
\end{proof}
\begin{lemma}\label{TFfusion}
(Bulk tree-Forest Fusion) Let $\delta,~\delta'>0$ and $1\leq l_1,~l_2\leq l-1$ such that $l_1+l_2=l$. Given $\left(\Tilde{\pi}_1,\Tilde{\pi}_2\right)\in\tilde{\mathcal{P}}_{2;s}$, we have
        \begin{multline}\label{iii}
            \int_{\mathbb{R}}du~\mathcal{F}^0_{s_1+1,l_1;\delta}\left(\Lambda;\tau_{\tilde{\pi}_{1}},\frac{1}{2\Lambda^2};Y_{\tilde{\pi}_1},u\right)\hat{\mathcal{F}}_{s_2+1,l_2;\delta'}\left(\Lambda;\tau_{\tilde{\pi}_{2}},\frac{1}{2\Lambda^2};Y_{\tilde{\pi}_2},u\right)\\
            \leq ~O(1)~\mathcal{F}^0_{s,l;\delta''}\left(\Lambda;\tau_{1,s};Y_{\sigma_s}\right)
        \end{multline}
        where $\delta''=\max\left(\delta,\delta'\right)$.
\end{lemma}
\begin{proof}
 Without loss of generality, we again consider the ordered sub-partitions $\sigma_{s_1}$ and $\sigma_{s_1+1:s}$.
In order to obtain the  bound (\ref{iii}), it is sufficient to bound for a given $\Pi_1\in\tilde{\mathcal{P}}_{s_1}$
 \begin{equation}\label{proFT1}
     \int_{\mathbb{R}}du~\mathcal{F}^{0}_{\delta}\left(\Lambda;\tau_{1,s_1},\frac{1}{2\Lambda^2};W^{s_1+1}_{l_1}(\Pi_1);Y_{\sigma_{s_1}},u\right)~\mathcal{F}_{\delta'}
\left(\Lambda;\tau_{s_1+1,s},\frac{1}{2\Lambda^2};\hat{T}^{s_2+1}_{l_2};Y_{\sigma_{s_1+1:s}},u\right).
 \end{equation}
Using the bound (\ref{deltaWeight}), we bound (\ref{proFT1}) by 
 \begin{equation}\label{proFT}
     \int_{\mathbb{R}}du~\mathcal{F}^{0}_{\delta''}\left(\Lambda;\tau_{1,s_1},\frac{1}{2\Lambda^2};W^{s_1+1}_{l_1}(\Pi_1);Y_{\sigma_{s_1}},u\right)~\mathcal{F}_{\delta''}
\left(\Lambda;\tau_{s_1+1,s},\frac{1}{2\Lambda^2};\hat{T}^{s_2+1}_{l_2};Y_{\sigma_{s_1+1:s}},u\right),
 \end{equation}
 where $\delta'':=\max\left(\delta,\delta'\right)$.
 Let $z_i$ and $z_j$ be respectively the internal vertices attached to $u$ in $W^{s_1+1}_{l_1}(\Pi_1)$ and to $u$ in $\hat{T}^{s_2+1}_{l_2}$. Interchanging the integral over $u$ with the integral over the internal vertices of the forest $W^{s_1+1}_{l_1}(\Pi_1)$ and the bulk tree $\hat{T}^{s_2+1}_{l_2}$ in their respective weight factors and using (\ref{rr+}) we deduce
 \begin{equation}\label{smg}
     \int_{\mathbb{R}}du~p_B\left(\frac{\alpha(1+\delta)}{2\Lambda^2};z_i,u\right)~p_B\left(\frac{1+\delta}{2\Lambda^2};z_j,u\right)=p_B\left(\frac{(\alpha+1)(1+\delta)}{2\Lambda^2};z_i,z_j\right)
 \end{equation}
 with $\alpha\in\left\{1,2\right\}$. Here, we proceed similarly to (\ref{semigroup}) to differentiate the surface trees with two external vertices from other surface trees with more than two external vertices.
 For $\alpha=2$, we keep the integration over $u$ and write 
 \begin{multline}\label{smg2}
     \int_{\mathbb{R}}du~p_B\left(\frac{1+\delta}{\Lambda^2};z_i,u\right)~p_B\left(\frac{1+\delta}{2\Lambda^2};z_j,u\right)\\\leq 2 \int_{0}^{\infty}du~p_B\left(\frac{1+\delta}{\Lambda^2};z_i,u\right)~p_B\left(\frac{1+\delta}{2\Lambda^2};z_j,u\right).
 \end{multline}
 Therefore, (\ref{smg}) and (\ref{smg2}) correspond to the fact that the two external legs attached to $(z_i,u)$ and $(z_j,u)$ are removed. If $\alpha=1$, the external lines are replaced by the internal line $(z_i,z_j)$ and for $\alpha=2$ the vertex $u$ becomes internal with incidence number $2$. The first case corresponds to the steps of merging the forest $W^{s_1+1}_{l_1}(\Pi_1)$ and the bulk tree $\hat{T}^{s_2+1}_{l_2}$ at the external points $(u,u)$ through the process a). In the second case, the forest and the tree are merged following the merging process b). From Proposition \ref{BulkMerge} we have 
 $$M^i_{u,u}\left(W^{s_1+1}_{l_1}(\Pi_1),\hat{T}^{s_2+1}_{l_2}\right)\in\mathcal{W}^s_l(\Pi')~,~~i\in\left\{a,b\right\}$$
 where $\Pi'=\Pi_1^{s_1+1}\cup\sigma_{s_1+1:s}$. This implies that (\ref{proFT}) is  bounded by 
 $$2~\mathcal{F}^{0}_{\delta''}\left(\Lambda;\tau_{1,s};W^{s}_{l}(\Pi');Y_{\sigma_s}\right).$$
 Therefore we deduce 
 \begin{multline*}
     \int_{\mathbb{R}}du~\mathcal{F}_{s_1+1,l_1;\delta}^{0}\left(\Lambda;\tau_{\pi_1},\frac{1}{2\Lambda^2};Y_{\pi_1},u\right)\hat{\mathcal{F}}_{s_2+1,l_2;\delta'}\left(\Lambda;\tau_{\pi_2},\frac{1}{2\Lambda^2};Y_{\pi_2},u\right)\\
     \leq O(1)~\mathcal{F}_{s,l;\delta''}^{0}\left(\Lambda;{\tau}_{1,s};Y_{\sigma_s}\right)~,
 \end{multline*}
 where $O(1)$ is a constant which depends on $s$ and $l$.
\end{proof}
\section{Results and Proofs}\label{ResProof}
\noindent Our main result is summarized in the following theorem:
\begin{theorem}(Boundedness)\label{theoremReno}~Let $0\leq \Lambda\leq \Lambda_0<\infty$, $r_i\in\mathbb{N}$ such that $0\leq r_i\leq 4$ and $0\leq s \leq n$. For $\star\in\left\{R,N\right\}$, adopting (\ref{BCDTS})-(\ref{renoS}) we claim
\begin{multline}\label{c1}
    \left|  \partial^w\mathcal{S}_{l,n;\star;r_1,r_2}^{\Lambda,\Lambda_0}\left(\vec{p}_n;\phi_{\tau_{1,s},y_{1,s}}\right)\right|\\\leq \left(\Lambda+m\right)^{3-n-r_1-r_2-|w|}\mathcal{P}_1\left(\log \frac{\Lambda+m}{m}\right)\mathcal{P}_2\left(\frac{\left\|\vec{p}_n\right\|}{\Lambda+m}\right) \mathcal{Q}_1\left(\frac{\tau^{-\frac{1}{2}}}{\Lambda+m}\right)\ \mathcal{F}_{s,l;\delta}^{\Lambda,0}\left({\tau}_{1,s}\right),~~\forall n\geq 2~.
\end{multline}
Here and subsequently $\mathcal{P}_i$ and $\mathcal{Q}_i$ denote polynomials with non-negative coefficients which depend on $l,n,|w|,\delta, r_1, r_2$, but not on $\left \{ p_i \right \}$, $\Lambda$, $\Lambda_0$ and $c$. The polynomial $\mathcal{Q}_i$ is reduced to a constant for $s=1$, and for $l=0$ all polynomials $\mathcal{P}_i$ reduce to constants. The parameter $\delta$ depends on the loop order $l$ and verifies $0<\delta_l\leq\delta_{l+1}< 1$.
\end{theorem}
\noindent As a consequence of Theorem \ref{theoremReno}, we have:
\begin{proposition}\label{Prop44}
     For fixed $0\leq \Lambda\leq \Lambda_0<\infty$, $\tau>0$ and $\left(y_1,\cdots,y_n\right)\in(\mathbb{R}^+)^n$, we have: 
    \begin{equation}\label{limit}
       \mathcal{S}_{l,n;D}^{\Lambda,\Lambda_0}\left(\vec{p}_n;\prod_{i=1}^n~p_D\left(\tau_i;\cdot,y_i\right)\right) =\lim_{c\rightarrow +\infty}\mathcal{S}_{l,n;R}^{\Lambda,\Lambda_0}\left(\vec{p}_n;\prod_{i=1}^n~p_R\left(\tau_i;\cdot,y_i\right)\right) ,
    \end{equation}
    where the parameter $c$ denotes the Robin parameter. 
\end{proposition} 
\begin{corollary}\label{Cor1}
    For Dirichlet boundary conditions, adopting (\ref{BCDTS})-(\ref{BCDTS2}) we have
\begin{multline}\label{cD1}
    \left|  \mathcal{S}_{l,n;D}^{\Lambda,\Lambda_0}\left(\vec{p}_n;\prod_{i=1}^n p_D\left(\tau_i;\cdot,y_i\right)\right)\right|\\\leq \left(\Lambda+m\right)^{3-n}\mathcal{P}_3\left(\log \frac{\Lambda+m}{m}\right)\mathcal{P}_4\left(\frac{\left\|\vec{p}_n\right\|}{\Lambda+m}\right) \mathcal{Q}_2\left(\frac{\tau^{-\frac{1}{2}}}{\Lambda+m}\right)\ \mathcal{F}_{n,l;\delta}^{\Lambda,0}({\tau}_{1,n}),~~\forall n\geq 4~,
\end{multline}
and for $n=2$ we have 
\begin{multline}\label{cD2}
    \left| \mathcal{S}_{l,2;D}^{\Lambda,\Lambda_0}\left(p,-p;\prod_{i=1}^2 p_D\left(\tau_i;\cdot,y_i\right)\right)\right|\\\leq \left(\Lambda+m\right)^{-1}\tau_1^{-\frac{1}{2}}\tau_2^{-\frac{1}{2}}\mathcal{P}_5\left(\log \frac{\Lambda+m}{m}\right)\mathcal{P}_6\left(\frac{\left|{p}\right|}{\Lambda+m}\right) \mathcal{Q}_3\left(\frac{\tau^{-\frac{1}{2}}}{\Lambda+m}\right)\ \mathcal{F}_{2,l;\delta}^{\Lambda,0}({\tau}_{1,2})~.
\end{multline}
\end{corollary}  
\begin{theorem}(Convergence)\label{convergence}~
    Let $0\leq\Lambda\leq \Lambda_0<\infty$. Using the same notations, conventions and adopting the same renormalization conditions (\ref{BCDTS})-(\ref{BCDDS}) as in Theorem \ref{theoremReno} and Proposition \ref{Prop44}, we have the following bounds
\begin{multline}\label{convergence1}
    (A)~\left| \partial_{\Lambda_0}\partial^w \mathcal{S}_{l,n;\star;r_1,r_2}^{\Lambda,\Lambda_0}(\vec{p}_n;\phi_{\tau_{1,s},y_{1,s}})\right| \leq \frac{\left(\Lambda+m\right)^{4-n-|w|-r_1-r_2}}{\left(\Lambda_0+m\right)^2}\tilde{\mathcal{P}}_1\left(\log \frac{\Lambda_0+m}{m}\right)\tilde{\mathcal{P}}_2\left(\frac{\left\|\vec{p}_n\right\|}{\Lambda+m}\right)\\\times \tilde{\mathcal{Q}}_1\left(\frac{\tau^{-\frac{1}{2}}}{\Lambda+m}\right)\mathcal{F}^{\Lambda,0}_{s,l;\delta}(\tau_{1,s}),~~~\forall n+|w|+r_1+r_2\geq 2,~\star \in\left\{N,R\right\}.
\end{multline}
\begin{multline}\label{convergenced2}
     (B)~\left| \partial_{\Lambda_0} \mathcal{S}_{l,n;D}^{\Lambda,\Lambda_0}(\vec{p}_n;\phi^{D}_{\tau_{1,n},y_{1,n}})\right|\\ \leq \frac{\left(\Lambda+m\right)^{4-n}}{\left(\Lambda_0+m\right)^2}\tilde{\mathcal{P}}_1\left(\log \frac{\Lambda_0+m}{m}\right)\tilde{\mathcal{P}}_2\left(\frac{\left\|\vec{p}_n\right\|}{\Lambda+m}\right) \tilde{\mathcal{Q}}_1\left(\frac{\tau^{-\frac{1}{2}}}{\Lambda+m}\right)\mathcal{F}^{\Lambda;0}_{n,l;\delta}(\tau_{1,n}),~~~\forall n\geq 4.
\end{multline}
\begin{multline}\label{convergenced3}
     (C)~\left| \partial_{\Lambda_0} \mathcal{S}_{l,2;D}^{\Lambda,\Lambda_0}(\vec{p}_n;\phi^{D}_{\tau_{1,n},y_{1,n}})\right|\\ \leq \tau^{-1}\left(\Lambda_0+m\right)^{-2}\tilde{\mathcal{P}}_1\left(\log \frac{\Lambda_0+m}{m}\right)\tilde{\mathcal{P}}_2\left(\frac{\left\|\vec{p}_n\right\|}{\Lambda+m}\right) \tilde{\mathcal{Q}}_1\left(\frac{\tau^{-\frac{1}{2}}}{\Lambda+m}\right)\mathcal{F}^{\Lambda;0}_{2,l;\delta}(\tau_{1,2}).
\end{multline}
\end{theorem}
\begin{remark}
\begin{itemize}
\item There are two differences between the Robin/Neumann case (\ref{c1}) and the Dirichlet case (\ref{cD1})-(\ref{cD2}): The boundary conditions (\ref{BCDDS}) for $\mathcal{S}_{l,n;D}^{\Lambda,\Lambda_0}$ are imposed at scale $\Lambda=\Lambda_0$ only, whereas for $\mathcal{S}_{l,n;\star;r_1,r_2}^{\Lambda,\Lambda_0}$ we imposed mixed boundary conditions (\ref{BCDTS})-(\ref{renoS}). The second difference concerns the type of test functions considered, which in the case of Dirichlet are product of Dirichlet heat kernels (i.e. $\prod_{i=1}^n p_D\left(\tau_i;z_i,y_i\right)$), whereas in the case of Robin and Neumann b.c. the test functions are product of bulk heat kernels and characteristic functions of the semi-lines (i.e. $\prod_{i=1}^s p_B\left(\tau_i;z_i,y_i\right)\prod_{i=s+1}^n \chi^+(z_i)$).
\item The bounds (\ref{c1}) and (\ref{cD1})-(\ref{cD2}) can be established by induction separately using the associated flow equations. For the Dirichlet boundary conditions, the associated flow equations are integrated from $\Lambda$ to $\Lambda_0$. For the Robin/Neumann cases, the flow equations are integrated from $0$ to $\Lambda$ for the relevant terms using the boundary condition (\ref{renoS}) and from $\Lambda$ to $\Lambda_0$ for the irrelevant terms using the boundary condition (\ref{BCDTS}).
\item Adopting the boundary conditions (\ref{Dl1})-(\ref{Dl3}) together with (\ref{BCDTS})-(\ref{renoS}), the distributions $\mathcal{D}_{l,n}^{\Lambda,\Lambda_0}$ and $\mathcal{S}_{l,n;\star}^{\Lambda,\Lambda_0}$ are uniquely defined as the solutions of the flow equations (\ref{FED}) and (\ref{FEB}). Furthermore, their sum 
\begin{equation}\label{dec}
\mathcal{L}_{l,n;\star}^{\Lambda,\Lambda_0}=\mathcal{D}_{l,n}^{\Lambda,\Lambda_0}+\mathcal{S}_{l,n;\star}^{\Lambda,\Lambda_0}~,
\end{equation}
is the unique solution of the flow equations (\ref{FEL}) such that $\mathcal{D}_{l,n}^{\Lambda,\Lambda_0}$ and $\mathcal{S}_{l,n;\star}^{\Lambda,\Lambda_0}$ obey respectively (\ref{Dl1})-(\ref{Dl3}) and (\ref{BCDTS})-(\ref{renoS}). Theorem 1 together with Proposition \ref{Prop1} gives for $s\geq 1$
\begin{multline}\label{Bound}
    \left|\partial^w \mathcal{L}_{l,n;r,\star}^{\Lambda,\Lambda_0}\left(\vec{p}_n;\prod_{i=1}^sp_B\left(\tau_i;\cdot,y_i\right)\right)\right|\\\leq \left\{\left(\Lambda+m\right)^{4-n-|w|-r}\hat{\mathcal{F}}_{s,l;\delta}^{\Lambda}\left({\tau}_{1,s}\right)+\left(\Lambda+m\right)^{3-n-|w|-r}\mathcal{F}_{s,l;\delta}^{\Lambda,0}\left({\tau}_{1,s}\right)\right\}\\\times \mathcal{P}_5\left(\log \frac{\Lambda+m}{m}\right)\mathcal{P}_6\left(\frac{\left\|\vec{p}_n\right\|}{\Lambda+m}\right) \mathcal{Q}_3\left(\frac{\tau^{-\frac{1}{2}}}{\Lambda+m}\right),
\end{multline}
where 
\begin{multline*}
    \partial^w \mathcal{L}_{l,n;r,\star}^{\Lambda,\Lambda_0}\left(\vec{p}_n;\prod_{i=1}^sp_B\left(\tau_i;\cdot,y_i\right)\right)\\:=\int_{z_1,\cdots,z_n}(z_1-z_2)^r\partial^w\mathcal{L}_{l,n;\star}^{\Lambda,\Lambda_0}\left((z_1,p_1),\cdots,(z_n,p_n)\right)\prod_{i=1}^sp_B\left(\tau_i;z_i,y_i\right)
\end{multline*}
and $\hat{\mathcal{F}}_{s,l;\delta}^{\Lambda}\left({\tau}_{1,s}\right):=\int_{z_1}~p_B\left(\tau_1;z_1,y_1\right)\mathcal{F}_{s,l;\delta}^{\Lambda}(\tau_{2,s})$.\\ The bound (\ref{Bound}) implies that $\partial^w \mathcal{L}_{l,n;r\star}^{\Lambda,\Lambda_0}\left(\vec{p}_n;\prod_{i=1}^sp_B\left(\tau_i;\cdot,y_i\right)\right)$ are bounded uniformly w.r.t. $\Lambda_0$. It is also possible to deduce a convergence Theorem which implies the existence of the limit $\Lambda\rightarrow 0$ and $\Lambda_0\rightarrow \infty$ for $\partial^w \mathcal{L}_{l,n;r\star}^{\Lambda,\Lambda_0}$ which we do not explicit here. We refer to \cite{BorjiKopper2} for more details.
\item We do not prove Theorem \ref{convergence} since there is no novelty in the proof, which is mainly based on combining arguments from the proof of Theorem \ref{theoremReno} with the steps of the proof of the convergence theorem in \cite{BorjiKopper2}. 
\item The difference between $\mathcal{D}_{l,n}^{\Lambda,\Lambda_0}+\mathcal{S}_{l,n;\star}^{\Lambda,\Lambda_0}$, and $\mathcal{L}_{l,n;\star}^{\Lambda,\Lambda_0}$ studied in \cite{BorjiKopper2}, is their distributional structure, in the sense that one can prove inductively using the FEs (\ref{FEB}) and the boundary conditions (\ref{BCDTS})-(\ref{BCDTS2}) that 
\begin{multline}
    \mathcal{D}_{l,n}^{\Lambda,\Lambda_0}\left(z_1;\phi_{\tau_{2,s},y_{2,s}}\right)+\mathcal{S}_{l,n;\star}^{\Lambda,\Lambda_0}\left(z_1;\phi_{\tau_{2,s},y_{2,s}}\right)=a_{l,n;\star}^{\Lambda,\Lambda_0}\left(z_1,y_{2,s},\tau_{2,s}\right)+b_{l,n;\star}^{\Lambda,\Lambda_0}(y_{2,s},\tau_{2,s})\delta_{z_1}\\+c_{l,n;\star}^{\Lambda,\Lambda_0}(y_{2,s},\tau_{2,s})\delta'_{z_1},
\end{multline}
where $a_{l,n;\star}^{\Lambda,\Lambda_0}$ is smooth w.r.t. $z_1$ and $a_{l,n;\star}^{\Lambda,\Lambda_0}$, $b_{l,n;\star}^{\Lambda,\Lambda_0}$ and $c_{l,n;\star}^{\Lambda,\Lambda_0}$ are smooth w.r.t. $y_{2,s}$ and $\tau_{2,s}$. However, the semi-infinite correlation distributions $\mathcal{L}_{l,n;\star}^{\Lambda,\Lambda_0}$ considered in \cite{BorjiKopper2} are smooth w.r.t. $z_1$ which is a consequence of the type of mixed b.c.s imposed on the semi-infinite correlation distributions.
\item Note that if the bulk correlation distributions obey the bound
\begin{multline}
    \left|  \partial^w\mathcal{D}_{l,n;r_1,r_2}^{\Lambda,\Lambda_0}\left(\vec{p}_n;\phi_{\tau_{1,s},y_{1,s}}\right)\right|\\\leq \left\{\left(\Lambda+m\right)^{3-n-r_1-r_2-|w|}\mathcal{F}_{s,l;\delta}^{\Lambda,0}\left({\tau}_{1,s}\right)+\left(\Lambda+m\right)^{4-n-r_1-r_2-|w|}\mathcal{F}_{s,l;\delta}^{\Lambda}\left({\tau}_{1,s}\right)\right\}\\\times\mathcal{P}_1\left(\log \frac{\Lambda+m}{m}\right)\mathcal{P}_2\left(\frac{\left\|\vec{p}_n\right\|}{\Lambda+m}\right) \mathcal{Q}_1\left(\frac{\tau^{-\frac{1}{2}}}{\Lambda+m}\right)\ ,~~\forall n\geq 2~,~~\forall s \geq 1
\end{multline}
instead of (\ref{c1b})-(\ref{c1bb}), the bound (\ref{c1}) still holds. 
\item The bound (\ref{c1}) holds also for the surface correlation distributions folded with $\star$ heat kernels (i.e. $\star\in\left\{N,R\right\}$), that is 
\begin{multline}\label{c1'}
    \left|  \partial^w\mathcal{S}_{l,n;\star;r_1,r_2}^{\Lambda,\Lambda_0}\left(\vec{p}_n;\phi_{\tau_{1,s},y_{1,s}}^{\star}\right)\right|\\\leq \left(\Lambda+m\right)^{3-n-r_1-r_2-|w|}\mathcal{P}_1\left(\log \frac{\Lambda+m}{m}\right)\mathcal{P}_2\left(\frac{\left\|\vec{p}_n\right\|}{\Lambda+m}\right) \mathcal{Q}_1\left(\frac{\tau^{-\frac{1}{2}}}{\Lambda+m}\right)\ \mathcal{F}_{s,l;\delta}^{\Lambda,0}\left({\tau}_{1,s}\right),~~\forall n\geq 2~,
\end{multline}
where the external points $y_{1,s}$ belong to $(\mathbb{R}^+)^s$. This is a direct consequence of (\ref{pN})-(\ref{pR}) together with the bounds (\ref{bulkBou}). In particular, the bound (\ref{c1'}) implies that $\mathcal{S}_{l,n;R;r_1,r_2}^{\Lambda,\Lambda_0}\left(\vec{p}_n;\phi_{\tau,y_{1,s}}^{R}\right)$ is uniformly bounded w.r.t. the Robin parameter $c$.
    \end{itemize}
    \end{remark}
\subsection{Proof of Theorem 1}
\textit{Outline of the proof:} The bound (\ref{c1}) and (\ref{cD1})-(\ref{cD2}) are proven inductively using the standard inductive scheme which proceeds upwards in $l$, 
for given $l$ upwards in $n$, and for given $(n,l)$ downwards in $|w|$ starting from some arbitrary $|w_{\max}|\geq 3$. The bounds (\ref{c1}) and (\ref{cD1})-(\ref{cD2}) can be proven separately. Let us explain the general steps in establishing (\ref{c1}). First, we verify that the bounds (\ref{c1}) hold at the tree order. The terms on the RHS of the FE are prior to the one in the LHS in the inductive order, therefore we use the induction hypothesis for the terms in the RHS to bound the term on the LHS. Afterwards, we integrate this bound from $\Lambda$ to $\Lambda_0$ for the irrelevant terms using the boundary conditions (\ref{BCDTS2}), and from $0$ to $\Lambda$ for the relevant terms using the boundary conditions at $\Lambda=0$  fixed by the renormalization conditions (\ref{renoS}).\\ The tree and forest formalism emerges from the structure of the flow equations when considered in position space. As we mentioned before, the connected amputated Schwinger $n$-point distributions must be folded with test functions. Since the renormalization proof by the method of flow equations is inductive, the possible choices of test functions are limited by the flow equations. We choose the bulk heat kernels and this choice is not unique, but it is suitable due to its simplicity. The bound (\ref{c1}) consists of familiar factors which are also present in the inductive bounds to prove perturbative renormalization of the scalar field theory in momentuum space \cite{Keller} in $\mathbb{R}^4$, typically the power counting factor $(\Lambda+m)^{4-n-|w|-r}$ as well as the polynomials $\mathcal{P}_1$ and $\mathcal{P}_2$. For the surface weight factor, it appears mainly since we work in the $z$-direction in position space. The idea behind the global surface weight factor is to bound a complicated combinatorial object composed of intricated loops by a tree decay consisting of the product of heat kernels associated to the internal and external legs of trees and forests which are much simpler to manipulate. The behaviour w.r.t. the flow parameter $\Lambda$ in these trees and forests is traced by counting the number of vertices of incidence number $2$. This number is related to the loop order by the following bound $$v_2\leq 3l-2+\frac{s}{2},$$ 
which is compatible with the inductive scheme but is not optimal in the sense that sharper upper bounds rendering the proof more complicated could have been imposed.  \\
\begin{proof}
\noindent We establish the proof in the case of the Robin boundary conditions. For the Neumann boundary conditions, we proceed similarly. In the sequel, we omit the subscript $R$ from $\mathcal{S}_{l,n;R}^{\Lambda,\Lambda_0}$. \\
The induction starts at the tree order for which we have
$$\mathcal{S}_{0,4}^{\Lambda,\Lambda_0}\left((z_1,p_1),\cdots,(z_4,p_4)\right)=0~,$$
and the bound (\ref{c1}) obviously holds.\\
\subsubsection{The right-hand side of the FEs}
The bounds that we want to obtain for the RHS of the flow equations (\ref{FEB}) are of the form    
\begin{multline}\label{c'}
    \left| \partial_{\Lambda}\partial^w \mathcal{S}_{l,n;r_1,r_2}^{\Lambda,\Lambda_0}\left(\vec{p}_n;\phi_{\tau_{1,s},y_{1,s}}\right) \right|\\\leq \left(\Lambda+m\right)^{2-n-|w|-r_1-r_2}\mathcal{P}_1\left(\log\frac{\Lambda+m}{m}\right)\mathcal{P}_2\left(\frac{\left\|\vec{p}_n\right\|}{\Lambda+m}\right)\mathcal{Q}_1\left(\frac{\tau^{-\frac{1}{2}}}{\Lambda+m}\right)\mathcal{F}^{\Lambda,0}_{s,l;\delta}\left({\tau}_{1,s}\right)~,
\end{multline}
for all $n\geq 2$, $0\leq s \leq n$ and $0\leq r_1,~r_2\leq 4$.
In the sequel, we drop the lower indices from the polynomials $\mathcal{P}_1$, $\mathcal{P}_2$ and $\mathcal{Q}_1$. But one should keep in mind that these polynomials, whenever they appear, may have different positive coefficients which depend on $l,n,|w|,\delta_l$ only and not on $\left\{p_i\right\}$, $\Lambda$, $\Lambda_0$, ${\tau}_{1,s}$ and the Robin parameter $c$.\\
The bound (\ref{c'}) is established by bounding each of the terms on the RHS of the FE (\ref{FEB}). We consider first the case $r_1=r_2=0$.
\begin{itemize}
\item  We start by treating the linear terms $R_1^D$ and $R_1^S$ given by
\begin{multline}\label{R1S}
    R_1^S:=\int_{z,z'}\int_{z_{1,n}}\int_k 
\partial^w\mathcal{S}_{l-1,n+2}^{\Lambda,\Lambda_0}
\left((z_1,p_1),\cdots,(z_n,p_n),(z,k),(z',-k)\right)\\\times
\dot{C}^{\Lambda}(k)p_R\left(\frac{1}{\Lambda^2};z,z'\right)\prod_{i=1}^s~p_B\left(\tau_i;z_i,y_i\right)
\end{multline}
and 
\begin{multline}\label{D-term}
    R_1^D:=\int_{z,z'}\int_{z_{1,n}}\int_k 
\partial^w\mathcal{D}_{l-1,n+2}^{\Lambda,\Lambda_0}
\left((z_1,p_1),\cdots,(z_n,p_n),(z,k),(z',-k)\right)\\\times
\dot{C}^{\Lambda}(k)p_{S,R}\left(\frac{1}{\Lambda^2};z,z'\right)\prod_{i=1}^s~p_B\left(\tau_i;z_i,y_i\right)~,
\end{multline}
where $p_R$ and $p_{S,R}$ are given by (\ref{pR}) and (\ref{43}).
First, we bound $R_1^S$. Using the decomposition of the Robin heat kernel (\ref{pR}), we obtain that $R_1^S$ can be written as the sum of three contributions such that for each contribution the Robin heat kernel $p_R$ in $R_1^S$ is replaced by a term from the decomposition (\ref{pR}). We analyze first the term
\begin{multline*}
     \tilde{R}_1^S:=\int_{z,z'}\int_{z_{1,n}}\int_k 
\partial^w\mathcal{S}_{l-1,n+2}^{\Lambda,\Lambda_0}
\left((z_1,p_1),\cdots,(z_n,p_n),(z,k),(z',-k)\right)\\\times
\dot{C}^{\Lambda}(k)p_B\left(\frac{1}{\Lambda^2};z,z'\right)\prod_{i=1}^s~p_B\left(\tau_i;z_i,y_i\right)~.
\end{multline*}
Using the semi-group property for the bulk heat kernel (\ref{rr+}), $\tilde{R}_1^S$ can be rewritten as 
\begin{multline*}
\int_{\mathbb{R}}du~\int_{z,z'}\int_{z_{1,n}}\int_k 
\partial^w\mathcal{S}_{l-1,n+2}^{\Lambda,\Lambda_0}
\left((z_1,p_1),\cdots,(z_n,p_n),(z,k),(z',-k)\right)\\
\times\dot{C}^{\Lambda}(k)p_B\left(\frac{1}{2\Lambda^2};z,u\right)p_B\left(\frac{1}{2\Lambda^2};z',u\right)\prod_{i=1}^s~p_B\left(\tau_i;z_i,y_i\right).
\end{multline*}
We now insert the induction hypothesis to obtain that $\tilde{R}_1^S$ is bounded by 
\begin{multline}
\left(\Lambda+m\right)^{1-n-|w|}\mathcal{P}\left(\log \frac{\Lambda+m}{m}\right)\mathcal{Q}\left(\frac{\tau^{-\frac{1}{2}}}{\Lambda+m}\right)\int_k \dot{C}(k)\mathcal{P}\left(\frac{|k|}{\Lambda+m}, \frac{\|\vec{p}_n\|}{\Lambda+m}\right)\\\int_{\mathbb{R}}du~\mathcal{F}_{s+2,l-1;\delta_1}^{\Lambda,0}\left(\tau_{1,s},\frac{1}{2\Lambda^2},\frac{1}{2\Lambda^2};Y_{\sigma_s},u,u\right).
\end{multline}
From the bound one the $3$-dimensional covariance, we have
\begin{multline}
|\Tilde{R}_1^S|\leq ~\left(\Lambda+m\right)^{1-n-|w|}\mathcal{P}\left(\log \frac{\Lambda+m}{m}\right)\mathcal{Q}\left(\frac{\tau^{-\frac{1}{2}}}{\Lambda+m}\right)\mathcal{P}\left(\frac{\|\vec{p}_n\|}{\Lambda+m}\right)\\\int_{\mathbb{R}}du~\mathcal{F}_{s+2,l-1;\delta_1}^{\Lambda,0}\left(\tau_{1,s},\frac{1}{2\Lambda^2},\frac{1}{2\Lambda^2};Y_{\sigma_s},u,u\right).
\end{multline}
Applying Lemma \ref{reduction}, we obtain the bound
\begin{equation}
|\tilde{R}_1^S|\leq \left(\Lambda+m\right)^{2-n-|w|}\mathcal{P}\left(\log\frac{\Lambda+m}{m}\right)\mathcal{P}\left(\frac{\left\|\vec{p}_n\right\|}{\Lambda+m}\right)\mathcal{Q}\left(\frac{\tau^{-\frac{1}{2}}}{\Lambda+m}\right)\mathcal{F}^{\Lambda,0}_{s,l;\delta_1}\left({\tau}_{1,s}\right).
\end{equation}
The other contributions to $R_1^S$ are
\begin{multline}\label{R1S2}
    \int_{z,z'}\int_{z_{1,n}}\int_k 
\partial^w\mathcal{S}_{l-1,n+2}^{\Lambda,\Lambda_0}
\left((z_1,p_1),\cdots,(z_n,p_n),(z,k),(z',-k)\right)\\\times
\dot{C}^{\Lambda}(k)p_B\left(\frac{1}{\Lambda^2};z,-z'\right)\prod_{i=1}^s~p_B\left(\tau_i;z_i,y_i\right)
\end{multline}
and 
\begin{multline}\label{R1S3}
    -2~\int_{z,z'}\int_{z_{1,n}}\int_k 
\partial^w\mathcal{S}_{l-1,n+2}^{\Lambda,\Lambda_0}
\left((z_1,p_1),\cdots,(z_n,p_n),(z,k),(z',-k)\right)\\\times
\dot{C}^{\Lambda}(k)\int_{v}e^{-v}p_B\left(\frac{1}{\Lambda^2};z,-z'-\frac{v}{c}\right)\prod_{i=1}^s~p_B\left(\tau_i;z_i,y_i\right).
\end{multline}
These terms can be rewritten using (\ref{rr+}) as
 \begin{equation}\label{141}
     \int_{\mathbb{R}}du~\partial^w\mathcal{S}_{l-1,n+2}^{\Lambda,\Lambda_0}
\left(\Vec{p}_n,k,-k;\phi_{\tau_{1,s},y_{1,s}}\times p_B\left(\frac{1}{2\Lambda^2};\cdot,u\right)p_B\left(\frac{1}{2\Lambda^2};\cdot,-u\right)\right)
 \end{equation}
 and 
  \begin{equation}\label{142}
     \int_{\mathbb{R}}du~\int_{v}~e^{-{v}}~\partial^w\mathcal{S}_{l-1,n+2}^{\Lambda,\Lambda_0}
\left(\Vec{p}_n,k,-k;\phi_{\tau_{1,s},y_{1,s}}\times p_B\left(\frac{1}{2\Lambda^2};\cdot,u\right)p_B\left(\frac{1}{2\Lambda^2};\cdot,-\frac{v}{c}-u\right)\right).
 \end{equation}
Applying the induction hypothesis, we obtain that (\ref{141}) is bounded by 
\begin{multline}
\left(\Lambda+m\right)^{1-n-|w|}\mathcal{P}\left(\log \frac{\Lambda+m}{m}\right)\mathcal{Q}\left(\frac{\tau^{-\frac{1}{2}}}{\Lambda+m}\right)\mathcal{P}\left(\frac{\|\vec{p}_n\|}{\Lambda+m}\right)\\\int_{\mathbb{R}}du~\mathcal{F}_{s+2,l-1;\delta_1}^{\Lambda,0}\left(\tau_{1,s},\frac{1}{2\Lambda^2},\frac{1}{2\Lambda^2};Y_{\sigma_s},u,-u\right).
\end{multline}
Similarly, we have the following bound for (\ref{142}) 
\begin{multline}
\left(\Lambda+m\right)^{1-n-|w|}\mathcal{P}\left(\log \frac{\Lambda+m}{m}\right)\mathcal{Q}\left(\frac{\tau^{-\frac{1}{2}}}{\Lambda+m}\right)\int_k \dot{C}(k)~\mathcal{P}\left(\frac{|k|}{\Lambda+m}, \frac{\|\vec{p}_n\|}{\Lambda+m}\right)\\\int_{\mathbb{R}}du~\int_v~e^{-v}~\mathcal{F}_{s+2,l-1;\delta_1}^{\Lambda,0}\left(\tau_{1,s},\frac{1}{2\Lambda^2},\frac{1}{2\Lambda^2};Y_{\sigma_s},u,-u-\frac{v}{c}\right).
\end{multline}
Using the bounds (\ref{bulkBou})and remembering the definition (\ref{treeStr'}) of the surface weight factor, we deduce that 
$$\int_{\mathbb{R}}du~\int_v~e^{-v}~\mathcal{F}_{s+2,l-1;\delta_1}^{\Lambda,0}\left(\tau_{1,s},\frac{1}{2\Lambda^2},\frac{1}{2\Lambda^2};Y_{\sigma_s},u,-u-\frac{v}{c}\right)$$
and 
$$\int_{\mathbb{R}}du~\mathcal{F}_{s+2,l-1;\delta_1}^{\Lambda,0}\left(\tau_{1,s},\frac{1}{2\Lambda^2},\frac{1}{2\Lambda^2};Y_{\sigma_s},u,-u\right)$$
are bounded by 
$$\int_{\mathbb{R}}du~\mathcal{F}_{s+2,l-1;\delta_1}^{\Lambda,0}\left(\tau_{1,s},\frac{1}{2\Lambda^2},\frac{1}{2\Lambda^2};Y_{\sigma_s},u,u\right).$$
The rest of the proof follows the steps used to obtain the final bound for $\tilde{R}_1^S$, which gives 
\begin{equation}
|{R}_1^S|\leq \left(\Lambda+m\right)^{2-n-|w|}\mathcal{P}\left(\log\frac{\Lambda+m}{m}\right)\mathcal{P}\left(\frac{\left\|\vec{p}_n\right\|}{\Lambda+m}\right)\mathcal{Q}\left(\frac{\tau^{-\frac{1}{2}}}{\Lambda+m}\right)\mathcal{F}^{\Lambda,0}_{s,l;\delta_1}\left({\tau}_{1,s}\right).
\end{equation}
Now we analyse $R_1^D$. This term is independent of the induction hypothesis and will be bounded using only the bound (\ref{c1b}) for $\mathcal{D}_{l,n}^{\Lambda,\Lambda_0}$. Using (\ref{rr+}), $R_1^D$ can be rewritten as 
\begin{multline*}
\int_{\mathbb{R}}du~\int_{z,z'}\int_{z_{1,n}}\int_k 
\partial^w\mathcal{D}_{l-1,n+2}^{\Lambda,\Lambda_0}
\left((z_1,p_1),\cdots,(z_n,p_n),(z,k),(z',-k)\right)\\\times 
\dot{C}^{\Lambda}(k)p_B\left(\frac{1}{2\Lambda^2};z,u\right)p_B\left(\frac{1}{2\Lambda^2};z',-u\right)\prod_{i=1}^s~p_B\left(\tau_i;z_i,y_i\right)~.
\end{multline*}
The bound (\ref{c1b}) implies that $R_1^D$ is bounded by 
\begin{multline}\label{79}
\left(\Lambda+m\right)^{2-n-|w|} \mathcal{P}\left(\log \frac{\Lambda+m}{m}\right)\mathcal{P}\left(\frac{\left\|\vec{p}_n\right\|}{\Lambda+m}\right)\mathcal{Q}\left(\frac{\tau^{-\frac{1}{2}}}{\Lambda+m}\right) \int_{\mathbb{R}}du~\int_{z_1}~p_B\left(\tau_1;z_1,y_1\right) \\
  \times\sum_{T^{s+2}_{l-1}(z_1,y_{2,s},u,-u)}\int_{\vec{z}}~\mathcal{F}_{\delta_2}\left(\Lambda;\left\{\tau_{2,s},\frac{1}{2\Lambda^2},\frac{1}{2\Lambda^2}\right\};T^{s+2}_{l-1}(z_1,y_{2,s},u,-u,\vec{z})\right)~.
\end{multline}
For any contribution to (\ref{79}) we denote by $z', z''$ the vertices in the tree $T^{s+2}_{l-1}(z_1,y_{2,s},u,-u)$ to which the test functions $p_B\left(\frac{1+\delta_2}{2\Lambda^2};u,\cdot\right)$ and $p_B\left(\frac{1+\delta_2}{2\Lambda^2};\cdot,-u\right)$ are attached. Performing the integral over $u$ we obtain using (\ref{rr+})
\begin{multline}\label{tuii}
\int_{\mathbb{R}}du~ p_B\left(\frac{1+\delta_2}{2\Lambda^2};z',u\right)\ 
p_B\left(\frac{1+\delta_2}{2\Lambda^2};-u,z''\right)=p_B\left(\frac{1+\delta_2}{\Lambda^2};z',-z''\right)\\\leq  p_B\left(\frac{1+\delta_2}{\Lambda^2};z',0\right).
\end{multline}
The bound (\ref{tuii}) implies that the legs $(z,u)$ and $(z',u)$ are amputated from the tree $T^{s+2}_{l-1}$ and $(z',u)$ is replaced by the surface external leg $(z',0)$ with the parameter $\Lambda$. If $z''$ is of incidence number one, it is removed using $\int_{z''}p_B\left((1+\delta_2)/\Lambda_I^2;z,z''\right)\leq 1$, and this operation is iterated until a vertex $\tilde{z}$ such that $c(\tilde{z})\geq 2$ is reached. This iteration process converges to a non-empty tree since for $s\geq 1$, there exists at least one internal vertex of incidence number greater or equal to $3$ in the tree $T^{s+2}_{l-1}$. The integration over $z_1$ in (\ref{79}) implies that $z_1$ becomes an internal vertex attached to $y_1$. Therefore, the reduction process produces a tree which belongs to $\mathcal{T}^{s,0}$. Furthermore, $v'_2$ which denotes the number of vertices of incidence number $2$ of the new tree, is increased at most by $2$. This stems from the reduction process which can produce one additional internal vertex such that $c(z)=2$ when the vertex $z''$ is removed, but also from the vertex $z_1$ which was initially a root vertex. If $z_1$ had an incidence number equal to one then after introducing the test function $p_B(\tau_1;z_1,y_1)$, it becomes internal of incidence number $2$. If $v_2$ is the number of vertices of incidence number $2$ of $T^{s+2}_{l-1}$, then $v'_2\leq v_2+\delta_{c_1,1}+1$ which implies
\begin{eqnarray*}
v'_2\leq v_2+\delta_{c_1,1}+1\leq 3(l-1)-2+\frac{s+2}{2}+1\leq3l-2+s/2~.
\end{eqnarray*}
This also means that the obtained tree is a surface tree in $\mathcal{T}_l^{s,0}\equiv\mathcal{W}^s_l\left(\sigma_s\right)$, which can also be seen as the set of forests corresponding to the trivial partition. Therefore, $R_1^D$ is bounded by 
\begin{equation}
\left(\Lambda+m\right)^{2-n-|w|} \mathcal{P}\left(\log \frac{\Lambda+m}{m}\right)\mathcal{P}\left(\frac{\left\|\vec{p}_n\right\|}{\Lambda+m}\right)\mathcal{Q}\left(\frac{\tau^{-\frac{1}{2}}}{\Lambda+m}\right)
\mathcal{F}^0_{\delta_2}\left(\Lambda;\tau_{1,s};W^s_l(\sigma_s);Y_{\sigma_s}\right)
\end{equation}
which implies (see (\ref{Sur101}))
\begin{equation}
\left|R_1^D\right|\leq\left(\Lambda+m\right)^{2-n-|w|} \mathcal{P}\left(\log \frac{\Lambda+m}{m}\right)\mathcal{P}\left(\frac{\left\|\vec{p}_n\right\|}{\Lambda+m}\right)\mathcal{Q}\left(\frac{\tau^{-\frac{1}{2}}}{\Lambda+m}\right)
\mathcal{F}^{\Lambda;0}_{s,l;\delta_2}\left(\tau_{1,s}\right)~.
\end{equation}
\item In this part, we treat the quadratic terms on the RHS of the flow equations. It is enough to analyse the terms from the symmetrized sum in which the arguments $(z_i,p_i)$ appear ordered in  $(\mathcal{S}^{\Lambda,\Lambda_0}_{l_1,n_1+1},\mathcal{S}^{\Lambda,\Lambda_0}_{l_2,n_2+1})$, $(\mathcal{D}^{\Lambda,\Lambda_0}_{l_1,n_1+1},\mathcal{S}^{\Lambda,\Lambda_0}_{l_2,n_2+1})$ and $(\mathcal{D}^{\Lambda,\Lambda_0}_{l_1,n_1+1},\mathcal{D}^{\Lambda,\Lambda_0}_{l_2,n_2+1})$. These terms are given by
\begin{multline*}
     R_2^{SS}:=\int_{z_{1,n}}\int_{z,z'}z_1^{r_1}~z_2^{r_2}~\partial^{w_1}\mathcal{S}_{l_1,n_1+1}^{\Lambda,\Lambda_0}((z_1,p_1),\cdots,(z_{n_1},p_{n_1}),(z,p))\partial^{w_3}\dot{C}^{\Lambda}(p)p_R\left(\frac{1}{\Lambda^2};z,z'\right)\\ 
\times\partial^{w_2}\mathcal{S}_{l_2,n_2+1}^{\Lambda,\Lambda_0}((z',-p),\cdots,(z_{n},p_{n}))\prod_{i=1}^s p_B(\tau_i;z_i,y_i)~,
\end{multline*}
\begin{multline*}
    R_2^{DS}:=\int_{z_{1,n}}\int_{z,z'}z_1^{r_1}~z_2^{r_2}~\partial^{w_1}\mathcal{S}_{l_1,n_1+1}^{\Lambda,\Lambda_0}((z_1,p_1),\cdots,(z_{n_1},p_{n_1}),(z,p))\partial^{w_3}\dot{C}^{\Lambda}(p)p_R\left(\frac{1}{\Lambda^2};z,z'\right)\\ 
\times\partial^{w_2}\mathcal{D}_{l_2,n_2+1}^{\Lambda,\Lambda_0}((z',-p),\cdots,(z_{n},p_{n}))\prod_{i=1}^s p_B(\tau_i;z_i,y_i)~,
\end{multline*}
and \begin{multline*}
    R_2^{DD}:=\int_{z_{1,n}}\int_{z,z'}z_1^{r_1}~z_2^{r_2}~\partial^{w_1}\mathcal{D}_{l_1,n_1+1}^{\Lambda,\Lambda_0}((z_1,p_1),\cdots,(z_{n_1},p_{n_1}),(z,p))\partial^{w_3}\dot{C}^{\Lambda}(p)p_{S,R}\left(\frac{1}{\Lambda^2};z,z'\right)\\ 
\times\partial^{w_2}\mathcal{D}_{l_2,n_2+1}^{\Lambda,\Lambda_0}((z',-p),\cdots,(z_{n},p_{n}))\prod_{i=1}^s p_B(\tau_i;z_i,y_i)~.
\end{multline*}
First, we treat the case $(r_1,r_2)=(0,0)$.
\begin{itemize}
    \item We start with the term $R_2^{DS}$. The property (\ref{10'}) implies that $R_2^{DS}$ can be rewritten as
    \begin{multline*}
    \int_{\mathbb{R}^+}du~\int_{z_{1,n}}\int_{z,z'}\partial^{w_1}\mathcal{S}_{l_1,n_1+1}^{\Lambda,\Lambda_0}((z_1,p_1),\cdots,(z_{n_1},p_{n_1}),(z,p))\partial^{w_3}\dot{C}^{\Lambda}(p)\\ 
\times\partial^{w_2}\mathcal{D}_{l_2,n_2+1}^{\Lambda,\Lambda_0}((z',-p),\cdots,(z_{n},p_{n}))\prod_{i=1}^s p_B(\tau_i;z_i,y_i)p_R\left(\frac{1}{2\Lambda^2};z,u\right)p_R\left(\frac{1}{2\Lambda^2};z',u\right).
\end{multline*}
Using the decomposition of the Robin heat kernel (\ref{pR}), we restrict our analysis to the following term only
\begin{multline*}
    \tilde{R}_2^{DS}:=\int_{\mathbb{R}^+}du~\int_{z_{1,n}}\int_{z,z'}\partial^{w_1}\mathcal{S}_{l_1,n_1+1}^{\Lambda,\Lambda_0}((z_1,p_1),\cdots,(z_{n_1},p_{n_1}),(z,p))\partial^{w_3}\dot{C}^{\Lambda}(p)\\ 
\times\partial^{w_2}\mathcal{D}_{l_2,n_2+1}^{\Lambda,\Lambda_0}((z',-p),\cdots,(z_{n},p_{n}))\prod_{i=1}^s p_B(\tau_i;z_i,y_i)p_B\left(\frac{1}{2\Lambda^2};z,u\right)p_B\left(\frac{1}{2\Lambda^2};z',u\right).
\end{multline*}
The line of reasoning in treating the remaining contributions in $R_2^{DS}$ is similar to the one used in bounding $R^S_1$.
We define
\begin{equation}\label{o1}
\phi'_{s_1}(z_{1,n_1})=\prod^{n_1}_{r=1}\phi_i(z_i)
,~~~\phi''_{s_2}(z_{n_1+1,n-1})
=\prod_{r=n_1+1}^{n}\! \phi_i(z_i)~,
\end{equation}
where 
\begin{equation}\label{o2}
\phi_i(z_i) = \left\{
    \begin{array}{ll}
        p_B(\tau_i;z_i,y_i) & \mbox{if } i\leq s \\
        \chi^+(z_i) & \mbox{otherwise~.}
    \end{array}
\right.
\end{equation}
Note that $s_1=n_1$ if $n_1\leq s$ and $s_2=s-n_1$. Otherwise, we have $s_1=s$ and $s_2=0$. Therefore, $\Tilde{R}_2^{DS}$ can be rewritten as
\begin{multline}\label{R2}
      \tilde{R}_2^{DS}=\int_{\mathbb{R}^+}du~\int_{z'}  \partial^{w_1}
\mathcal{S}^{\Lambda,\Lambda_0}_{l_1,n_1+1}
\Bigl(\vec{p}_{n_1},p;\phi'_{s_1}\times p_B\left(\frac{1}{2\Lambda^2};.,u\right)
\Bigr)
\, \partial^{w_3}\dot{C}^{\Lambda}(p)\\
     \times \partial^{w_2}\mathcal{D}^{\Lambda,\Lambda_0}_{l_2,n_2+1}
\Bigl(z';-p,\vec{p}_{n_1+1,n};\phi''_{s_2}\Bigr)
\times p_B\left(\frac{1}{2\Lambda^2};u,z'\right)\ .
 \end{multline}
Applying the induction hypothesis to $ \mathcal{S}_{l_1,n_1+1}^{\Lambda,\Lambda_0}$ and using the bound (\ref{c1b}) for $\mathcal{D}_{l_2,n_2+1}^{\Lambda,\Lambda_0}$, we obtain that $\tilde{R}_2^{DS}$ is bounded by 
    \begin{multline*}
    \left(\Lambda+m\right)^{2-n-|w|}\mathcal{P}\left(\log\frac{\Lambda+m}{m}\right) \mathcal{P}\left(\frac{\left \|\vec{p}_n\right\|}{\Lambda+m}\right)\mathcal{Q}\left(\frac{\tau^{-\frac{1}{2}}}{\Lambda+m}\right)\\\times \int_{\mathbb{R}^+}du~\mathcal{F}_{s_1+1,l_1;\delta_3}^{0}\left(\Lambda;\tau_{1,s_1},\frac{1}{2\Lambda^2};Y_{\sigma_{s_1}},u\right)\\\times\int_{z'}\mathcal{F}_{s_2,l_2;\delta'_3}\left(\Lambda;\tau_{s_1+1,s};z';Y_{\sigma_{s_1+1:s}}\right)p_B\left(\frac{1}{2\Lambda^2};z',u\right).
    \end{multline*}
    Since the global weight factor $\mathcal{F}_{s_2,l_2;\delta'_3}\left(\Lambda;\tau_{s_1+1,s};z';Y_{\sigma_{s_1+1:s}}\right)$ is a sum of the weight factors of all trees $T^{s_2}_{l_2}(z';\tau_{s_1+1,s};Y_{\sigma_{s_1+1:s}})$ in $\mathcal{T}^{s_2}_{l_2}$, we deduce that integrating over $z'$ gives the global weight factor of the bulk trees $\hat{\mathcal{T}}^{s_2+1}_{l_2}$, and therefore we can write
    \begin{multline}\label{BulkWeight}
       \hat{\mathcal{F}}_{s_2+1,l_2;\delta'_3}\left(\Lambda;\tau_{s_1+1,s},\frac{1}{2\Lambda^2};Y_{\sigma_{s_1+1:s}},u\right)\\ =\sum_{\hat{T}^{s_2+1}_{l_2} \in \hat{\mathcal{T}}^{s_2+1}_{l_2}}\mathcal{F}_{\delta'_3}
\left(\Lambda;\tau_{s_1+1,s},\frac{1}{2\Lambda^2};\hat{T}^{s_2+1}_{l_2};Y_{\sigma_{s_1+1:s}},u\right)\ ,
    \end{multline}
where 
\begin{multline}
\mathcal{F}_{\delta'_3}
\left(\Lambda;\tau_{s_1+1,s},\frac{1}{2\Lambda^2};\hat{T}^{s_2+1}_{l_2};Y_{\sigma_{s_1+1:s}},u\right)\\:=\int_{z'}~\mathcal{F}_{\delta'_3}
\left(\Lambda;\tau_{s_1+1,s};{T}^{s_2}_{l_2};z';Y_{\sigma_{s_1+1:s}}\right)p_B\left(\frac{1}{2\Lambda^2};z',u\right).
\end{multline}
Applying Lemma \ref{TFfusion}, we deduce that $\tilde{R}_2^{DS}$  is bounded by 
 \begin{equation*}
    \left(\Lambda+m\right)^{2-n-|w|}\mathcal{P}\left(\log\frac{\Lambda+m}{m}\right) \mathcal{P}\left(\frac{\left \|\vec{p}_n\right\|}{\Lambda+m}\right)\mathcal{Q}\left(\frac{\tau^{-\frac{1}{2}}}{\Lambda+m}\right)\mathcal{F}_{s,l;\delta''_{3}}^{\Lambda,0}\left({\tau}_{1,s}\right)~,
    \end{equation*}
    where $\delta''_{3}:=\max\left(\delta_3,\delta'_3\right)$.
 \item In this part we bound the term $R_2^{SS}$. As for ${R}_2^{DS}$, we only treat the term 
 \begin{multline*}
    \tilde{R}_2^{SS}:=\int_{\mathbb{R}^+}du~\int_{z_{1,n}}\int_{z,z'}\partial^{w_1}\mathcal{S}_{l_1,n_1+1}^{\Lambda,\Lambda_0}((z_1,p_1),\cdots,(z_{n_1},p_{n_1}),(z,p))\partial^{w_3}\dot{C}^{\Lambda}(p)\\ 
\times\partial^{w_2}\mathcal{S}_{l_2,n_2+1}^{\Lambda,\Lambda_0}((z',-p),\cdots,(z_{n},p_{n}))\prod_{i=1}^s p_B(\tau_i;z_i,y_i)p_B\left(\frac{1}{2\Lambda^2};z,u\right)p_B\left(\frac{1}{2\Lambda^2};z',u\right).
\end{multline*}
Using the same notations (\ref{o1})-(\ref{o2}), we rewrite $\,\tilde{R}_2^{SS}$ as follows, 
 \begin{multline}
      \tilde{R}_2^{SS}=\int_{\mathbb{R}^+}du~  \partial^{w_1}
\mathcal{S}^{\Lambda,\Lambda_0}_{l_1,n_1+1}
\Bigl(p_1,\cdots,p_{n_1},p;\phi'_{s_1}\times p_B\left(\frac{1}{2\Lambda^2};.,u\right)
\Bigr)
\, \partial^{w_3}\dot{C}^{\Lambda}(p)\\
     \times \partial^{w_2}\mathcal{S}^{\Lambda,\Lambda_0}_{l_2,n_2+1}
\Bigl(-p,p_{n_1+1},
\cdots,p_n;\phi''_{s_2}\times p_B\left(\frac{1}{2\Lambda^2};.,u\right)\Bigr)\ .
 \end{multline}
 Using the induction hypothesis, we obtain 
 \begin{multline*}
    \left|\tilde{R}_2^{SS}\right|\leq\left(\Lambda+m\right)^{1-n-|w|}\mathcal{P}\left(\log\frac{\Lambda+m}{m}\right) \mathcal{P}\left(\frac{\left \|\vec{p}_n\right\|}{\Lambda+m}\right)\mathcal{Q}\left(\frac{\tau^{-\frac{1}{2}}}{\Lambda+m}\right)\\\times \int_{\mathbb{R}}du~\mathcal{F}_{s_1+1,l_1;\delta_3}^{0}\left(\Lambda;\tau_{1,s_1},\frac{1}{2\Lambda^2};Y_{\sigma_{s_1}},u\right)~\mathcal{F}_{s_2+1,l_2;\delta_4}^{0}\left(\Lambda;\tau_{s_1+1,s},\frac{1}{2\Lambda^2};Y_{\sigma_{s_1+1:s}},u\right).
    \end{multline*}
Applying Lemma \ref{FFfusion}, we deduce that $\tilde{R}_2^{SS}$ is bounded by 
 \begin{equation*}
    \left(\Lambda+m\right)^{2-n-|w|}\mathcal{P}\left(\log\frac{\Lambda+m}{m}\right) \mathcal{P}\left(\frac{\left \|\vec{p}_n\right\|}{\Lambda+m}\right)\mathcal{Q}\left(\frac{\tau^{-\frac{1}{2}}}{\Lambda+m}\right)\mathcal{F}_{s,l;\delta_{5}}^{\Lambda,0}
\left({\tau}_{1,s}\right),
\end{equation*}
    where $\delta_{5}:=\max \left(\delta_{3},\delta_4\right)$.
\item In this part, we bound the term $R_2^{DD}$, which we rewrite using (\ref{o1})-(\ref{o2}) as follows,
\begin{multline*}
\int_{z,z'}\partial^{w_1}\mathcal{D}_{l_1,n_1+1}^{\Lambda,\Lambda_0}\left(z;\Vec{p}_{1,n_1},p;\phi'_{s_1}\right)\partial^{w_3}\dot{C}^{\Lambda}(p)p_{S,R}\left(\frac{1}{\Lambda^2};z,-z'\right)\\ 
\times\partial^{w_2}\mathcal{D}_{l_2,n_2+1}^{\Lambda,\Lambda_0}\left(z';\Vec{p}_{n_1+1,n},-p;\phi''_{s_2}\right).
\end{multline*}
Using the bounds (\ref{SurfHeatB}) and (\ref{c1b}), we obtain that $R_2^{DD}$ is bounded by 
\begin{multline}\label{R2DD}
   \left(\Lambda+m\right)^{3-n-|w|}e^{-\frac{m^2}{2\Lambda^2}}\mathcal{P}\left(\log\frac{\Lambda+m}{m}\right) \mathcal{P}\left(\frac{\left \|\vec{p}_n\right\|}{\Lambda+m}\right)\mathcal{Q}\left(\frac{\tau^{-\frac{1}{2}}}{\Lambda+m}\right)\\\times \int_{z,z'}\mathcal{F}_{s_1,l_1;\delta'_4}\left(\Lambda;\tau_{1,s_1};z;Y_{\sigma_{s_1}}\right)\mathcal{F}_{s_2,l_2;\delta'_3}\left(\Lambda;\tau_{s_1+1,s};z';Y_{\sigma_{s_1+1:s}}\right)p_B\left(\frac{1}{\Lambda^2};z',-z\right).
\end{multline}
The bound 
\begin{equation}\label{nice}
    p_B\left(\frac{1}{\Lambda^2};z',-z\right)\leq \sqrt{2\pi}~\Lambda^{-1}~p_B\left(\frac{1}{\Lambda^2};z',0\right)~p_B\left(\frac{1}{\Lambda^2};z,0\right)
\end{equation}
together with
\begin{equation}\label{emlam}
    \Lambda^{-\alpha}e^{-\frac{m^2}{2\Lambda^2}}\leq O(1)\left(\Lambda+m\right)^{-\alpha}~~~~\mathrm{for~} \alpha \in\mathbb{N}~,
\end{equation}
gives
\begin{multline*}
   e^{-\frac{m^2}{2\Lambda^2}} \int_{z,z'}\mathcal{F}_{s_1,l_1;\delta'_4}\left(\Lambda;\tau_{1,s_1};z;Y_{\sigma_{s_1}}\right)\mathcal{F}_{s_2,l_2;\delta'_3}\left(\Lambda;\tau_{s_1+1,s};z';Y_{\sigma_{s_1+1:s}}\right)p_B\left(\frac{1}{\Lambda^2};z',-z\right)\\\leq O(1)~(\Lambda+m)^{-1}~\int_{z}\mathcal{F}_{s_1,l_1;\delta'_4}\left(\Lambda;\tau_{1,s_1};z;Y_{\sigma_{s_1}}\right)p_B\left(\frac{1}{\Lambda^2};z,0\right)\\\times\int_{z'}\mathcal{F}_{s_2,l_2;\delta'_3}\left(\Lambda;\tau_{s_1+1,s};z';Y_{\sigma_{s_1+1:s}}\right)p_B\left(\frac{1}{\Lambda^2};z',0\right).
\end{multline*}
The definition of the global surface weight factor (\ref{treeStr})-(\ref{treeStr3}) implies
\begin{multline}\label{1366}
\int_{z}\mathcal{F}_{s_1,l_1;\delta'_4}\left(\Lambda;\tau_{1,s_1};z;Y_{\sigma_{s_1}}\right)p_B\left(\frac{1}{\Lambda^2};z,0\right)\\\times\int_{z'}\mathcal{F}_{s_2,l_2;\delta'_3}\left(\Lambda;\tau_{s_1+1,s};z';Y_{\sigma_{s_1+1:s}}\right)p_B\left(\frac{1}{\Lambda^2};z',0\right)\\\leq \left(\sum_{W^s_l\in\mathcal{W}^s_l}\sum_{\left\{\Pi\in\mathcal{P}_s,~l_{\Pi}=2\right\}}\mathcal{F}_{\delta_6}^{0}\left(\Lambda;\tau_{1,s};W^s_l(\Pi);Y_{\sigma_s}\right)\right),
\end{multline}
where $\Pi:=\Pi_1\cup\Pi_2$,  $\Pi_1:=\left\{1,\cdots,s_1\right\}$ and $\Pi_2:=\left\{s_1+1,\cdots,s\right\}$. In (\ref{1366}), we also used the bound (\ref{deltaWeight}) and $\delta_{6}:=\max\left(\delta'_3,\delta'_4\right)$.
Therefore, we find that ${R}_2^{DD}$ is bounded by
\begin{multline}\label{1377}
   \left(\Lambda+m\right)^{2-n-|w|}\mathcal{P}\left(\log\frac{\Lambda+m}{m}\right) \mathcal{P}\left(\frac{\left \|\vec{p}_n\right\|}{\Lambda+m}\right)\mathcal{Q}\left(\frac{\tau^{-\frac{1}{2}}}{\Lambda+m}\right)\\\times \left(\sum_{W^s_l\in\mathcal{W}^s_l}\sum_{\left\{\Pi\in\mathcal{P}_s,~l_{\Pi}=2\right\}}\mathcal{F}_{\delta_{6}}^{0}\left(\Lambda;\tau_{1,s};W^s_l(\Pi);Y_{\sigma_s}\right)\right).
\end{multline}
This shows that $R_2^{DD}$ is bounded by 
\begin{equation*}
    \left(\Lambda+m\right)^{2-n-|w|}\mathcal{P}\left(\log\frac{\Lambda+m}{m}\right) \mathcal{P}\left(\frac{\left \|\vec{p}_n\right\|}{\Lambda+m}\right)\mathcal{Q}\left(\frac{\tau^{-\frac{1}{2}}}{\Lambda+m}\right)\mathcal{F}_{s,l;\delta_{6}}^{\Lambda,0}\left({\tau}_{1,s}\right).
    \end{equation*}
\end{itemize}
\item \underline{Case $(r_1, r_2) \neq (0,0)$:} The linear terms $R_1^S$ and $R_1^D$ together with the non-linear term $R_2^{SS}$ are treated following the same steps as before. The only terms that require a careful analysis are $R_2^{DS}$ and $R_2^{DD}$. To shorten the discussion, we analyze the term $R_2^{DD}$ only,  $R_2^{DS}$ may be treated using similar arguments. We write $$z_1^{r_1}=\sum\limits_{\alpha_1+\beta_1=r_1}^{}{{r_1}\choose{\alpha_1}} \left(z_1-z\right)^{\alpha_1}z^{\beta_1},~~~~z_2^{r_2}=\sum\limits_{\alpha_2+\beta_2=r_2}^{}{{r_2}\choose{\alpha_2}}\left(z_2-z'\right)^{\alpha_2}{z'}^{\beta_2}.$$ 
This allows to rewrite $R_2^{DD}$ for all $n_1\geq 2$ as follows 
\begin{multline}\label{r1r2}
       \sum\limits_{\alpha_1+\beta_1=r_1, \alpha_2+\beta_2=r_2} {{r_1}\choose{\alpha_1}}~{{r_2}\choose{\alpha_2}}~z^{\beta_1+\beta_2}~\partial^{w_1}
\mathcal{D}^{\Lambda,\Lambda_0;(1,2)}_{l_1,n_1+1;\alpha_1,\alpha_2}
\Bigl(z;\vec{p}_{1,n_1},p;\phi'_{s_1}
\Bigr)
\, \partial^{w_3}\dot{C}^{\Lambda}(p)\\
     \times \partial^{w_2}\mathcal{D}^{\Lambda,\Lambda_0}_{l_2,n_2+1}
\Bigl(z';\vec{p}_{n_1+1,n},-p
\cdots,p_n;\phi''_{s_2}\Bigr)
p_{S,R}\left(\frac{1}{\Lambda^2};z,z'\right)\ ,
 \end{multline}
and for $n_1=1$ we have 
 \begin{multline}\label{r1r2n}
       \sum\limits_{\alpha_1+\beta_1=r_1, \alpha_2+\beta_2=r_2} {{r_1}\choose{\alpha_1}}~{{r_2}\choose{\alpha_2}}~{z}^{\beta_1}~{z'}^{\beta_2}\partial^{w_1}
\mathcal{D}^{\Lambda,\Lambda_0;(1)}_{l_1,2;\alpha_1}
\Bigl(z;{p}_{1},p;\phi'_{s_1}
\Bigr)
\, \partial^{w_3}\dot{C}^{\Lambda}(p)\\
     \times \partial^{w_2}\mathcal{D}^{\Lambda,\Lambda_0;(2)}_{l_2,n;\alpha_2}
\Bigl(z';\vec{p}_{2,n},-p;\phi''_{s_2}\Bigr)
p_{S,R}\left(\frac{1}{\Lambda^2};z,z'\right)\ .
 \end{multline}
Using the bounds (\ref{c1b})-(\ref{c1bb}), we deduce that the summands in (\ref{r1r2})-(\ref{r1r2n}) are bounded by 
\begin{multline}\label{inn1}
\left(\Lambda+m\right)^{3-n-|w|-\alpha_1-\alpha_2}e^{-\frac{m^2}{2\Lambda^2}}\mathcal{P}\left(\log\frac{\Lambda+m}{m}\right) \mathcal{P}\left(\frac{\left \|\vec{p}_n\right\|}{\Lambda+m}\right)\mathcal{Q}\left(\frac{\tau^{-\frac{1}{2}}}{\Lambda+m}\right)\\\times \int_{z,z'}\mathcal{F}_{s_1,l_1;\delta''_1}\left(\Lambda;\tau_{1,s_1};z;Y_{\sigma_{s_1}}\right)\mathcal{F}_{s_2,l_2;\delta''_2}\left(\Lambda;\tau_{s_1+1,s};z';Y_{\sigma_{s_1+1:s}}\right)~\chi^{\Lambda}_{n_1}\left(z,z'\right),
\end{multline}
where 
\begin{equation*}
    \chi^{\Lambda}_{n_1}\left(z,z'\right)=\left\{\begin{array}{cc}
       z^{\beta_1}~{z'}^{\beta_2}p_B\left(\frac{1}{\Lambda^2};z',-z\right)  & \mathrm{if~} n_1=1,   \\
       z^{\beta_1+\beta_2}~p_B\left(\frac{1}{\Lambda^2};z',-z\right)  &  \mathrm{otherwise~.}
    \end{array}\right.
\end{equation*}
Using the bound (\ref{nice}) together with (\ref{in1}), we deduce for all $\delta,~\tilde{\delta}>0$ 
\begin{equation}
    \chi^{\Lambda}_{n_1}\left(z,z'\right) \leq O(1)~\Lambda^{-1-\beta_1-\beta_2}~p_B\left(\frac{1+\delta}{\Lambda^2};z,0\right)p_B\left(\frac{1+\tilde{\delta}}{\Lambda^2};z',0\right),
\end{equation}
which implies together with (\ref{emlam}) that (\ref{inn1}) is bounded by 
\begin{multline}
\left(\Lambda+m\right)^{2-n-|w|-r_1-r_2}\mathcal{P}\left(\log\frac{\Lambda+m}{m}\right) \mathcal{P}\left(\frac{\left \|\vec{p}_n\right\|}{\Lambda+m}\right)\mathcal{Q}\left(\frac{\tau^{-\frac{1}{2}}}{\Lambda+m}\right)\\\times \int_{z}\mathcal{F}_{s_1,l_1;\delta''_1}\left(\Lambda;\tau_{1,s_1};z;Y_{\sigma_{s_1}}\right)~p_B\left(\frac{1+\delta''_1}{\Lambda^2};z,0\right)\\\times\int_{z'}\mathcal{F}_{s_2,l_2;\delta''_2}\left(\Lambda;\tau_{s_1+1,s};z';Y_{\sigma_{s_1+1:s}}\right)~p_B\left(\frac{1+\delta''_2}{\Lambda^2};z',0\right),
\end{multline}
which together with the bounds (\ref{1366}) and (\ref{1377}) implies the final bound for $R_2^{DD}$ given by
\begin{multline}
    \left(\Lambda+m\right)^{2-n-|w|-r_1-r_2}\mathcal{P}\left(\log\frac{\Lambda+m}{m}\right) \mathcal{P}\left(\frac{\left \|\vec{p}_n\right\|}{\Lambda+m}\right)\mathcal{Q}\left(\frac{\tau^{-\frac{1}{2}}}{\Lambda+m}\right)\mathcal{F}_{s,l;\delta_7}^{\Lambda,0}\left({\tau}_{1,s}\right),
    \end{multline} 
where $\delta_7:=\max\left(\delta''_1,\delta''_2\right)$. \\
Using the bound (\ref{deltaWeight}), we deduce (\ref{c'}) where $\delta:=\max\left\{\delta_i,~\delta'_i,~\delta''_i,~1\leq i\leq 7\right\}$.
\end{itemize}
 \subsubsection{ Integration of the FEs}
\begin{itemize}
    \item We start by integrating the irrelevant terms for which $n+|w|+r_1+r_2\geq 4$. In this case, (\ref{c'}) is integrated  from $\Lambda$ to $\Lambda_0$ using the boundary condition (\ref{BCDTS})-(\ref{BCDTS2}) together with (\ref{lambWeight}) and we obtain
    \begin{multline}\label{ci}
    \left| \partial^w \mathcal{S}_{l,n;r_1,r_2}^{\Lambda,\Lambda_0}\left(\vec{p}_n;\phi_{\tau_{1,s},y_{1,s}}\right) \right|\\\leq \left(\Lambda+m\right)^{3-n-|w|-r_1-r_2}\mathcal{P}\left(\log\frac{\Lambda+m}{m}\right)\mathcal{P}\left(\frac{\left\|\vec{p}_n\right\|}{\Lambda+m}\right)\mathcal{Q}\left(\frac{\tau^{-\frac{1}{2}}}{\Lambda+m}\right)\mathcal{F}^{\Lambda,0}_{s,l;\delta}({\tau}_{1,s})~.
\end{multline}
\item The relevant terms for which $n+|w|+r_1+r_2\leq 3$ are written 
\begin{equation}\label{int1}
\int_0^{\infty}dz_1~\int_0^{\infty}dz_2~z_1^{r_1}z_2^{r_2}\partial_{\Lambda}\mathcal{S}_{l,2}^{\Lambda,\Lambda_0}\left((z_1,p),(z_2,-p)\right)\prod_{i=1}^s p_B(\tau_i;z_i,y_i)~,
\end{equation}
where $r_1$, $r_2$ are integers such that $r_1+r_2\leq 1$, and $0\leq s\leq 2$. We restrict our analysis to the case $s=2$, the case $s=1$ can be treated similarly and the case $s=0$ will be integrated in the sequel. For $s=2$, the relevant part is extracted from 
\begin{equation}\label{int1}
\int_{z_1,z_2}~\partial_{\Lambda}\mathcal{S}_{l,2}^{\Lambda,\Lambda_0}\left((z_1,p),(z_2,-p)\right) \phi_1(z_1)\phi_2(z_2)
\end{equation}
by performing a Taylor expansion of $\phi_1$ and $\phi_2$ around $z_i=0$ at $p=0$, where $\phi_i(z_i):=p_B(\tau_i;z_i,y_i)$, using (\ref{f0}) and (\ref{restI}). The bound (\ref{c'}) for $s=0$ and $r_1+r_2\leq 1$ gives 
\begin{multline}\label{intES}
    \left|\partial_{\Lambda}s^{\Lambda,\Lambda_0}_l\right|\leq \mathcal{P}\left(\log\frac{\Lambda+m}{m}\right),~~~~ \left|\partial_{\Lambda}e^{\Lambda,\Lambda_0}_l\right|\leq \left(\Lambda+m\right)^{-1} \mathcal{P}\left(\log\frac{\Lambda+m}{m}\right).
\end{multline}
Integrating (\ref{intES}) from $0$ to $\Lambda$ and using the renormalization conditions (\ref{renoS}), we have
\begin{equation}\label{f2}
    \left|s^{\Lambda,\Lambda_0}_l\right|\leq \left(\Lambda+m\right) \mathcal{P}\left(\log\frac{\Lambda+m}{m}\right),~~ \left|e^{\Lambda,\Lambda_0}_l\right|\leq \mathcal{P}\left(\log\frac{\Lambda+m}{m}\right).
\end{equation}
Applying Lemma \ref{lemma7} from the Appendix, we have
\begin{multline}\label{1227}
    \left|s_l^{\Lambda,\Lambda_0}\phi_1(0)\phi_2(0)+e_l^{\Lambda,\Lambda_0}\left\{\phi_1(0)(\partial_{n}\phi_2)(0)+\phi_2(0)(\partial_{n}\phi_1)(0)\right\}\right|\\\leq 
    \left(\Lambda+m\right) \mathcal{P}\left(\log\frac{\Lambda+m}{m}\right)\mathcal{Q}\left(\frac{\tau^{-\frac{1}{2}}}{\Lambda+m}\right)\mathcal{F}^{\Lambda,0}_{2,l;\delta}\left(\tau_{1,2}\right).
\end{multline}
Now, we bound and integrate the remainder $\partial_{\Lambda}l_{l,2}^{\Lambda,\Lambda_0}(\phi_1,\phi_2)$, which is irrelevant as we will see in the sequel, from $\Lambda$ to $\Lambda_0$. 
We distinguish between the two cases:
\begin{itemize}
    \item \underline{$\Lambda\leq 3 \sqrt{l}\tau^{-\frac{1}{2}}$:} Using (\ref{c'}) and (\ref{1227}), we deduce that 
 $\partial_{\Lambda}l_{l,2}^{\Lambda,\Lambda_0}\left(\phi_1,\phi_2\right)$ is bounded by  
 \begin{equation}\label{h0}
        O(1)~\max\left(m^2,\tau^{-1}\right)\left(\Lambda+m\right)^{-2}\mathcal{P}\left(\log\frac{\Lambda+m}{m}\right)\mathcal{Q}\left(\frac{\tau^{-\frac{1}{2}}}{\Lambda+m}\right)\mathcal{F}^{\Lambda,0}_{2,l;\delta}\left({\tau}_{1,2}\right).
    \end{equation}
\item \underline{$\Lambda\geq 3\sqrt{l}\tau^{-\frac{1}{2}}$:}~
In the sequel, we restrict our analysis to the integration of the following terms, for which we need to proceed differently.
\begin{align}
    \Dot{\mathcal{H}}_1&:=\left(\int_{z_1,z_2}z_1z_2\partial_{\Lambda}\mathcal{S}_{l,2}^{\Lambda,\Lambda_0}\left((z_1,0),(z_2,0)\right)\right)\left(\partial_n\phi_1\right)(0)\left(\partial_n\phi_2\right)(0)~,\label{h1}\\
    \Dot{\mathcal{H}}_2&:=(\partial_n\phi_2)(0)\int_{z_1,z_2}z_2~\partial_{\Lambda}\mathcal{S}_{l,2}^{\Lambda,\Lambda_0}\left((z_1,0),(z_2,0)\right)\int_0^1dt ~(1-t)\left(\partial_t^2\phi_1(tz_1)\right)~,
\end{align}
and 
\begin{multline}\label{h3}
    \Dot{\mathcal{H}}_3:=\int_{z_1,z_2}\partial_{\Lambda}\mathcal{S}_{l,2}^{\Lambda,\Lambda_0}\left((z_1,0),(z_2,0)\right)\left(\int_0^1dt~(1-t)\partial_t^2\phi_1(tz_1)\right)\\\times\left(\int_0^1dt'~(1-t')\partial_{t'}^2\phi_2(t'z_2)\right).
\end{multline}
The other terms which also contribute to $\partial_{\Lambda}l_{l,2}^{\Lambda,\Lambda_0}(\phi_1,\phi_2)$ can be treated similarly.
\begin{itemize}
    \item We start first with $\Dot{\mathcal{H}}_1$ for which the bound (\ref{c'}) implies that
\begin{equation}
    \left|\left(\int_{z_1,z_2}z_1z_2\partial_{\Lambda}\mathcal{S}_{l,2}^{\Lambda,\Lambda_0}\left((z_1,0),(z_2,0)\right)\right)\right|\leq \left(\Lambda+m\right)^{-2}\mathcal{P}\left(\log \frac{\Lambda+m}{m}\right).
\end{equation}
Using Lemma \ref{lemma7}, we obtain   
\begin{multline}\label{h11}
     \left|\left(\int_{z_1,z_2}z_1~z_2\partial_{\Lambda}\mathcal{S}_{l,2}^{\Lambda,\Lambda_0}\left((z_1,0),(z_2,0)\right)\right)\left(\partial_n\phi_1\right)(0)\left(\partial_n\phi_2\right)(0)\right|\\\leq \left(\Lambda+m\right)^{-2}\tau_1^{-\frac{1}{2}}\tau_2^{-\frac{1}{2}}\mathcal{P}\left(\log \frac{\Lambda+m}{m}\right)\mathcal{Q}\left(\frac{\tau^{-\frac{1}{2}}}{\Lambda+m}\right)\mathcal{F}_{2,l;\delta}^{\Lambda,0}\left({\tau}_{1,2}\right).
\end{multline}
\item The term $\Dot{\mathcal{H}}_2$: we have for $0\leq t\leq 1$
\begin{equation}\label{phit}
\phi_i(tz_i)=\frac{1}{\sqrt{2\pi\tau_i}}e^{-\frac{(tz_i-y_i)^2}{2\tau_i}}.
\end{equation}
Differentiating (\ref{phit}) twice w.r.t. $t$, we obtain 
\begin{equation}\label{phitt}
    \partial_t^2\left(\phi_i(tz_i)\right)=\frac{1}{t}\left[-\frac{z_i^2}{\tau_i}+\frac{z_i^2(tz_i-y_i)^2}{\tau_i^2}\right]p_B\left(\frac{\tau_i}{t^2};z_i,\frac{y_i}{t}\right),
\end{equation}
which implies that the term 
\begin{equation}\label{riri}
    \int_{z_1,z_2}z_2~\partial_{\Lambda}\mathcal{S}_{l,2}^{\Lambda,\Lambda_0}\left((z_1,0),(z_2,0)\right)\int_0^1dt ~(1-t)\partial_t^2\phi_1(tz_1)
\end{equation}
can be rewritten as 
\begin{multline}\label{int5}
\sum_{(\alpha,\beta)\in\mathcal{I}_2}c_{\alpha\beta}\frac{y_1^{\beta}}{\tau_1^{1+\frac{\beta+\alpha}{2}}}\int_0^1~dt~{t^{\alpha-1}(1-t)}~\\\int_{z_1,z_2}z_2~z_1^{2+\alpha}\partial_{\Lambda}\mathcal{S}_{l,2}^{\Lambda,\Lambda_0}\left((z_1,0),(z_2,0)\right)p_B\left(\frac{\tau_{1}}{t^2};z_1,\frac{y_1}{t}\right),
\end{multline}
where $\mathcal{I}_2:= \left\{(0,0),(\alpha,\beta)\left|\right.\alpha+\beta=2,~~(\alpha,\beta)\in\mathbb{N}^2\right\}$, and the coefficients $c_{\alpha\beta}\in\mathbb{R}$ depend only on the exponents $\alpha$ and $\beta$. The bound (\ref{c'}) implies that the term 
\begin{equation}\label{star}
    \int_{z_1,z_2}z_2~z_1^{2+\alpha}\partial_{\Lambda}\mathcal{S}_{l,2}^{\Lambda,\Lambda_0}\left((z_1,0),(z_2,0)\right)p_B\left(\frac{\tau_{1}}{t^2};z_1,\frac{y_1}{t}\right)
\end{equation}
is bounded by 
\begin{equation}\label{sstar}
\left(\Lambda+m\right)^{-3-\alpha}e^{-\frac{m^2}{2\Lambda^2}}\mathcal{Q}\left(\frac{t\tau_{1}^{-\frac{1}{2}}}{\Lambda+m}\right)\mathcal{F}^{0}_{1,l;\delta}\left(\Lambda;\frac{\tau_1}{t^2};\frac{y_1}{t}\right)~.
\end{equation}
From Lemma \ref{lemma8}, we obtain 
\begin{equation}
    \left(\frac{y_1}{\sqrt{\tau}_1}\right)^{\beta}\mathcal{F}^{0}_{1,l;\delta}\left(\Lambda;\frac{\tau_1}{t^2};\frac{y_1}{t}\right)\leq O(1)~t~\left(1+\frac{\sqrt{\tau_1}}{\Lambda}\right)^{\beta}\mathcal{F}_{1,l;\delta'}^{\Lambda,0}\left(\tau_{1}\right),
\end{equation}
where $0<\delta<\delta'$. Using (\ref{emlam}) together with (\ref{star}) and (\ref{sstar}), we deduce that (\ref{int5}) is bounded by 
 \begin{equation}
   \left(\Lambda+m\right)^{-3}\tau_1^{-1}\mathcal{P}\left(\log \frac{\Lambda+m}{m}\right){\mathcal{Q}}\left(\frac{\tau^{-\frac{1}{2}}}{\Lambda+m}\right)\mathcal{F}_{1,l;\delta'}^{\Lambda,0}\left({\tau}_{1}\right).
\end{equation}
Lemma \ref{lemma7} together with (\ref{borneweights2}) imply that $\Dot{\mathcal{H}}_2$ is bounded by
\begin{equation}\label{h2}
   \left(\Lambda+m\right)^{-3}\tau_1^{-1}\tau_2^{-\frac{1}{2}}\mathcal{P}\left(\log \frac{\Lambda+m}{m}\right){\mathcal{Q}}\left(\frac{\tau^{-\frac{1}{2}}}{\Lambda+m}\right)\mathcal{F}_{2,l;\delta'}^{\Lambda,0}\left({\tau}_{1,2}\right).
\end{equation}
Following similar steps, we obtain
\begin{multline}\label{h4}
    \left|(\partial_n\phi_1)(0)\int_{z_1,z_2}z_2~\partial_{\Lambda}\mathcal{S}_{l,2}^{\Lambda,\Lambda_0}\left((z_1,0),(z_2,0)\right)\int_0^1dt ~(1-t)\left(\partial_t^2\phi_2(tz_1)\right)\right|\\
    \leq  \left(\Lambda+m\right)^{-3}\tau_1^{-\frac{1}{2}}\tau_2^{-1}\mathcal{P}\left(\log \frac{\Lambda+m}{m}\right){\mathcal{Q}}\left(\frac{\tau^{-\frac{1}{2}}}{\Lambda+m}\right)\mathcal{F}_{2,l;\delta'}^{\Lambda,0}\left({\tau}_{1,2}\right)
\end{multline}
and 
\begin{multline}\label{h5}
    \left|\phi_i(0)\int_{z_1,z_2}~\partial_{\Lambda}\mathcal{S}_{l,2}^{\Lambda,\Lambda_0}\left((z_1,0),(z_2,0)\right)\int_0^1dt ~(1-t)\left(\partial_t^2\phi_j(tz_1)\right)\right|\\
    \leq  \left(\Lambda+m\right)^{-2}\tau_j^{-1}\mathcal{P}\left(\log \frac{\Lambda+m}{m}\right){\mathcal{Q}}\left(\frac{\tau^{-\frac{1}{2}}}{\Lambda+m}\right)\mathcal{F}_{2,l;\delta'}^{\Lambda,0}\left({\tau}_{1,2}\right).
\end{multline}
\item The term $\Dot{\mathcal{H}}_3$: We start from (\ref{h3}) and (\ref{phitt}) and write the term $\Dot{\mathcal{H}}_3$ as follows
\begin{multline}\label{int6}
\sum_{(\alpha,\beta)\in\mathcal{I}_2}\sum_{(\alpha',\beta')\in\mathcal{I}_2}c_{\alpha\beta}~c_{\alpha'\beta'}~\frac{y_1^{\beta}}{\tau_1^{1+\frac{\beta+\alpha}{2}}}~\frac{y_2^{\beta'}}{\tau_2^{1+\frac{\beta'+\alpha'}{2}}}~\int_0^1~dt~dt'~{t^{\alpha-1}(1-t)}~{{t'}^{\alpha'-1}(1-t')}\\\int_{z_1,z_2}z_2^{2+\alpha'}~z_1^{2+\alpha}\partial_{\Lambda}\mathcal{S}_{l,2}^{\Lambda,\Lambda_0}\left((z_1,0),(z_2,0)\right)p_B\left(\frac{\tau_{1}}{t^2};z_1,\frac{y_1}{t}\right)p_B\left(\frac{\tau_{2}}{t'^2};z_2,\frac{y_2}{t'}\right).
\end{multline}
The bound (\ref{c'}) implies that the term 
\begin{equation*}
    \int_{z_1,z_2}z_2^{2+\alpha'}~z_1^{2+\alpha}\partial_{\Lambda}\mathcal{S}_{l,2}^{\Lambda,\Lambda_0}\left((z_1,0),(z_2,0)\right)p_B\left(\frac{\tau_{1}}{t^2};z_1,\frac{y_1}{t}\right)p_B\left(\frac{\tau_{2}}{t'^2};z_2,\frac{y_2}{t'}\right)
\end{equation*}
is bounded by 
\begin{equation*}
\left(\Lambda+m\right)^{-4-\alpha-\alpha'}e^{-\frac{m^2}{2\Lambda^2}}\mathcal{Q}\left(\frac{t\tau^{-\frac{1}{2}}}{\Lambda+m}\right)\mathcal{F}^{0}_{2,l;\delta}\left(\Lambda;\frac{\tau_1}{t^2},\frac{\tau_2}{t'^2};\frac{y_1}{t},\frac{y_2}{t'}\right).
\end{equation*}
From Lemma \ref{lemma9} and (\ref{emlam}), we deduce that (\ref{int6}) is bounded by 
\begin{equation}\label{h33}
  \tau_{1}^{-\frac{1}{2}}\tau_2^{-\frac{1}{2}}\left(\Lambda+m\right)^{-2} \mathcal{P}\left(\log \frac{\Lambda+m}{m}\right)\tilde{\mathcal{Q}}\left(\frac{\tau^{-\frac{1}{2}}}{\Lambda+m}\right)\mathcal{F}_{2,l;\delta'}^{\Lambda;0}\left({\tau}_{1,2}\right),
\end{equation}
where $0<\delta<\delta'<1$.
\end{itemize}
The boundary conditions (\ref{BCDTS}) together with (\ref{restI}), (\ref{int5}) and (\ref{int6}) imply that $$l_{l,2;R}^{\Lambda_0,\Lambda_0}\left(\phi_1,\phi_2\right)=0.$$ Integrating from $\Lambda$ to $\Lambda_0$, we obtain for $\Lambda_0\geq 3\sqrt{l}{\tau}^{-\frac{1}{2}}$
\begin{equation}
     l_{l,2}^{\Lambda,\Lambda_0}\left(\phi_1,\phi_2\right)=\int_{\Lambda}^{3\sqrt{l}\tau^{-\frac{1}{2}}}d\lambda~\partial_{\lambda}l_{l,2}^{\lambda,\Lambda_0}\left(\phi_1,\phi_2\right)+\int_{3\sqrt{l}\tau^{-\frac{1}{2}}}^{\Lambda_0}d\lambda~\partial_{\lambda}l_{l,2}^{\lambda,\Lambda_0}\left(\phi_1,\phi_2\right).
\end{equation}
Using the bounds (\ref{h0}), (\ref{h11}), (\ref{h2}) and (\ref{h33}) together with (\ref{h4}) and (\ref{h5}), we obtain that the remainder $l_{l,2}^{\Lambda,\Lambda_0}$ is bounded by 
\begin{equation}
   \left(\Lambda+m\right) \max\left\{\frac{\tau^{-1}}{\left(\Lambda+m\right)^2},\frac{m^2}{\left(\Lambda+m\right)^2}\right\}\mathcal{P}\left(\log \frac{\Lambda+m}{m}\right)\tilde{\mathcal{Q}}\left(\frac{\tau^{-\frac{1}{2}}}{\Lambda+m}\right)\mathcal{F}_{2,l;\delta'}^{\Lambda,0}\left(\tau_{1,2}\right).
\end{equation}
For $\Lambda_0\leq 3\sqrt{l}{\tau}^{-\frac{1}{2}}$, we conclude by integrating (\ref{h0}) from $\Lambda$ to $\Lambda_0$ and we deduce
\begin{multline}\label{sirop}
    \left|  \partial^{w}\mathcal{S}_{l,2;r_1,r_2}^{\Lambda,\Lambda_0}\left(0;\phi_{\tau_{1,s},y_{1,s}}\right) \right|\\\leq \left(\Lambda+m\right)^{1-|w|-r_1-r_2}\mathcal{P}\left(\log\frac{\Lambda+m}{m}\right)\mathcal{Q}\left(\frac{\tau^{-\frac{1}{2}}}{\Lambda+m}\right)\mathcal{F}^{\Lambda,0}_{s,l;{\delta'}}({\tau}_{1,s})~,
\end{multline}
where we used (\ref{deltaWeight}).
\end{itemize}
\item Extension to general momenta: We now extend the bound (\ref{sirop}) to general momenta using the Taylor formula with integral remainder, which reads
\begin{equation}
    \partial^{w}\mathcal{S}_{l,2;r_1,r_2}^{\Lambda,\Lambda_0}\left(p;\phi_{\tau_{1,s},y_{1,s}}\right)=\partial^{w}\mathcal{S}_{l,2;r_1,r_2}^{\Lambda,\Lambda_0}\left(0;\phi_{\tau_{1,s},y_{1,s}}\right)+\int_0^1 dt~ \partial_t\partial^{w}\mathcal{S}_{l,2;r_1,r_2}^{\Lambda,\Lambda_0}\left(tp;\phi_{\tau_{1,s},y_{1,s}}\right). 
\end{equation}
Applying this formula, the bound of
the integrand ( due to the derivative ) yields an additional factor $(\Lambda+m)^{-1}$ which combines with the momentum produced by the $t$-derivative to give a bound as in (\ref{c1}).
\end{itemize}
This ends the proof of Theorem \ref{theoremReno}.
\end{proof}
\subsection{Proof of Proposition \ref{Prop44}}
\noindent \begin{proof} 
First, we prove (\ref{limit}). The proof follows the same inductive scheme used in the proof of Theorem \ref{theoremReno}. For the tree order, we have 
$$\mathcal{S}_{0,n;\star}^{\Lambda,\Lambda_0}\left((z_1,p_1),\cdots,(z_n,p_n)\right)=0,~~~\forall n\geq 2,~\star\in\left\{D,R,N\right\}.$$
Clearly, the statement (\ref{limit}) holds. 
\begin{itemize}
    \item [A1)] We start by verifying inductively the following statement
    \begin{equation}\label{ter}
     \partial_{\Lambda}\mathcal{S}_{l,n;D}^{\Lambda,\Lambda_0}\left(\Vec{p}_n;\phi^D_{\tau_{1,n},y_{1,n}}\right)=\lim_{c\rightarrow +\infty}\partial_{\Lambda}\mathcal{S}_{l,n;R}^{\Lambda,\Lambda_0}\left(\Vec{p}_n;\phi^R_{\tau_{1,n},y_{1,n}}\right).
\end{equation}
In the sequel, we use the symbol $\star$ to denote either Dirichlet or Robin boundary conditions. Given $\Pi=\left(\pi_1,\pi_2\right)\in\mathcal{P}_n$ such that $|\pi_i|=n_i$ and $n_1+n_2=n$, we introduce the following notations:
\begin{multline}\label{not1}
    \mathcal{S}_{l_i,n_i+1;\star}^{\Lambda,\Lambda_0}\left(\vec{p}_{\pi_i},(-1)^ip;\Phi^{\Lambda;\star}_{\pi_i};Y_{\pi_{i}},u\right)\\:=\int_{\vec{z}_{\pi_i},z} \mathcal{S}_{l_i,n_i+1;\star}^{\Lambda,\Lambda_0}\left((\Vec{z}_{\pi_i},\Vec{p}_{\pi_i}),(z,p)\right)\prod_{i\in\pi_{i}}p_{\star}\left(\tau_i;z_i,y_{i}\right)p_{\star}\left(\frac{1}{2\Lambda^2};z,u\right),~~~i\in\left\{1,2\right\}
\end{multline}
and 
\begin{multline}\label{not2}
   \mathcal{S}_{l-1,n+2;\star}^{\Lambda,\Lambda_0}\left(\vec{p}_{n},k,-k;\Phi^{\Lambda;\star}_{n+2};Y_{\sigma_n},u,u\right):=\int_{\vec{z}_{n},z,z'} \mathcal{S}_{l-1,n+2;\star}^{\Lambda,\Lambda_0}\left((\Vec{z}_{n},\Vec{p}_{n}),(z,k),(z',-k)\right)\\\times\phi_{\tau_{1,n},y_{1,n}}^{\star}\left(z_{1,n}\right)p_{\star}\left(\frac{1}{2\Lambda^2};z,u\right)p_{\star}\left(\frac{1}{2\Lambda^2};z',u\right),
\end{multline}
where 
\begin{equation*}
    \Phi^{\Lambda;\star}_{\pi_i}\left(z_{\pi_i},z\right)=\phi^{\star}_{\pi_i}\left(z_{\pi_i}\right)p_{\star}\left(\frac{1}{2\Lambda^2};z,u\right)~~~~\mathrm{with~~~}\phi^{\star}_{\pi_i}\left(z_{\pi_i}\right):=\prod_{i\in\pi_{i}}p_{\star}\left(\tau_i;z_i,y_{i}\right)~,
\end{equation*}
and 
\begin{equation*}
    \Phi^{\Lambda;\star}_{n+2}\left(z_{1,n},z,z'\right)=\phi_{\tau_{1,n},y_{1,n}}^{\star}\left(z_{1,n}\right)p_{\star}\left(\frac{1}{2\Lambda^2};z,u\right)p_{\star}\left(\frac{1}{2\Lambda^2};z',u\right).
\end{equation*}
We consider the flows equations (\ref{FEB}) folded with the test functions $\phi_{\tau_{1,n},y_{1,n}}^{\star}$ given by
\begin{align}\label{FES}
 \partial_{\Lambda}&\mathcal{S}_{l,n;\star}^{\Lambda,\Lambda_0}\left(\Vec{p}_n;\phi_{\tau_{1,n},y_{1,n}}^{\star}\right)\nonumber\\
&=\frac{1}{2}\int_{u}\int_k 
\mathcal{S}_{l-1,n+2;\star}^{\Lambda,\Lambda_0}
\left(\Vec{p}_n,k,-k;\Phi_{n+2}^{\Lambda;\star};Y_{\sigma_n},u,u\right)
\dot{C}^{\Lambda}(k)\nonumber\\
&+\frac{1}{2}\int_{z,z'}~\int_k 
\mathcal{D}_{l-1,n+2}^{\Lambda,\Lambda_0}
\left((\Vec{z}_n.\Vec{p}_n),(z,k),(z',-k)\right)\Dot{C}^{\Lambda}_{S,\star}(k;z,z')\phi_{\tau_{1,n},y_{1,n}}^{\star}\left(z_{1,n}\right)\nonumber\\
   ~ &-\frac{1}{2} \sum_{l_1,l_2}'\sum_{\pi_1,\pi_2}''\left[\int_{u}\left\{\mathcal{S}_{l_1,n_1+1;\star}^{\Lambda,\Lambda_0}\left(\Vec{p}_{\pi_1},-p;\Phi_{\pi_1}^{\Lambda;\star};Y_{\pi_{1}},u\right)\dot{C}^{\Lambda}(p)~ \right.\right.\nonumber\\
&~~~~~~~\times\mathcal{S}_{l_2,n_2+1;\star}^{\Lambda,\Lambda_0}\left(\Vec{p}_{\pi_2},p;\Phi_{\pi_2}^{\Lambda;\star};Y_{\pi_{2}},u\right)\nonumber\\
&+\mathcal{D}_{l_1,n_1+1}^{\Lambda,\Lambda_0}\left(\Vec{p}_{\pi_1},-p;\Phi_{\pi_1}^{\Lambda;\star};Y_{\pi_{1}},u\right)~\dot{C}^{\Lambda}(p)~ 
\mathcal{S}_{l_2,n_2+1;\star}^{\Lambda,\Lambda_0}\left(\Vec{p}_{\pi_2},p;\Phi_{\pi_2}^{\Lambda;\star};Y_{\pi_{2}},u\right)\nonumber\\
&\left.+\mathcal{S}_{l_1,n_1+1;\star}^{\Lambda,\Lambda_0}\left(\Vec{p}_{\pi_1},-p;\Phi_{\pi_1}^{\Lambda;\star};Y_{\pi_{1}},u\right)~\dot{C}^{\Lambda}(p)~
\mathcal{D}_{l_2,n_2+1}^{\Lambda,\Lambda_0}\left(\Vec{p}_{\pi_2},p;\Phi_{\pi_2}^{\Lambda;\star};Y_{\pi_{2}},u\right)\right\}\nonumber\\
&\left.+\int_{z,z'}~\mathcal{D}_{l_1,n_1+1}^{\Lambda,\Lambda_0}\left(z;-p,\Vec{p}_{1,n_1};\phi_{\pi_1}^{\star}\right)\dot{C}^{\Lambda}_{S,\star}(p;z,z')
\mathcal{D}_{l_2,n_2+1}^{\Lambda,\Lambda_0}\left(z';p,\Vec{p}_{n_1+1,n};\phi_{\pi_2}^{\star}\right)\right]_{rsym},\nonumber\\
   &~~~~~~ p=-\sum_{i\in\pi_1}p_i=\sum_{i\in\pi_2}p_{i}\ ,
\end{align}
where we used (\ref{semigroup}) and the notations (\ref{not1})-(\ref{not2}).
Here, the prime denotes all pairs $(l_1,l_2)$ such that $l_1+l_2=l$, and the double prime refers to a summation over $(\pi_1,\pi_2)\in\tilde{\mathcal{P}}_{2;n}$ with $n_i:=|\pi_i|$.\\
Using the induction hypothesis, we obtain 
\begin{equation}\label{lo}
\mathcal{S}_{l-1,n+2;D}^{\Lambda,\Lambda_0}
\left(\Vec{p}_n,k,-k;\Phi_{n+2}^{\Lambda;D};Y_{\sigma_n},u,u\right)
=\lim_{c\rightarrow +\infty}\mathcal{S}_{l-1,n+2;R}^{\Lambda,\Lambda_0}
\left(\Vec{p}_n,k,-k;\Phi_{n+2}^{\Lambda;R};Y_{\sigma_n},u,u\right)~,
\end{equation}
\begin{multline}\label{l2}
    \mathcal{S}_{l_i,n_i+1;D}^{\Lambda,\Lambda_0}\left(\vec{p}_{\pi_i},(-1)^ip;\Phi^{\Lambda;D}_{\pi_i};Y_{\pi_{i}},u\right)\\=\lim_{c\rightarrow +\infty} \mathcal{S}_{l_i,n_i+1;R}^{\Lambda,\Lambda_0}\left(\vec{p}_{\pi_i},(-1)^ip;\Phi^{\Lambda;R}_{\pi_i};Y_{\pi_{i}},u\right),~~~i\in\left\{1,2\right\}.
\end{multline}
For $\tau_i>0$ and $y_i\in\mathbb{R}^+$, we have 
\begin{equation}\label{DRc}
\lim_{c\rightarrow +\infty}\mathcal{N}_p\left(p_D(\tau_i;\cdot,y_i)-p_R(\tau_i;\cdot,y_i)\right)=0~,
\end{equation}
where for $\phi$ in $\mathcal{S}(\mathbb{R}^+)$ the semi-norm $\mathcal{N}_p$ is given by
$$\mathcal{N}_p(\phi)=\sum_{\alpha,~\beta \leq p} ~\sup_{x\in\mathbb{R}^{+}}|x^{\alpha}\partial^{\beta}\phi(x)|~.$$ 
Remembering that $\mathcal{D}_{l,n}^{\Lambda,\Lambda_0}\in\mathcal{S}'\left(\mathbb{R}^{+n}\right)$ and using (\ref{DRc}), we deduce
\begin{equation}\label{l4}
   \mathcal{D}_{l_i,n_i+1}^{\Lambda,\Lambda_0}\left(\Vec{p}_{\pi_i},(-1)^{i}p;\Phi_{\pi_i}^{\Lambda;D};Y_{\pi_i},u\right)=\lim_{c\rightarrow +\infty}\mathcal{D}_{l_i,n_i+1}^{\Lambda,\Lambda_0}\left(\Vec{p}_{\pi_i},(-1)^{i}p;\Phi_{\pi_i}^{\Lambda;R};Y_{\pi_i},u\right).
\end{equation}
We rewrite the term
\begin{equation}
   \int_k \int_{z,z'}~ 
\mathcal{D}_{l-1,n+2}^{\Lambda,\Lambda_0}
\left((\Vec{z}_n.\Vec{p}_n),(z,k),(z',-k)\right)\Dot{C}^{\Lambda}_{S,R}(k;z,z')\phi_{\tau_{1,n},y_{1,n}}^{R}\left(z_{1,n}\right)
\end{equation}
as follows 
\begin{multline}
   \int_{k} \int_u \mathcal{D}_{l-1,n+2}^{\Lambda,\Lambda_0}
\left(\Vec{p}_n,k,-k;\Phi_{n+2}^{\Lambda;R};Y_{\sigma_n},u,u\right)\Dot{C}^{\Lambda}(k)\\
-\int_{k}\int_{\mathbb{R}}du~ \mathcal{D}_{l-1,n+2}^{\Lambda,\Lambda_0}
\left(\Vec{p}_n,k,-k;\phi_{\tau_{1,n},y_{1,n};R}\times p_B\left(\frac{1}{2\Lambda^2};\cdot,u\right)p_B\left(\frac{1}{2\Lambda^2};\cdot,u\right)\right)\Dot{C}^{\Lambda}(k)~.
\end{multline}
Following the same steps that led to (\ref{l4}), we obtain
\begin{equation}\label{l51}
 \mathcal{D}_{l-1,n+2}^{\Lambda,\Lambda_0}
\left(\Vec{p}_n,k,-k;\Phi_{n+2}^{\Lambda;D};Y_{\sigma_n},u,u\right)
= \lim_{c\rightarrow +\infty} \mathcal{D}_{l-1,n+2}^{\Lambda,\Lambda_0}
\left(\Vec{p}_n,k,-k;\Phi_{n+2}^{\Lambda;R};Y_{\sigma_n},u,u\right)
\end{equation}
and 
\begin{multline}\label{l50}
   \mathcal{D}_{l-1,n+2}^{\Lambda,\Lambda_0}
\left(\Vec{p}_n,k,-k;\phi_{\tau_{1,n},y_{1,n};D}\times p_B\left(\frac{1}{2\Lambda^2};\cdot,u\right)p_B\left(\frac{1}{2\Lambda^2};\cdot,u\right)\right)\\
= \lim_{c\rightarrow +\infty}\mathcal{D}_{l-1,n+2}^{\Lambda,\Lambda_0}
\left(\Vec{p}_n,k,-k;\phi_{\tau_{1,n},y_{1,n};R}\times p_B\left(\frac{1}{2\Lambda^2};\cdot,u\right)p_B\left(\frac{1}{2\Lambda^2};\cdot,u\right)\right).
\end{multline}
Part A1) in the proof of Theorem \ref{theoremReno} implies that the integrands of each term on the RHS of the FEs (\ref{FES}), in the case of Robin boundary conditions are bounded independently of the Robin parameter $c$, and the Lemmas \ref{reduction}-\ref{TFfusion} show that these bounds are integrable w.r.t. $u$. We refer the reader to the proof of Theorem \ref{theoremReno} for more details.
\item [A2)] \underline{Integration:} 
Lebesgue's dominated convergence theorem together with (\ref{lo})-(\ref{l50}) and the FES (\ref{FES}) gives 
\begin{equation}
   \partial_{\Lambda}\mathcal{S}_{l,n;D}^{\Lambda,\Lambda_0}\left(\Vec{p}_n;\phi_{\tau_{1,n},y_{1,n}}^D\right) =\lim_{c\rightarrow +\infty}\partial_{\Lambda}\mathcal{S}_{l,n;R}^{\Lambda,\Lambda_0}\left(\Vec{p}_n;\phi_{\tau_{1,n},y_{1,n}}^R\right)
    .
\end{equation}
This implies (again by the Lebesgue's dominated convergence theorem and the integrability of the bound (\ref{c'}) in the proof of Theorem \ref{theoremReno})
\begin{equation}\label{free}
    \int_{\Lambda}^{\Lambda_0}d\lambda~\partial_{\lambda}\mathcal{S}_{l,n;D}^{\lambda,\Lambda_0}\left(\Vec{p}_n;\phi_{\tau_{1,n},y_{1,n}}^D\right)=\lim_{c\rightarrow +\infty}\int_{\Lambda}^{\Lambda_0}d\lambda~\partial_{\lambda}\mathcal{S}_{l,n;R}^{\lambda,\Lambda_0}\left(\Vec{p}_n;\phi_{\tau_{1,n},y_{1,n}}^R\right).
\end{equation}
\begin{itemize}
    \item \underline{Irrelevant terms:} These terms are characterized by $n\geq 4$. Using the boundary condition
    \begin{equation}
\mathcal{S}_{l,n;\star}^{\Lambda_0\Lambda_0}\left(\Vec{p}_n;\phi^{\star}_{\tau_{1,n},y_{1,n}}\right)=0,~~~\star\in\left\{D,R\right\}
    \end{equation}
    together with (\ref{free}), we deduce 
    \begin{equation}\label{ter}
   \mathcal{S}_{l,n;D}^{\Lambda,\Lambda_0}\left(\Vec{p}_n;\phi^D_{\tau_{1,n},y_{1,n}}\right)=\lim_{c\rightarrow +\infty}~ 
   \mathcal{S}_{l,n;R}^{\Lambda,\Lambda_0}\left(\Vec{p}_n;\phi^R_{\tau_{1,n},y_{1,n}}\right)
   .
\end{equation}
\item \underline{Relevant terms ($n=2$):} We have 
\begin{multline}
    \int_{\Lambda}^{\Lambda_0}d\lambda~\partial_{\lambda} \mathcal{S}^{\lambda,\Lambda_0}_{l,2;R}\left(\Vec{p}_n;\phi^{R}_{\tau_{1,n},y_{1,n}}\right)\\=\mathcal{S}^{\Lambda_0,\Lambda_0}_{l,2;R}\left(\Vec{p}_n;\phi^{R}_{\tau_{1,n},y_{1,n}}\right)-\mathcal{S}^{\Lambda,\Lambda_0}_{l,2;R}\left(\Vec{p}_n;\phi^{R}_{\tau_{1,n},y_{1,n}}\right)
\end{multline}
and 
\begin{equation}\label{suite}
    \int_{\Lambda}^{\Lambda_0}d\lambda~\partial_{\lambda} \mathcal{S}^{\lambda,\Lambda_0}_{l,2;D}\left(\Vec{p}_n;\phi^{R}_{\tau_{1,n},y_{1,n}}\right)=-\mathcal{S}^{\Lambda,\Lambda_0}_{l,2;D}\left(\Vec{p}_n;\phi^{D}_{\tau_{1,n},y_{1,n}}\right).
\end{equation}
In (\ref{suite}), we used the boundary condition (\ref{BCDDS}) for the Dirichlet case. The boundary condition (\ref{BCDTS}) implies 
\begin{multline}\label{tui}
    \mathcal{S}^{\Lambda_0,\Lambda_0}_{l,2;R}\left(\Vec{p}_n;\phi^{R}_{\tau_{1,2},y_{1,2}}\right)=s_{l;R}^{\Lambda_0,\Lambda_0}\prod_{i=1}^2~p_R\left(\tau_i;y_i,0\right)\\+e_{l;R}^{\Lambda_0,\Lambda_0}\left\{p_R\left(\tau_1;y_1,0\right)~\partial_n p_R\left(\tau_2;y_2,0\right)+p_R\left(\tau_2;y_2,0\right)~\partial_n p_R\left(\tau_1;y_1,0\right)\right\}.
\end{multline}
Using $$\left |\partial_n p_R\left(\tau_1;y_1,0\right) \right|\leq O(1)~\tau_i^{-\frac{1}{2}}~p_B\left(\tau_{i,\delta};y_i,0\right),$$
and the fact that $s_{l;R}^{\Lambda_0,\Lambda_0}$ and $e_{l;R}^{\Lambda_0,\Lambda_0}$ are uniformly bounded w.r.t. the Robin parameter $c$, which is implied by the bound given in Theorem \ref{theoremReno} for $s=1$, $r_1\in \left\{0, 1\right\}$ and $r_2=0$, we obtain from $p_R\rightarrow_{c\rightarrow +\infty} p_D$
\begin{equation}
  \lim_{c\rightarrow +\infty}\mathcal{S}^{\Lambda_0,\Lambda_0}_{l,2;R}\left(\Vec{p}_n;\phi^{R}_{\tau_{1,2},y_{1,2}}\right)= 0.
\end{equation}
Therefore, we deduce 
\begin{equation}
  \mathcal{S}^{\Lambda,\Lambda_0}_{l,2;D}\left(\Vec{p}_n;\phi^{D}_{\tau_{1,2},y_{1,2}}\right)=\lim_{c\rightarrow +\infty}\mathcal{S}^{\Lambda,\Lambda_0}_{l,2;R}\left(\Vec{p}_n;\phi^{R}_{\tau_{1,2},y_{1,2}}\right)~.
\end{equation}
This ends the proof of (\ref{limit}).
\end{itemize}
\end{itemize}
\end{proof}
\subsection{Proof of Corollary \ref{Cor1}}
\begin{proof}
In this part, we prove the bounds (\ref{cD1}) and (\ref{cD2}). As a consequence of Theorem \ref{theoremReno}, we have for Robin boundary conditions
\begin{multline}
    \left| \mathcal{S}_{l,n;R}^{\Lambda,\Lambda_0}\left(\vec{p}_n;\phi_{\tau_{1,n},y_{1,n}}^{R}\right)\right|\\\leq \left(\Lambda+m\right)^{3-n}\mathcal{P}_1\left(\log \frac{\Lambda+m}{m}\right)\mathcal{P}_2\left(\frac{\left\|\vec{p}_n\right\|}{\Lambda+m}\right) \mathcal{Q}_1\left(\frac{\tau^{-\frac{1}{2}}}{\Lambda+m}\right)\ \mathcal{F}_{n,l;\delta}^{\Lambda,0}({\tau}_{1,n}),~~\forall n\geq 2~.
\end{multline}
Using Theorem 1 and taking the limit $c\rightarrow+\infty$, we deduce 
\begin{multline}
    \left|  \mathcal{S}_{l,n;D}^{\Lambda,\Lambda_0}\left(\vec{p}_n;\phi_{\tau,y_{1,s}}^{D}\right)\right|\\\leq \left(\Lambda+m\right)^{3-n}\mathcal{P}_1\left(\log \frac{\Lambda+m}{m}\right)\mathcal{P}_2\left(\frac{\left\|\vec{p}_n\right\|}{\Lambda+m}\right) \mathcal{Q}_1\left(\frac{\tau^{-\frac{1}{2}}}{\Lambda+m}\right)\ \mathcal{F}_{n,l;\delta}^{\Lambda,0}({\tau}_{1,n}),~~\forall n\geq 2~.
\end{multline}
For $n=2$, it is possible to obtain a sharper bound by performing a Taylor expansion around $0$ of the test functions $p_R\left(\tau_i;\cdot,y_i\right)$ as follows
\begin{multline}
   \mathcal{S}_{l,2;R}^{\Lambda,\Lambda_0}\left(\Vec{p}_n;\phi_{\tau_{1,2},y_{1,2}}^R\right)
    =s_{l;R}^{\Lambda,\Lambda_0}p_R\left(\tau_1;y_1,0\right)p_R\left(\tau_2;y_2,0\right)+e_{l;R}^{\Lambda,\Lambda_0}\left\{p_R\left(\tau_1;y_1,0\right)(\partial_{n}p_R)\left(\tau_2;y_2,0\right)\right.\\\left.+p_R\left(\tau_2;y_2,0\right)(\partial_{n}p_R)\left(\tau_1;y_1,0\right)\right\}+l_{l,2;R}^{\Lambda,\Lambda_0}\left(p_R\left(\tau_1;\cdot,y_1\right),p_R\left(\tau_2;\cdot,y_2\right)\right)\ .
\end{multline}
Taking the limit $c\rightarrow \infty$, we deduce
\begin{equation}\label{f1D}
   \mathcal{S}_{l,2;D}^{\Lambda,\Lambda_0}\left(\vec{p}_n;\phi_{\tau_{1,2},y_{1,2}}^D\right)\\
    =\lim_{c\rightarrow +\infty}\tilde{l}_{l,2;R}^{\Lambda,\Lambda_0}\left(p_R\left(\tau_1;\cdot,y_1\right),p_R\left(\tau_2;\cdot,y_2\right)\right)\ ,
\end{equation}
where the remainder $\tilde{l}_{l,2;R}^{\Lambda,\Lambda_0}$ is given by
\begin{align}
    &\left(\int_{z_1,z_2}z_1z_2~\mathcal{S}_{l,2;R}^{\Lambda,\Lambda_0}\left((z_1,p),(z_2,-p)\right)\right)\left(\partial_n\phi_1\right)(0)\left(\partial_n\phi_2\right)(0)\nonumber\\&+\int_{z_1,z_2}\mathcal{S}_{l,2;R}^{\Lambda,\Lambda_0}\left((z_1,p),(z_2,-p)\right)\left(\int_0^1dt~(1-t)\left(\partial_t^2\phi_1\right)(tz_1)\right)\times\left(\int_0^1dt'~(1-t')\left(\partial_{t'}^2\phi_2\right)(t'z_2)\right)\nonumber\\
    &+(\partial_n\phi_2)(0)\int_{z_1,z_2}z_2~\mathcal{S}_{l,2;R}^{\Lambda,\Lambda_0}\left((z_1,p),(z_2,-p)\right)\int_0^1dt ~(1-t)\left(\partial_t^2\phi_1\right)(tz_1)\nonumber\\
    &+(\partial_n\phi_1)(0)\int_{z_1,z_2}z_1~\mathcal{S}_{l,2;R}^{\Lambda,\Lambda_0}\left((z_1,p),(z_2,-p)\right)\int_0^1dt ~(1-t)\left(\partial_t^2\phi_2\right)(tz_2)~
\end{align}
and $\phi_i(z_i):=p_R\left(\tau_i;z_i,y_i\right)$. These terms can be bounded in a similar way as $\Dot{\mathcal{H}}_1$, $\Dot{\mathcal{H}}_2$ and $\Dot{\mathcal{H}}_3$ in the proof of Theorem 1. One should keep in mind that the test functions considered whithin the proof of Theorem \ref{theoremReno} were products of bulk heat kernels. However, the same bounds (\ref{h1}), (\ref{h2}) and (\ref{h3}) which are uniform in $c$, hold for Robin type test functions using the bounds (\ref{bulkBou}). Therefore, we deduce  
\begin{equation}\label{RobIn}
    \left|\tilde{l}_{l,2;R}^{\Lambda,\Lambda_0}\left(\phi_1,\phi_2\right)\right|\leq 
      \left(\Lambda+m\right)^{-2}\tau_1^{-\frac{1}{2}}\tau_2^{-\frac{1}{2}}\mathcal{P}\left(\log \frac{\Lambda+m}{m}\right)\mathcal{Q}\left(\frac{\tau^{-\frac{1}{2}}}{\Lambda+m}\right)\mathcal{F}_{2,l;\delta}^{\Lambda,0}\left({\tau}_{1,2}\right),
\end{equation}
which gives 
\begin{equation}
    \left|\mathcal{S}_{l,2;D}^{\Lambda,\Lambda_0}\left(\vec{p}_n;\phi_{\tau_{1,2},y_{1,2}}^R\right)\right|\leq \left(\Lambda+m\right)^{-2}\tau_1^{-\frac{1}{2}}\tau_2^{-\frac{1}{2}}\mathcal{P}\left(\log \frac{\Lambda+m}{m}\right)\mathcal{Q}\left(\frac{\tau^{-\frac{1}{2}}}{\Lambda+m}\right)\mathcal{F}_{2,l;\delta}^{\Lambda,0}\left({\tau}_{1,2}\right),
\end{equation}
and this ends the proof of Corollary \ref{Cor1}.
\end{proof}
\subsection{The minimal form of the bare interaction}
\noindent In this section, we show that the bare interaction (\ref{surface}) corresponds to the boundary conditions imposed in Theorem \ref{theoremReno} for $L_{\star}^{\Lambda,\Lambda_0}$. 
Given $\phi\in\mathrm{supp}~\mu_{\star}^{\Lambda,\Lambda_0}$, we expand $L^{\Lambda_0,\Lambda_0}_{\star}(\phi)$ in powers of the field $\phi$:
\begin{multline}
    L^{\Lambda_0,\Lambda_0}_{\star}(\phi)=\sum_{n=1}^{+\infty}\frac{1}{n!}\int_{\vec{z}_n}\int_{\mathbb{R}^{3n}}\prod_{i=1}^{n}\frac{d^3p_i}{(2\pi)^3}\mathcal{L}_{l,n;\star}^{\Lambda_0,\Lambda_0}\left((z_1,p_1),\cdots,(z_n,p_n)\right)\\\times\delta^{(3)}\left(p_1+\cdots+p_n\right)\phi(z_1,p_1)\cdots\phi(z_n,p_n).
\end{multline}
Using (\ref{dec}), we can write
\begin{equation*}
    L^{\Lambda_0,\Lambda_0}_{\star}(\phi)=D^{\Lambda_0,\Lambda_0}(\phi)+S^{\Lambda_0,\Lambda_0}_{\star}(\phi),
\end{equation*}
where 
\begin{multline}
    D^{\Lambda_0,\Lambda_0}(\phi):=\sum_{n=1}^{+\infty}\frac{1}{n!}\int_{\vec{z}_n}\int_{\mathbb{R}^{3n}}\prod_{i=1}^{n}\frac{d^3p_i}{(2\pi)^3}\mathcal{D}_{l,n}^{\Lambda_0,\Lambda_0}\left((z_1,p_1),\cdots,(z_n,p_n)\right)\delta^{(3)}\left(p_1+\cdots+p_n\right)\\\times\phi(z_1,p_1)\cdots\phi(z_n,p_n)
\end{multline}
and 
\begin{multline}
    S^{\Lambda_0,\Lambda_0}_{\star}(\phi)=\sum_{n=1}^{+\infty}\frac{1}{n!}\int_{\vec{z}_n}\int_{\mathbb{R}^{3n}}\prod_{i=1}^{n}\frac{d^3p_i}{(2\pi)^3}\mathcal{S}_{l,n;\star}^{\Lambda_0,\Lambda_0}\left((z_1,p_1),\cdots,(z_n,p_n)\right)\delta^{(3)}\left(p_1+\cdots+p_n\right)\\\times\phi(z_1,p_1)\cdots\phi(z_n,p_n)~.
\end{multline}
Proposition 1 in \cite{BorjiKopper2} implies that there exists $f$ in $L^2(\mathbb{R}^+)$ such that 
$$\phi(z,p)=\int_0^{\infty}dz'~C^{\Lambda,\Lambda_0}_\star(p;z,z')f(p,z')~,$$
which can be rewritten as
$$\phi(p,z)=\int_0^{\infty}dz'\int_{\frac{1}{\Lambda^2}}^{\frac{1}{\Lambda_0^2}}d\lambda~ e^{-\lambda(p^2+m^2)}p_{\star}\left(\lambda;z,z'\right)f(p,z')~.$$
Therefore, we can write
\begin{multline}
    \int_{z_1,\cdots,z_n}\mathcal{S}_{l,n;\star}^{\Lambda_0,\Lambda_0}\left((z_1,p_1),\cdots,(z_n,p_n)\right)\phi(z_1,p_1)\cdots\phi(z_n,p_n)\\=\int_{\vec{z}'_n}\prod_{i=1}^n\int_{\frac{1}{\Lambda^2}}^{\frac{1}{\Lambda_0^2}}d\lambda_i~e^{-\lambda_i(p_i^2+m^2)}f(p_i,z'_i)\mathcal{S}_{l,n;\star}^{\Lambda_0,\Lambda_0}\left(\Vec{p}_n;\phi_{\lambda_{1,n},z'_{1,n}}^{\star}\right).
\end{multline}
The boundary conditions (\ref{BCDTS2}) and (\ref{BCDDS}) imply that
\begin{equation}
    S^{\Lambda_0,\Lambda_0}_{\star}(\phi)=\frac{1}{2}\int_{\mathbb{R}^3}\frac{d^3p}{(2\pi)^3}\int_{z_1,z_2}\mathcal{S}_{l,2;\star}^{\Lambda_0,\Lambda_0}\left((z_1,p),(z_2,-p)\right)\phi(z_1,p)\phi(z_2,-p)~,
\end{equation}
where $\mathcal{S}_{l,2;\star}^{\Lambda_0,\Lambda_0}\left((z_1,p),(z_2,-p)\right)$ is given by (\ref{BCDTS}) for Robin/Neumann boundary conditions and by (\ref{BCDDS}) for Dirichlet boundary conditions:
\begin{itemize}
    \item \underline{Robin/Neumann boundary conditions ($c\geq 0$):} In this case, we obtain 
\begin{equation}
    S^{\Lambda_0,\Lambda_0}_{\star}(\phi)=\int_{\mathbb{R}^3}\frac{d^3p}{(2\pi)^3}\left(\frac{1}{2}s^{\Lambda_0}_{\star}+ce^{\Lambda_0}_{\star}\right)\phi(0,p)\phi(0,-p)~,
\end{equation}
where 
\begin{equation}
    s^{\Lambda_0}_{\star}=\int_{z_1,z_2}\mathcal{S}_{l,2;\star}^{\Lambda_0,\Lambda_0}\left((z_1,p),(z_2,-p)\right),~~~e^{\Lambda_0}_{\star}=\int_{z_1,z_2}z_1\mathcal{S}_{l,2;\star}^{\Lambda_0,\Lambda_0}\left((z_1,p),(z_2,-p)\right)~
\end{equation}
and $\star\in\left\{R,N\right\}$.
\item \underline{Dirichlet boundary conditions:} For Dirichlet boundary conditions, we obtain  $$S^{\Lambda_0,\Lambda_0}_D(\phi)=0~.$$
\end{itemize}

\section{The Amputated vs the Non-Amputated theory}
\noindent For quantum field theories on spaces without boundary, the renormalization problem of the amputated and the non-amputated theory is equivalent in the sense that the required counter-terms render finite the amputated and unamputated amplitudes, independently of the location of the external points of the unamputated diagrams \cite{KopperMuller,Keller}. However, first order calculations \cite{Albu2} give clear evidence that this is not the case when one considers the renormalization of the semi-infinite model. The tadpole diverges w.r.t. the UV cutoff, and its renormalization depends on the location of the external points (i.e. if they are on the surface or not). If the two external points are not on the surface, then in addition to the usual mass counter-term only one additional surface counter-term, which diverges linearly in the UV cutoff, is needed. This is not the case when one considers the tadpole with at least one external point on the surface. The latter needs one additional surface counter-term which diverges logarithmically w.r.t. the UV cutoff. This suggests that the amputated and non-amputated diagrams are renormalized differently for the semi-infinite model. In this section, we prove the following proposition which sheds some light on this finding:
\begin{proposition}\label{popp}
Let $\star\in\left\{R,N\right\}$. We denote by $C_{\star}$ the unregularized propagator associated to the boundary condition $\star$. For nonvanishing $s_{l;\star}^{\Lambda_0}$ and $e_{l;\star}^{\Lambda_0}$, we have for $y_2>0$
\begin{multline}\label{127}
   \lim_{y_1\rightarrow 0^+} \int_{z_1,z_2}\mathcal{S}_{l,2;\star}^{\Lambda_0,\Lambda_0}\left((z_1,p),(z_2,-p)\right)C_R\left(p;z_1,y_1\right)C_R\left(p;z_2,y_2\right)\\\neq \int_{z_1,z_2}\mathcal{S}_{l,2;\star}^{\Lambda_0,\Lambda_0}\left((z_1,p),(z_2,-p)\right)C_R\left(p;z_1,0\right)C_R\left(p;z_2,y_2\right)
\end{multline}
and 
\begin{multline}\label{128}
   \lim_{y_1\rightarrow 0^+}\lim_{y_2\rightarrow 0^+} \int_{z_1,z_2}\mathcal{S}_{l,2;\star}^{\Lambda_0,\Lambda_0}\left((z_1,p),(z_2,-p)\right)C_R\left(p;z_1,y_1\right)C_R\left(p;z_2,y_2\right)\\\neq \int_{z_1,z_2}\mathcal{S}_{l,2;\star}^{\Lambda_0,\Lambda_0}\left((z_1,p),(z_2,-p)\right)C_R\left(p;z_1,0\right)C_R\left(p;z_2,0\right)~.
\end{multline}
\end{proposition}
\begin{proof}
We give the proof of Proposition \ref{popp} in the case of Robin boundary conditions. Neumann b.c. can be treated analogously.\\
We proved in Section 3 that for the particular choice of the boundary conditions (\ref{BCDTS})-(\ref{renoS}), the bare interaction is of the form (\ref{surface}) and we have 
\begin{equation}
    \mathcal{S}_{l,2;R}^{\Lambda_0,\Lambda_0}\left((z_1,p),(z_2,-p)\right)=\left(s_R^{\Lambda_0}+e_R^{\Lambda_0}(\partial_{z_1}+\partial_{z_2})\right)\delta_{z_1}\delta_{z_2}.
\end{equation}
Hence, we obtain for $y_1,~y_2>0$
\begin{multline}
    \int_{z_1,z_2}\mathcal{S}_{l,2;R}^{\Lambda_0,\Lambda_0}\left((z_1,p),(z_2,-p)\right)C_R\left(p;z_1,y_1\right)C_R\left(p;z_2,y_2\right)=s_R^{\Lambda_0}C_R\left(p;0,y_1\right)C_R\left(p;0,y_2\right)\\+e_R^{\Lambda_0}\left\{C_R\left(p;0,y_1\right)\partial_{n}C_R\left(p;0,y_2\right)+C_R\left(p;0,y_2\right)\partial_{n}C_R\left(p;0,y_1\right)\right\}.
\end{multline}
Using 
\begin{equation}
    \lim_{z\rightarrow 0}\lim_{y\rightarrow 0}\partial_z C_R\left(p;z,y\right)=-\kappa_pC_R\left(p;0,0\right),~~~~\lim_{y\rightarrow 0}\lim_{z\rightarrow 0}\partial_z C_R\left(p;z,y\right)= c~C_R\left(p;0,0\right)
\end{equation}
with $\kappa_p:=\sqrt{p^2+m^2}$, we deduce  
\begin{multline}\label{136}
    \int_{z_1,z_2}\mathcal{S}_{l,2;R}^{\Lambda_0,\Lambda_0}\left((z_1,p),(z_2,-p)\right)C_R\left(p;z_1,y_1\right)C_R\left(p;z_2,y_2\right)\\=\left(s_R^{\Lambda_0}+2c~e_R^{\Lambda_0}\right)C_R\left(p;0,y_1\right)C_R\left(p;0,y_2\right),~~~\forall y_1,y_2>0,
\end{multline}
\begin{multline}\label{137}
    \int_{z_1,z_2}\mathcal{S}_{l,2;R}^{\Lambda_0,\Lambda_0}\left((z_1,p),(z_2,-p)\right)C_R\left(p;z_1,y_1\right)C_R\left(p;z_2,0\right)\\=\left(s^{\Lambda_0}_R+2c~e^{\Lambda_0}_R\right)C_R\left(p;0,y_1\right)C_R\left(p;0,0\right)-e^{\Lambda_0}_RC_R\left(p;0,y_1\right)\left(\kappa_p+c\right)C_R\left(p;0,0\right),
\end{multline}
\begin{multline}\label{138}
    \int_{z_1,z_2}\mathcal{S}_{l,2}^{\Lambda_0,\Lambda_0}\left((z_1,p),(z_2,-p)\right)C_R\left(p;z_1,0\right)C_R\left(p;z_2,0\right)\\=\left(s^{\Lambda_0}_R+2c~e_R^{\Lambda_0}\right)C_R\left(p;0,0\right)C_R\left(p;0,0\right)-2~e^{\Lambda_0}_RC_R\left(p;0,0\right)\left(\kappa_p+c\right)C_R\left(p;0,0\right),
\end{multline}
from which (\ref{127}) and (\ref{128}) follow directly.
\end{proof}
\begin{remark}
 Denoting $g_R^{\Lambda_0}=s_R^{\Lambda_0}+2ce_R^{\Lambda_0}$, we deduce that (\ref{136}) implies that the unamputated two-point function of the semi-infinite model which has two external points in the interior of the bulk, requires only the surface counter-term $g_R^{\Lambda_0}$ to be renormalized. This is not the case when at least one of the external points is on the surface. From (\ref{137}) and (\ref{138}), we deduce that $g_R^{\Lambda_0}$ is not sufficient and the additional surface counter-term $e_R^{\Lambda_0}$ is required to make the two-point function finite. This generalizes the remarks given in \cite{Diehl,DiehlD2} and \cite{Albu2} concerning the tadpole to all loop orders.
\end{remark}
\newpage
\appendix

\section{Some properties of the surface weight factor for $s=2$ and $l\geq 1$}
In this Appendix, we prove several lemmas that we use in the proof of Theorem \ref{theoremReno}. These lemmas concern the case $s=2$ for which the set of partitions $\mathcal{P}_2$ simply reads
\begin{equation*}
    \mathcal{P}_2:=\left\{\Pi_0,\Pi_1\right\},
\end{equation*}
where $\Pi_0:=\sigma_2$, $\Pi_1=\pi_1\cup\pi_2$ and $\pi_i=\left\{i\right\}$.\\
From the definition (\ref{defs}), we have
$$\mathcal{W}^2_l(\sigma_2)=\left\{T^{2,0}_l(Y_{\sigma_2},0;\vec{z})|~~T^{2,0}_l\in\mathcal{T}^{2,0},~~~v_2\leq 3l-1\right\}$$
and \begin{align*}
\mathcal{W}^2_l(\Pi_1)&=\left\{T^{1,0}_{l;1}(y_{1},0;\vec{z})\cup T^{1,0}_{l;2}(y_{2},0;\vec{z'})|~~T^{1,0}_{l;1},~T^{1,0}_{l;2}\in\mathcal{T}^{1,0}_l\right\}~\\
&=\mathcal{W}^1_l(\pi_1)\cup\mathcal{W}^1_l(\pi_2),~
\end{align*}
which implies that the global surface weight factor $\mathcal{F}^{\Lambda,0}_{2,l;\delta}(\tau_{1,2})$ simply reads
\begin{multline}\label{weight}
    \mathcal{F}^{\Lambda,0}_{2,l;\delta}\left(\tau_{1,2}\right)=\sum_{T^{2,0}_l\in\mathcal{W}^2_l(\sigma_2)}\mathcal{F}^0_{\delta}\left(\Lambda,\tau_{1},\tau_{2};T^{2,0}_l;Y_{\sigma_2}\right)+\mathcal{F}^{\Lambda,0}_{1,l;\delta}\left(\tau_1\right)\mathcal{F}^{\Lambda,0}_{1,l;\delta}\left(\tau_2\right),
    \end{multline}
    where 
    \begin{equation}
       \mathcal{F}^{\Lambda,0}_{1,l;\delta}\left(\tau_i\right):= \sum_{T^{1,0}_l\in\mathcal{W}^1_l(\pi_i)}\mathcal{F}^0_{\delta}\left(\Lambda,\tau_{i};T^{1,0}_l;y_{i}\right).
    \end{equation}
Note that (\ref{weight}) implies
\begin{equation}\label{borneweights2}
    \mathcal{F}^{\Lambda,0}_{1,l;\delta}\left(\tau_1\right)\mathcal{F}^{\Lambda,0}_{1,l;\delta}\left(\tau_2\right)\leq \mathcal{F}^{\Lambda,0}_{2,l;\delta}\left(\tau_{1,2}\right).
\end{equation}
\begin{lemma}\label{lemmSi}
Let $v$ be the total number of vertices of incidence number $2$ of the tree $T^{2,0}_l$. For $p\geq 1$, we denote by $\vec{z}=\left(z_1,\cdots,z_p\right)$ the set of the internal vertices of $T^{2,0}_l$. For $\Lambda_{\mathcal{I}}:=\left\{\Lambda_i~|1\leq i\leq v-1\right\}$, $\tilde{\Lambda}\in\left[\Lambda,\Lambda_0\right]$ and $\tau_1,\tau_2>0$, we have 
\begin{multline}\label{lemo}
  \int_{\vec {z}}\mathcal{F}^0_{\delta}\left(\Lambda_{\mathcal{I}},\tilde{\Lambda};\tau_{1},\tau_{2};T^{2,0}_l;\vec{z};Y_{\sigma_2}\right)\\\leq \int_{0}^{\infty}dz~p_B\left(c_{1,\delta};z,y_1\right)p_B\left(c_{2,\delta};z,y_2\right)p_B\left(\frac{1+\delta}{\tilde{\Lambda}_1^2};z,0\right)~
\end{multline}
and
\begin{multline}\label{lemoo2}
\int_{0}^{\infty}dz~p_B\left(c_{1,\delta};z,y_1\right)p_B\left(c_{2,\delta};z,y_2\right)p_B\left(\frac{1+\delta}{\tilde{\Lambda}_1^2};z,0\right)\\\leq 2^v~\int_{\vec {z}}\mathcal{F}^0_{\delta}\left(\Lambda_{\mathcal{I}},\tilde{\Lambda};\tau_1,\tau_2;T^{2,0}_l;\vec{z};Y_{\sigma_2}\right),
\end{multline}
where $0\leq v_1,v_2,v_0\leq v$ such that $v_1+v_2+v_0=v$, $c_{1,\delta}:=(1+\delta)c_1$ and $c_{2,\delta}:=(1+\delta)c_2$. The parameters $c_1$, $c_2$ and $\tilde{\Lambda}_1$ are given by
\begin{align}\label{c1c2}
    c_1&= \tau_1+\left(\sum_{i=1}^{v_1}\frac{1}{\Lambda_i^2}\right)\left(1-\delta_{v_1,0}\right),\\c_2&=\tau_2+\left(\sum_{i=v_1+1}^{v_1+v_2-1}\frac{1}{\Lambda_{i}^2}\right)\left(1-\delta_{v_2,0}\right),\\\frac{1}{\tilde{\Lambda}_1^2}&=\left(\sum_{i=v_1+v_2}^{v-1}\frac{1}{{\Lambda}_{i}^2}\right)\left(1-\delta_{v,0}\right)+\frac{1}{\tilde{\Lambda}^2}~.\label{A8} 
    \end{align}
\end{lemma}
\begin{proof}
\begin{itemize}
\item First, we prove the bound (\ref{lemo}). A tree $T^{2,0}_l$ in $\mathcal{W}^2_l(\sigma_2)$ has the following structure
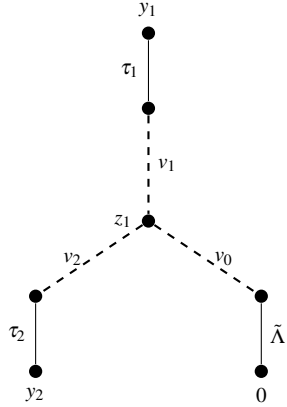
\begin{center}
\begin{tikzpicture}[font=\footnotesize]
  \tikzset{
    level 1/.style={level distance=10mm,sibling distance=30mm},
    level 2/.style={level distance=15mm,sibling distance=30mm},
    level 3/.style={level distance=10mm,sibling distance=30mm},
    level 4/.style={level distance=10mm},
    emph/.style={edge from parent/.style={dashed,thick,draw}},
    norm/.style={edge from parent/.style={solid,thin,draw}}
  }

  \node[circle,fill=black,inner sep=0pt,minimum size=5pt,label=above:{$y_1$}]{}
    child[norm]{node[circle,fill=black,inner sep=0pt,minimum size=5pt,label=left:{}]{}
    child[emph]{node[circle,fill=black,inner sep=0pt,minimum size=5pt,label=left:{$z_1$}]{}
    child{node[circle,fill=black,inner sep=0pt,minimum size=5pt,label=left:{}]{}
    child[norm]{node[circle,fill=black,inner sep=0pt,minimum size=5pt,label=below:{$y_2$}]{}edge from parent node[left]{$\tau_2$}}edge from parent node[left]{$v_2$}}
    child[emph]{node[circle,fill=black,inner sep=0pt,minimum size=5pt,label=left:{}]{}
    child[norm]{node[circle,fill=black,inner sep=0pt,minimum size=5pt,label=below:{$0$}]{}edge from parent node[right]{$\tilde{\Lambda}$}}
        edge from parent node[right]{$v_0$}}
        edge from parent node[right]{$v_1$}
     }
     edge from parent node[left]{$\tau_1$}
    }
  ;
\end{tikzpicture}
\end{center}

It contains one internal vertex of incidence number $3$ and all the other internal vertices are of incidence number $2$. We assume that each dashed line contains a number $v_i\geq 1$ of internal vertices of incidence number $2$. The case $v_i=0$ can be treated similarly. Remember that $v_0+v_1+v_2=v$.
Let $\left\{z_2,\cdots,z_{v_1+1}\right\}$, $\left\{z_{v_1+2},\cdots,z_{v_1+v_2}\right\}$ and $\left\{z_{v_1+v_2+1},\cdots,z_{v}\right\}$ be respectively the internal vertices on the paths from $z_1$ to $y_1$, $z_1$ to $y_2$ and $z_1$ to $0$. From (\ref{treeStr'}), the integral surface weight factor of $T^{2,0}_l$ is then given by 
\begin{multline}\label{SurfaceWeight}
    \int_{\vec {z}}\mathcal{F}^0_{\delta}\left(\Lambda_{\mathcal{I}},\tilde{\Lambda};\tau_{1},\tau_{2};T^{2,0}_l;\vec{z};Y_{\sigma_2}\right)=\int_{z_1,\cdots,z_v}~\prod_{j=2}^{v_1+1}p_B\left(\frac{1+\delta}{\Lambda_{j-1}^2};z_{j-1},z_j\right)p_B\left(\tau_{1,\delta};z_{v_1+1},y_1\right)\\
    \times~p_B\left(\frac{1+\delta}{\Lambda_{v_1+1}^2};z_1,z_{v_1+2}\right)\prod_{j=v_1+3}^{v_1+v_2}p_B\left(\frac{1+\delta}{\Lambda_{j-1}^2};z_{j-1},z_j\right)p_B\left(\tau_{2,\delta};z_{v_1+v_2},y_2\right)\\\times~p_B\left(\frac{1+\delta}{\Lambda_{v_1+v_2}^2};z_1,z_{v_1+v_2+1}\right)\prod_{j=v_1+v_2+2}^{v}p_B\left(\frac{1+\delta}{\Lambda_{j-1}^2};z_{j-1},z_j\right)p_B\left(\frac{1+\delta}{\tilde{\Lambda}^2};z_{v},0\right)~.
\end{multline}
Bounding the integral over $\mathbb{R}^+$ by the integral over $\mathbb{R}$ and using (\ref{rr+}), we obtain
\begin{eqnarray}
\int_{z_2\cdots z_{v_1+1}}~\prod_{j=2}^{v_1+1}p_B\left(\frac{1+\delta}{\Lambda_{j-1}^2};z_{j-1},z_j\right)p_B(\tau_{1,\delta};z_{v_1+1},y_1)\leq  p_B(c_{1,\delta};z_1,y_1)~,
\end{eqnarray}
where $c_{1,\delta}=(1+\delta)\left(\tau_1+\sum_{i=1}^{v_1}\frac{1}{\Lambda_i^2}\right)$. Proceeding similarly on the paths from $z_1$ to $y_2$ and  from $z_1$ to $0$, we obtain that the weight factor of a tree in $\mathcal{W}^2_l(\sigma_2)$ is bounded by the weight factor of the tree 
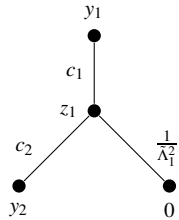
\begin{center}
\begin{tikzpicture}[font=\footnotesize]
  \tikzset{
    level 1/.style={level distance=10mm,sibling distance=20mm},
    level 2/.style={level distance=10mm,sibling distance=20mm},
    level 3/.style={level distance=10mm,sibling distance=20mm},
    level 4/.style={level distance=10mm,sibling distance=20mm},
    emph/.style={edge from parent/.style={dashed,thick,draw}},
    norm/.style={edge from parent/.style={solid,thin,draw}}
  }

  \node[circle,fill=black,inner sep=0pt,minimum size=5pt,label=above:{$y_1$}]{}
    child[norm]{node[circle,fill=black,inner sep=0pt,minimum size=5pt,label=left:{$z_1$}]{}
    child[norm]{node[circle,fill=black,inner sep=0pt,minimum size=5pt,label=below:{$y_2$}]{}edge from parent node[left,xshift=-5]{$c_2$}}
    child[norm]{node[circle,fill=black,inner sep=0pt,minimum size=5pt,label=below:{$0$}]{}edge from parent node[right,xshift=5]{$\frac{1}{\tilde{\Lambda}_1^2}$}}
    edge from parent node[left]{$c_1$}}
  ;
\end{tikzpicture}
\end{center}
where $c_1$, $c_2$ and $\frac{1}{\tilde{\Lambda}_1^2}$ are the new parameters associated respectively to the edges $(z_1,y_1)$, $(z_1,y_2)$ and $(z_1,0)$. The relations between these new parameters and those of the tree $T^{2,0}_l$ are given by
$$c_2=\tau_2+\left(\sum_{i=1}^{v_2}\frac{1}{\Lambda_{i+v_1}^2}\right)~~~~~\mathrm{and}~~\frac{1}{\tilde{\Lambda}_1^2}=\left(\sum_{i=v_1+v_2}^{v-1}\frac{1}{{\Lambda}_{i}^2}\right)+\frac{1}{\tilde{\Lambda}^2}~.$$
This proves the statement (\ref{lemo}). 
\item To prove (\ref{lemoo2}), we assume without loss of generality that $v_{1}\geq 1$. If $y_1\leq 0$, we have for $z_1\geq 0$
\begin{equation}
  p_B\left(c_{1,\delta};z_1,y_1\right)\leq p_B\left(c_{1,\delta};z_1,-y_1\right).
\end{equation}
Hence, we restrict our discussion to $y_1\geq 0$. Using (\ref{rr++}), we find 
\begin{multline}\label{A8}
p_B(c_{1,\delta};z_1,y_1)=  \int_{\mathbb{R}^{v_1}}dz_2\cdots dz_{v_1+1}~\prod_{j=2}^{v_1+1}p_B\left(\frac{1+\delta}{\Lambda_{j-2}^2};z_j,z_{j-1}\right)p_B(\tau_{1,\delta};z_{v_1+1},y_1)\\
\leq 2^{v_1}\int_{(\mathbb{R}^+)^{v_1}}dz_2\cdots dz_{v_1+1}~\prod_{j=2}^{v_1+1}p_B\left(\frac{1+\delta}{\Lambda_{j-2}^2};z_j,z_{j-1}\right)p_B(\tau_{1,\delta};z_{v_1+1},y_1)~.
\end{multline}
 Proceeding similarly for $p_B(c_{2,\delta};z_1,y_2)$ and $p_B\left(\frac{1+\delta}{\tilde{\Lambda}_1^2};z,0\right)$, we deduce 
\begin{multline}\label{sirol}
 \int_{0}^{\infty}dz~p_B\left(c_{1,\delta};z,y_1\right)p_B\left(c_{2,\delta};z,y_2\right)p_B\left(\frac{1+\delta}{\tilde{\Lambda}_1^2};z,0\right)\\\leq 2^{v}\int_{\vec {z}}\mathcal{F}^0_{\delta}\left(\Lambda_{\mathcal{I}},\tilde{\Lambda};\tau_{1},\tau_{2};T^{2,0}_l;\vec{z};Y_{\sigma_2}\right).
\end{multline}
\end{itemize}
\end{proof}
\begin{lemma}\label{LemmPi}
Let $W^2_l(\Pi_1):=T^{1,0}_{l;1}(y_{1},0;\vec{z})\cup T^{1,0}_{l;2}(y_{2},0;\vec{z'})$ be a forest in $\mathcal{W}^2_l(\Pi_1)$ and $v_{2,1}$ (resp. $v_{2,2}$) the total number of vertices of incidence number $2$ of the tree $T^{1,0}_{l;1}$ (resp. $T^{1,0}_{l;2}$). For $p,~q\geq 1$, we denote by $\vec{z}=\left(z_1,\cdots,z_p\right)$ (resp. $\vec{z}'=\left(z'_1,\cdots,z'_q\right)$) the set of the internal vertices of $T^{1,0}_{l;1}$ (resp. $T^{1,0}_{l;2}$). For $\Lambda_{\mathcal{I}}:=\left\{\Lambda_i,~\Lambda'_j~|1\leq i\leq v_{2,1}-1,~1\leq j \leq v_{2,2}-1\right\}$, $\tilde{\Lambda}_1,~\tilde{\Lambda}_2$  $\in\left[\Lambda,\Lambda_0\right]$ and $\tau_1,\tau_2>0$, we have
\begin{equation}
  \int_{\vec {z}}\mathcal{F}^0_{\delta}\left(\Lambda_{\mathcal{I}},\tilde{\Lambda}_1,\tilde{\Lambda}_2;\tau_1,\tau_2;W^2_l(\Pi_1);\vec{z};Y_{\sigma_2}\right)\leq p_B\left(\tilde{c}_{1,\delta};y_1,0\right)p_B\left(\tilde{c}_{2,\delta};y_2,0\right)
\end{equation}
and
\begin{equation}\label{lemo2}
p_B\left(\tilde{c}_{1,\delta};y_1,0\right)p_B\left(\tilde{c}_{2,\delta};y_2,0\right)\leq 2^{v_{2,1}+v_{2,2}}~\int_{\vec {z}}\mathcal{F}^0_{\delta}\left(\Lambda_{\mathcal{I}},\tilde{\Lambda}_1,\tilde{\Lambda}_2;\tau_1,\tau_2;W^2_l(\Pi_1);\vec{z};Y_{\sigma_2}\right),
\end{equation}
where $\tilde{c}_1$ and $\tilde{c}_2$ are given by
\begin{align*}
    \tilde{c}_1&=2\tau_1+\sum_{i=1}^{v_{2,1}-1}\frac{1}{\Lambda_{i,1}^2}+\frac{1}{\tilde{\Lambda}_1^2}~,\\\tilde{c}_2&=2\tau_2+\sum_{i=1}^{v_{2,2}-1}\frac{1}{{\Lambda'}_{i,2}^2}+\frac{1}{\tilde{\Lambda}_2^2}~. 
    \end{align*}
\end{lemma}
\begin{proof}
The forest $W^2_l(\Pi_1)$ is of the following form 
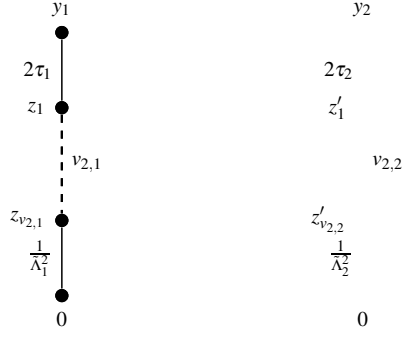
\begin{center}
\begin{tikzpicture}[font=\footnotesize]
  \tikzset{
    level 1/.style={level distance=10mm,sibling distance=30mm},
    level 2/.style={level distance=15mm,sibling distance=30mm},
    level 3/.style={level distance=10mm,sibling distance=30mm},
    level 4/.style={level distance=10mm,sibling distance=40mm},
    emph/.style={edge from parent/.style={dashed,thick,draw}},
    norm/.style={edge from parent/.style={solid,thin,draw}}
  }

  \node[circle,fill=black,inner sep=0pt,minimum size=5pt,label=above:{$y_1$}]{}
    child[norm]{node[circle,fill=black,inner sep=0pt,minimum size=5pt,label=left:{$z_1$}]{}
    child[emph]{node[circle,fill=black,inner sep=0pt,minimum size=5pt,label=left:{$z_{v_{2,1}}$}]{}
    child[norm]{node[circle,fill=black,inner sep=0pt,minimum size=5pt,label=below:{$0$}]{}edge from parent node[left]{$\frac{1}{\tilde{\Lambda}_1^2}$}}
        edge from parent node[right]{$v_{2,1}$}
     }
     edge from parent node[left]{$2\tau_1$}
    }
    
    ;
    
  \hspace*{4cm}
  \node[circle,fill=black,inner sep=0pt,minimum size=5pt,label=above:{$y_2$}]{}
    child[norm]{node[circle,fill=black,inner sep=0pt,minimum size=5pt,label=left:{$z'_1$}]{}
    child[emph]{node[circle,fill=black,inner sep=0pt,minimum size=5pt,label=left:{$z'_{v_{2,2}}$}]{}
    child[norm]{node[circle,fill=black,inner sep=0pt,minimum size=5pt,label=below:{$0$}]{}edge from parent node[left]{$\frac{1}{\tilde{\Lambda}_2^2}$}}
        edge from parent node[right]{$v_{2,2}$}
     }
     edge from parent node[left]{$2\tau_2$}
    }
    
    ;
\end{tikzpicture}
\end{center}
The integrated weight factor of this forest reads
\begin{multline*}
\int_{\vec {z}}\mathcal{F}^0_{\delta}\left(\Lambda_{\mathcal{I}},\tilde{\Lambda}_1,\tilde{\Lambda}_2;\tau_{1,2};W^2_l(\Pi_1);\vec{z};y_1,y_2\right)\\=\int_0^{\infty}dz_1\cdots dz_{v_{2,1}}~\prod_{j=2}^{v_{2,1}}p_B\left(\frac{1+\delta}{\Lambda_{j-1}^2};z_j,z_{j-1}\right)p_B(2\tau_{1,\delta};z_1,y_1)p_B\left(\frac{1+\delta}{\tilde{\Lambda}_{1}^2};z_{v_{2,1}},0\right)\\\times \int_0^{\infty}d{z'}_1\cdots dz'_{v_{2,2}}~\prod_{j=2}^{v_{2,2}}p_B\left(\frac{1+\delta}{{\Lambda'}_{j-1}^2};{z'}_j,{z'}_{j-1}\right)p_B(2\tau_{2,\delta};{z'}_1,y_2)p_B\left(\frac{1+\delta}{\tilde{\Lambda}_{2}^2};z_{v_{2,2}},0\right)~.
\end{multline*}
Bounding the integral over $\mathbb{R}^+$ by the integral over $\mathbb{R}$ and using (\ref{rr+}), we obtain 
\begin{equation*}
\int_{\vec {z}}\mathcal{F}^0\left(\Lambda_{\mathcal{I}},\tilde{\Lambda}_1,\tilde{\Lambda}_2;\tau_{1,2};W^2_l(\Pi_1);\vec{z};y_1,y_2\right)\leq p_B\left(\tilde{c}_{1,\delta};y_1,0\right)p_B\left(\tilde{c}_{2,\delta};y_2,0\right),
\end{equation*}
where for $v_{2,i}>1$ 
\begin{align}
\tilde{c}_1&=2\tau_1+\sum_{i=1}^{v_{2,1}-1}\frac{1}{\Lambda_{i,1}^2}+\frac{1}{\tilde{\Lambda}_1^2}~,\\\tilde{c}_2&=2\tau_2+\sum_{i=1}^{v_{2,2}-1}\frac{1}{\Lambda_{i,2}^2}+\frac{1}{\tilde{\Lambda}_2^2}~.
\end{align}
  If $v_{2,i} =1$, then $\tilde{c}_i=2\tau_i+\frac{1}{\tilde{\Lambda}_i^2}$.\\
Using again (\ref{rr++}) and proceeding as in  (\ref{A8}), we deduce 
\begin{equation}
p_B\left(\tilde{c}_{1,\delta};y_1,0\right)p_B\left(\tilde{c}_{2,\delta};y_2,0\right)\leq 2^{v_{2,1}+v_{2,2}}~\int_{\vec {z}}\mathcal{F}^0\left(\Lambda_{\mathcal{I}},\tilde{\Lambda}_1,\tilde{\Lambda}_2;\tau_{1,2};W^2_l(\Pi_1);\vec{z};y_1,y_2\right).
\end{equation}
\end{proof}
\begin{lemma}\label{lemma7}
For $0\leq \alpha \leq 1$ and $y_1,~y_2\in\mathbb{R}$,  we have 
\begin{equation}\label{resu11}
   \left|\partial_n^{\alpha}\phi_i(0)\right|\leq C_{0,\delta}~\tau_i^{-\frac{\alpha}{2}}~\left(1+\frac{\tau^{-\frac{1}{2}}_i}{\Lambda+m}\right) \mathcal{F}_{1,l;\delta}^{\Lambda,0}\left(\tau_i\right)~,~~~\forall~0 <\delta<1,
\end{equation}
where $\partial_n^{\alpha}\phi_i(0)=\lim_{z_i\rightarrow 0^+}\partial_{z_i}^{\alpha}\phi_i(z_i)$ with $\phi_i(z_i):=p_B\left(\tau_i;z_i,y_i\right)$ and $C_{0,\delta}$ is defined in (\ref{in1}).
\end{lemma}
\begin{proof}
For $\alpha=0$, we have 
    \begin{equation}\label{b1}
    \phi_i(0)=\frac{1}{\sqrt{2\pi\tau_i}}e^{-\frac{y_i^2}{2\tau_i}}\leq \sqrt{2}~p_B\left(2\tau_i;y_i,0\right).
    \end{equation}
Taking $\tilde{\Lambda}\geq \Lambda+m$, we write 
\begin{equation}\label{b1}
    p_B\left(2\tau_{i};y_i,0\right)\leq \sqrt{1+\frac{\tau_i^{-1}}{\tilde{\Lambda}^2}}~p_B\left(2\tau_{i}+\frac{1}{\Tilde{\Lambda}^2};y_i,0\right)\leq ~\left(1+\frac{\tau_i^{-\frac{1}{2}}}{\Lambda+m}\right)~p_B\left(2\tau_{i}+\frac{1}{\Tilde{\Lambda}^2};y_i,0\right).
\end{equation}
For $\alpha=1$, we have $\left|\partial_n\phi_i(0)\right|=\frac{y_i}{\tau_i}~\phi_i(0)$.
Using the bound (\ref{in1}) for $r=1$ together with (\ref{b1}), we obtain 
\begin{equation}\label{a22}
    \left|\partial_n\phi_i(0)\right|\leq \sqrt{2}~C_{0,\delta}~\tau_i^{-\frac{1}{2}}\left(1+\frac{\tau_i^{-\frac{1}{2}}}{\Lambda+m}\right)~p_B\left(2\tau_{i,\delta}+\frac{1+\delta}{\Tilde{\Lambda}^2};y_i,0\right).
\end{equation}
We consider the surface tree $T_{l;i}^{1,0}$ which consists of the internal vertex $z$, the external vertex $y_i$ and the surface external vertex $0$. We associate respectively to each of the external line $(z,y_i)$ and the surface external line $(z,0)$ the parameters $\tau_i$ and $\tilde{\Lambda}$. The integrated surface weight factor of the tree $T_{l;i}^{1,0}$ then reads
\begin{equation}\label{a19}
    \mathcal{F}^0_{1,l;\delta}\left(\tilde{\Lambda},\tau_i;T_{l;i}^{1,0};y_i\right)=\int_z p_B\left(2\tau_{i,\delta};z,y_i\right)~p_B\left(\frac{1+\delta}{\tilde{\Lambda}^2};z,0\right).
\end{equation}
Using (\ref{a19}), (\ref{rr+}) and (\ref{rr++}) we have
\begin{multline}\label{stot}
    p_B\left(2\tau_{i,\delta}+\frac{1+\delta}{\tilde{\Lambda}^2};y_i,0\right)=\int_{\mathbb{R}}dz~ p_B\left(2\tau_{i,\delta};y_i,z\right)~p_B\left(\frac{1+\delta}{\tilde{\Lambda}^2};z,0\right)\\\leq 2 ~\mathcal{F}^0_{1,l;\delta}\left(\tilde{\Lambda},\tau_i;T_{l;i}^{1,0};y_i\right).
\end{multline}
Combining the bounds (\ref{b1}), (\ref{a22}) and (\ref{stot}), we deduce 
\begin{equation}\label{resu1}
   \left|\partial_n^{\alpha}\phi_i(0)\right|\leq ~\sqrt{2}~C_{0,\delta}~\tau_i^{-\frac{\alpha}{2}}~\left(1+\frac{\tau_i^{-\frac{1}{2}}}{\Lambda+m}\right) \mathcal{F}_{1,l;\delta}^{\Lambda,0}\left(\tau_i\right),~~~\alpha\in\left\{0, 1\right\}.
\end{equation}
Furthermore, we also obtain
\begin{equation}\label{bb3}
    \left|\partial_n^{\alpha}\phi_1(0)\partial_n^{\beta}\phi_2(0)\right|\leq 2~C_{0,\delta}^2~\tau_1^{-\frac{\alpha}{2}}\tau_2^{-\frac{\beta}{2}}\left(1+\frac{\tau^{-\frac{1}{2}}}{\Lambda+m}\right)^2\prod_{i=1}^2 \mathcal{F}_{1,l;\delta}^{\Lambda,0}\left(\tau_i\right),~~~\alpha,~\beta\in\left\{0, 1\right\}.
\end{equation}
Recalling (\ref{borneweights2}), we deduce 
\begin{equation}\label{bb3}
    \left|\partial_n^{\alpha}\phi_1(0)\partial_n^{\beta}\phi_2(0)\right|\leq \tau_1^{-\frac{\alpha}{2}}\tau_2^{-\frac{\beta}{2}}~\mathcal{Q}\left(\frac{\tau^{-\frac{1}{2}}}{\Lambda+m}\right) \mathcal{F}_{2,l;\delta}^{\Lambda,0}\left(\tau_{1,2}\right),
\end{equation}
where ${\mathcal{Q}}(x)=2~C_{0,\delta}^2\left(1+x\right)^{2}$.\\
\end{proof}
\begin{lemma}\label{lemma8}
For $0<t\leq 1$, $\gamma\in\mathbb{N}$ and $0<\delta<\delta'< 1$, we have 
\begin{equation}
    \left(\frac{y_1}{\sqrt{\tau_1}}\right)^{\gamma}\mathcal{F}_{1,l;\delta}^{\Lambda,0}\left(t,\tau_{1}\right)\leq O(1)~t\left(1+\frac{\tau_1^{-\frac{1}{2}}}{\Lambda}\right)^{\gamma}\mathcal{F}_{1,l;\delta'}^{\Lambda,0}(\tau_{1})~,
\end{equation}
where
\begin{equation}
\mathcal{F}_{1,l;\delta}^{\Lambda,0}\left(t,\tau_{1}\right):=\sum_{W^1_l \in \mathcal{W}^1_l\left(\pi_1\right)}\mathcal{F}^{0}_{\delta}\left(\Lambda,\frac{\tau_{1}}{t^2};W^1_l;\frac{y_1}{t}\right)=\sum_{v=1}^{3l-1}\mathcal{F}^{0}_{\delta}\left(\Lambda,\frac{\tau_{1}}{t^2};T^{1,0}_l;\frac{y_1}{t}\right),
\end{equation}
and $O(1)$ is a constant which depends on $\delta$, $\delta'$, $v$, $\gamma$ and the loop order $l$.
\end{lemma}
\begin{proof}
For $s=1$ the global set of forests $\mathcal{W}_{1,l}(\pi_1)$ consists of all surface trees with two external vertices (including the surface external vertex) which have a number of vertices of incidence number $2$ less or equal to $3l-2$. We consider the surface tree $T^{1,0}_l\in\mathcal{W}^1_{l}(\pi_1)$ with the external vertex $\frac{y_1}{t}$ and the internal vertices $\vec{z}_v=\left(z_1,\cdots,z_v\right)$. Let $v$ be the number of its vertices of incidence number $2$ and $\Lambda_{\mathcal{I}}=\left({\Lambda}_i\right)_{1\leq i\leq v-1}$, $\Lambda_v$ and $\tau_1/t^2$ be respectively the parameters associated to the internal lines, the surface external line and the external line of $T^{1,0}_l$. Then the integrated surface weight factor of $T^{1,0}_l$ reads 
\begin{multline}\label{C22}
    \mathcal{F}^0_{\delta}\left(\Lambda_{\mathcal{I}};\frac{\tau_1}{t^2};T^{1,0}_l;\frac{y_1}{t}\right)=\int_{\vec{z}_v}p_B\left(\frac{2\tau_{1,\delta}}{t^2};\frac{y_1}{t},z_1\right)\prod_{i=1}^{v-1}p_B\left(\frac{1+\delta}{\Lambda_i^2};z_i,z_{i+1}\right)p_B\left(\frac{1+\delta}{\Lambda_v^2};z_v,0\right)\\
\leq t~p_B\left(2\tau_{1,\delta}+\sum_{i=1}^{v}\frac{t^2(1+\delta)}{\Lambda_i^2};y_1,0\right)\leq t~p_B\left(2\tau_{1,\delta}+\sum_{i=1}^{v}\frac{(1+\delta)}{\Lambda_i^2};y_1,0\right) ~,
\end{multline}
where we used (\ref{rr+}). Furthermore, using (\ref{in1}) we obtain for all $0\leq \delta<\delta'<1$ 
\begin{multline}\label{sure}
\left(\frac{y_1}{\sqrt{\tau_1}}\right)^{\gamma}p_B\left(2\tau_{1,\delta}+\sum_{i=1}^{v}\frac{1+\delta}{\Lambda_i^2};y_1,0\right)\\\leq ~C_{\delta,\delta'}~ \left(1+\sum_{i=1}^{v}\frac{\tau_1^{-\frac{1}{2}}}{\sqrt{2}\Lambda_i}\right)^{\gamma}p_B\left(2\tau_{1,\delta'}+\sum_{i=1}^{v}\frac{1+\delta'}{\Lambda_i^2};y_1,0\right).
\end{multline}
Since $\Lambda_i,~\Lambda_v\geq \Lambda$ for all $i\in\mathcal{I}$, we deduce that (\ref{sure}) is bounded by 
\begin{equation}
    C_{\delta,\delta'}~\max\left(\left(\frac{v}{\sqrt{2}}\right)^{\gamma},1\right) \left(1+\frac{\tau_1^{-\frac{1}{2}}}{\Lambda}\right)^{\gamma}p_B\left(2\tau_{1,\delta'}+\sum_{i=1}^{v}\frac{1+\delta'}{\Lambda_i^2};y_1,0\right)
\end{equation}
Proceeding similarly to (\ref{A8}), we deduce  
\begin{equation}\label{A24}
   p_B\left(2\tau_{1,\delta'}+\sum_{i=1}^{v}\frac{1+\delta'}{\Lambda_i^2};y_1,0\right)\leq   2^{v}~\mathcal{F}^0_{\delta'}\left(\Lambda_{\mathcal{I}},\Lambda_v;{\tau_{1}};T^{1,0}_l;{y_1}\right), 
\end{equation}
which together with (\ref{C22}) and (\ref{sure}) imply
\begin{equation}
    \left(\frac{y_1}{\sqrt{\tau_1}}\right)^{\gamma}\mathcal{F}^0_{\delta}\left(\Lambda_{\mathcal{I}};\frac{\tau_{1}}{t^2};T^{1,0}_l;\frac{y_1}{t}\right)\leq C~t\left(1+\frac{\tau_1^{-\frac{1}{2}}}{\Lambda}\right)^{\gamma+1}\mathcal{F}^0_{\delta'}\left(\Lambda_{\mathcal{I}};\tau_{1};T^{1,0}_l;y_1\right),s
\end{equation}
where $C:=2^{v}~C_{\delta,\delta'}\max\left(\left(\frac{v}{\sqrt{2}}\right)^{\gamma},1\right)$. Using 
\begin{equation}
\mathcal{F}^{\Lambda,0}_{1,l;\delta}\left(\tau_{1}\right):=\sum_{v_2=1}^{3l-1}\mathcal{F}^{0}_{\delta}\left(\Lambda,\tau_{1};T^{1,0}_l;y_1\right),
\end{equation}
we deduce 
\begin{equation}
    \left(\frac{y_1}{\sqrt{\tau_{1}}}\right)^{\gamma}\mathcal{F}^{0}_{1,l;\delta}\left(t,\tau_{1}\right)\leq O(1)~t\left(1+\frac{\tau^{-\frac{1}{2}}_1}{\Lambda}\right)^{\gamma}\mathcal{F}^{\Lambda,0}_{1,l;\delta'}(\tau_{1})~.
\end{equation}

\end{proof}
\begin{lemma}\label{lemma9}
Let $\Lambda \geq  3\sqrt{l}\tau^{-\frac{1}{2}}$, $0<\delta<\delta'\leq 1$ and $y_1,~y_2\in \mathbb{R}$. For $0< t, t'\leq 1$, $l\geq 1$ and $\gamma_1,\gamma_2\in\mathbb{N}$, we have 
\begin{multline}
    \left(\frac{y_1}{\sqrt{\tau_1}}\right)^{\gamma_1}\left(\frac{y_2}{\sqrt{\tau_2}}\right)^{\gamma_2}\mathcal{F}_{2,l;\delta}^{0}\left(\Lambda;\frac{\tau_{1}}{t^2},\frac{\tau_{2}}{t'^2};\frac{y_1}{t},\frac{y_2}{t'}\right)\\\leq ~tt'\mathcal{Q}\left(\frac{\tau^{-\frac{1}{2}}}{\Lambda}\right)\mathcal{F}_{2,l;\delta'}^{0}\left(\Lambda;\tau_{1,2};y_{1,2}\right).
\end{multline}
The polynomial $\mathcal{Q}$ has nonnegative coefficients which are independent of $\tau_{1}$, $\tau_2$ and $\Lambda$ but depend on $l$, $\delta$, $\delta'$, $\gamma_1$ and $\gamma_2$.
\end{lemma}
\begin{proof}
Using (\ref{weight}), $\mathcal{F}_{2,l;\delta}^{\Lambda,0}\left(\frac{\tau_{1}}{t^2},\frac{\tau_{2}}{t'^2};\frac{y_1}{t},\frac{y_2}{t'}\right)$ can be written as follows
\begin{multline}\label{C288}
    \sum_{T^{2,0}_l\in\mathcal{W}^2_l(\sigma_2)}\mathcal{F}^0_{\delta}\left(\Lambda;\frac{\tau_{1}}{t^2},\frac{\tau_{2}}{{t'}^2};T^{2,0}_l;\frac{y_1}{t},\frac{y_2}{t'}\right)\\
    +\left(\sum_{T^{1,0}_l\in\mathcal{W}^1_l(\pi_1)}\mathcal{F}^0_{\delta}\left(\Lambda,\frac{\tau_{1}}{t^2};T^{1,0}_l;\frac{y_{1}}{t}\right)\right)\cdot\left(\sum_{T^{1,0}_l\in\mathcal{W}^1_l(\pi_2)}\mathcal{F}^0_{\delta}\left(\Lambda,\frac{\tau_{2}}{{t'}^2};T^{1,0}_l;\frac{y_{2}}{t'}\right)\right).
\end{multline}
\begin{itemize}
    \item First, we prove  
    \begin{multline}\label{t00}
    \left(\frac{y_1}{\sqrt{\tau_1}}\right)^{\gamma_1}\left(\frac{y_2}{\sqrt{\tau_2}}\right)^{\gamma_2}\sum_{T^{2,0}_l\in\mathcal{W}^2_l(\sigma_2)}\mathcal{F}^{0}\left(\Lambda,\frac{\tau_{1}}{t^2},\frac{\tau_{2}}{{t'}^2};T^{2,0}_l;Y_{\sigma_2}(t,t')\right)\\\leq tt'\mathcal{Q}\left(\frac{\tau^{-\frac{1}{2}}}{\Lambda}\right) \mathcal{F}_{2,l;\delta'}^{\Lambda,0}\left(\tau_{1,2}\right),
    \end{multline}
where $0<\delta<\delta'<1$. Given a surface tree $T^{2,0}_l$ in $\mathcal{W}_{2,l}(\sigma_2)$, we have by Lemma \ref{lemmSi}  
\begin{multline}\label{t0}
  \mathcal{F}^0_{\delta}\left(\Lambda_{\mathcal{I}},\tilde{\Lambda};\frac{\tau_1}{t^2},\frac{\tau_2}{{t'}^2};T^{2,0}_l,Y_{\sigma_2}(t,t')\right)\\\leq \int_{0}^{\infty}dz~p_B\left(c_{1,\delta}(t);z,\frac{y_1}{t}\right)p_B\left(c_{2,\delta}(t');z,\frac{y_2}{t'}\right)p_B\left(\frac{1+\delta}{\tilde{\Lambda}_1^2};z,0\right)~,
\end{multline}
where $c_{1,\delta}(t)=c_1(t)(1+\delta)$and $c_{2,\delta}(t')=c_2(t')(1+\delta)$. The parameters $c_1(t)$, $c_2(t')$ and $\tilde{\Lambda}_1$ are given by (\ref{c1c2}) with $\tau_1\rightarrow \tau_1/t^2$ and $\tau_2\rightarrow \tau_2/t'^2$.
For $y_i\leq 0$, we have for all $z\in\mathbb{R}^+$
\begin{equation*}
    p_B\left(c_{i,\delta}(t);z,\frac{y_i}{t}\right)\leq p_B\left(c_{i,\delta}(t);z,-\frac{y_i}{t}\right).
\end{equation*}
    Therefore, without loss of generality, we consider in the sequel $y_1, y_2\geq 0$. For $y_1\leq y_2$, we write
    \begin{multline}
      \int_{0}^{\infty}dz~p_B\left(c_{1,\delta}(t);z,\frac{y_1}{t}\right)p_B\left(c_{2,\delta}(t');z,\frac{y_2}{t'}\right)p_B\left(\frac{1+\delta}{\tilde{\Lambda}_1^2};z,0\right)\\=\mathcal{I}(0,y_1)+\mathcal{I}(y_1,y_2)+\mathcal{I}(y_2,+\infty)~,
    \end{multline}
    where 
    $$\mathcal{I}(a,b):=\int_{a}^{b}dz~p_B\left(c_{1,\delta}(t);z,\frac{y_1}{t}\right)p_B\left(c_{2,\delta}(t');z,\frac{y_2}{t'}\right)p_B\left(\frac{1+\delta}{\tilde{\Lambda}_1^2};z,0\right).$$
    \begin{itemize}
        \item First, we bound $\mathcal{I}(0,y_1)$. For $0\leq t,t'\leq 1$, we have 
        \begin{align}
     p_B\left(c_{1,\delta}(t);z,\frac{y_1}{t}\right)&\leq t~p_B\left(c_{1,\delta};tz,{y_1}\right)~,\label{so1}
    \\
    p_B\left(c_{2,\delta}(t');z,\frac{y_2}{t'}\right)&\leq t'~p_B\left(c_{2,\delta};t'z,{y_2}\right)~,\label{sosoo}
        \end{align}
with $c_i:=c_i(1)$. For $0\leq z\leq y_1\leq y_2$, we also have 
\begin{align}
 p_B\left(c_{1,\delta};tz,{y_1}\right)&\leq t~p_B\left(c_{1,\delta};z,{y_1}\right)~,
    \\
    p_B\left(c_{2,\delta};t'z,{y_2}\right)&\leq t'~ p_B\left(c_{2,\delta};z,{y_2}\right)~.\label{soson}
        \end{align}
This implies 
\begin{equation}\label{a46}
    \mathcal{I}(0,y_1)\leq t~t'~\int_{0}^{\infty}dz~p_B\left(c_{1,\delta};z,{y_1}\right)p_B\left(c_{2,\delta};z,{y_2}\right)p_B\left(\frac{1+\delta}{\tilde{\Lambda}_1^2};z,0\right)~,
\end{equation} 
which again by Lemma \ref{lemmSi} is bounded by $$O(1)\int_{\vec {z}}\mathcal{F}^0_{\delta}\left(\Lambda_{\mathcal{I}},\tilde{\Lambda};\tau_{\sigma_2};T^{2,0}_l;\vec{z};Y_{\sigma_2}\right)\leq O(1)\mathcal{F}^{\Lambda,0}_{2,l;\delta}\left(\tau_{1,2}\right)~.$$
For $(\gamma_1,\gamma_2)\neq (0,0)$, we need to bound also the following term 
\begin{equation}\label{C37}
\left(\frac{y_1}{\sqrt{\tau_1}}\right)^{\gamma_1}\left(\frac{y_2}{\sqrt{\tau_2}}\right)^{\gamma_2}\mathcal{I}(0,y_1)~.   
\end{equation}
Using (\ref{a46}), (\ref{C37}) is bounded by 
\begin{multline}\label{l1}
   t~t'~\tau^{-\gamma_1-\gamma_2}~ \sum_{k=0}^{\gamma_1}\sum_{k'=0}^{\gamma_2}~\binom{\gamma_1}{k}~\binom{\gamma_2}{k'} ~ \int_{0}^{\infty}dz~|y_1-z|^k|y_2-z|^{k'} ~z^{\gamma_1+\gamma_2-k-k'}~\\\times p_B\left(c_{1,\delta};z,{y_1}\right)p_B\left(c_{2,\delta};z,{y_2}\right)p_B\left(\frac{1+\delta}{\tilde{\Lambda}_1^2};z,0\right). 
\end{multline}
Using (\ref{in1}), we obtain 
\begin{equation}
    z^{\gamma_1+\gamma_2-k-k'}p_B\left(\frac{1+\delta}{\tilde{\Lambda}^2};z,0\right)\leq C_{\delta,\delta'}~{\tilde{\Lambda}_1^{-\gamma_1-\gamma_2+k'+k}}p_B\left(\frac{1+\delta'}{\tilde{\Lambda}_1^2};z,0\right).
\end{equation}
Since $\Lambda_i,~\Tilde{\Lambda}\geq \Lambda$ for all $i\in\mathcal{I}$, we deduce that
\begin{equation}\label{cilt}
    \frac{1}{\Tilde{\Lambda}_1^2}=\left(\sum_{i=v_1+v_2+1}^v \frac{1}{\Lambda_i^2}\right)\left(1-\delta_{v,0}\right)+\frac{1}{\Tilde{\Lambda}^2}\leq \frac{v_0+1}{\Lambda^2}~~\mathrm{and}~~c_i\leq \tau_i+\frac{v_i}{\Lambda^2}~.
\end{equation}
This gives
\begin{equation}\label{l2}
    z^{\gamma_1+\gamma_2-k-k'}p_B\left(\frac{1+\delta}{\tilde{\Lambda}^2};z,0\right)\leq O(1)~{{\Lambda}^{-\gamma_1-\gamma_2+k'+k}}p_B\left(\frac{1+\delta'}{\tilde{\Lambda}_1^2};z,0\right).
\end{equation}
 Similarly, we have
\begin{align}
    |y_i-z|^{\gamma_i-k}p_B\left(c_{i,\delta};y_i,z\right)&\leq C'_{\delta,\delta'}~c_{i}^{\frac{k}{2}}p_B\left(c_{i,\delta'};y_i,z\right),\\
    &\leq O(1)~\tau_i^{\frac{k}{2}}\left(1+\frac{\tau_i^{-\frac{1}{2}}}{\Lambda}\right)p_B\left(c_{i,\delta'};y_i,z\right),\label{l3}
\end{align}
where we used (\ref{cilt}). Whenever it appears, $O(1)$ denotes a constant which depends on $\delta$, $\delta'$, $l$ and $v$.
Combining (\ref{l1}), (\ref{l2}) and (\ref{l3}), we deduce that (\ref{C37}) is bounded by 
\begin{multline}
    O(1)~\left(\sum_{k=0}^{\gamma_1+\gamma_2}\left(\frac{\tau^{-\frac{1}{2}}}{\Lambda}\right)^{k}\right)\int_{0}^{\infty}dz~ p_B\left(c_{1,\delta'};z,{y_1}\right)p_B\left(c_{2,\delta'};z,{y_2}\right)p_B\left(\frac{1+\delta'}{\tilde{\Lambda}_1^2};z,0\right).
\end{multline}
By Lemma \ref{lemmSi}, we deduce that (\ref{C37}) is bounded by
\begin{equation}\label{t1}
     t~t'~\mathcal{Q}\left(\frac{\tau^{-\frac{1}{2}}}{\Lambda}\right)\int_{\vec {z}}\mathcal{F}^0_{\delta}\left(\Lambda_{\mathcal{I}},\tilde{\Lambda};\tau_{\sigma_2};T^{2,0}_l;\vec{z};Y_{\sigma_2}\right).
\end{equation}
\item Using the bounds (\ref{so1}), (\ref{sosoo}) and (\ref{soson}), the term $\mathcal{I}(y_1,y_2)$ is bounded by
    \begin{equation}\label{a56}
   t~t'~ \int_{y_1}^{y_2}dz~p_B\left(c_{1,\delta};tz,{y_1}\right)p_B\left(c_{2,\delta};z,{y_2}\right)p_B\left(\frac{1+\delta}{\tilde{\Lambda}_1^2};z,0\right).
\end{equation}
For $z\geq y_1$, we have 
\begin{equation}
    p_B\left(\frac{1+\delta}{\tilde{\Lambda}_1^2};z,0\right)\leq p_B\left(\frac{2(1+\delta)}{\tilde{\Lambda}_1^2};z,0\right)\exp\left(-\frac{y_1^2\tilde{\Lambda}_1^2}{4(1+\delta)}\right).
\end{equation}
Knowing that $v_0\leq 3l-1$ together with \begin{equation}\label{C28}
\forall \Lambda_i\in\Lambda_{\mathcal{I}}~,~~~~\Lambda_i\geq \Lambda~,~~~~~~\tilde{\Lambda}\geq \Lambda~,
\end{equation}
and recalling (\ref{A8}), we obtain 
\begin{equation}\label{a63}
\tilde{\Lambda}\geq \frac{\Lambda}{\sqrt{3l}}\geq 3\tau^{-1}.
\end{equation}
where we also used $\Lambda\geq 3\sqrt{l}\tau^{-\frac{1}{2}}$. This implies
\begin{equation}
   \exp\left(-\frac{y_1^2\tilde{\Lambda}_1^2}{4(1+\delta)}\right)\leq \exp\left(-\frac{y_1^2{\Lambda}^2}{12(1+\delta)l}\right)\leq \exp\left(-\frac{y_1^2}{2(1+\delta)\tau_1}\right).
\end{equation}
Furthermore, we have 
\begin{equation}\label{special}
    p_B\left(c_{1,\delta};tz,{y_1}\right)\leq \frac{1}{\sqrt{2\pi c_{1,\delta}}}~.
\end{equation}
Combining (\ref{special}) with the fact that $c_1\geq \tau_1$, we deduce 
\begin{equation*}
    \mathcal{I}(y_1,y_2) \leq \int_{\mathbb{R}}dz~p_B\left(c_{2,\delta};z,{y_2}\right)p_B\left(\frac{2(1+\delta)}{\tilde{\Lambda}_1^2};z,0\right)p_B\left(c_{1,\delta};y_1,0\right)~,
\end{equation*}
and by (\ref{rr+}) we deduce that $\mathcal{I}(y_1,y_2)$ is  bounded by 
$$p_B\left(c_{2,\delta}+\frac{2(1+\delta)}{\tilde{\Lambda}_1^2};{y_2},0\right)~p_B\left(c_{1,\delta};{y_1},0\right)~.$$
Using the property (\ref{in1}) of the bulk heat kernel together with (\ref{C28}), we obtain 
\begin{multline}\label{b27}
    \left(\frac{y_1}{\sqrt{\tau_1}}\right)^{\gamma_1}\left(\frac{y_2}{\sqrt{\tau_2}}\right)^{\gamma_2}~p_B\left(c_{2,\delta}+\frac{2(1+\delta)}{\tilde{\Lambda}_1^2};{y_2},0\right)~p_B\left(c_{1,\delta};{y_1},0\right)\\
   \leq \mathcal{Q}\left(\frac{\tau_i^{-\frac{1}{2}}}{\Lambda}\right)~p_B\left(c_{2,\delta'}+\frac{2(1+\delta')}{\tilde{\Lambda}_1^2};{y_2},0\right)~p_B\left(c_{1,\delta'};{y_1},0\right)~,
\end{multline}
where $0<\delta<\delta'<1$. For $\Lambda \geq 3\sqrt{l}\tau^{-\frac{1}{2}}$ and $l\geq 1$, we have 
\begin{equation}\label{riloo}
    \forall \Lambda_i\in\Lambda_{\mathcal{I}},~~~~\Lambda_i\geq \Lambda\geq \sqrt{3}
    \tau_2^{-\frac{1}{2}}~,
\end{equation}
and this implies   
\begin{multline}\label{riloo2}
    \frac{1}{\Lambda_{v_1+v_2}^2}\leq \frac{\tau_2}{3},~~\frac{1}{\Lambda_{v_1+v_2-1}^2}\leq \frac{\tau_2}{3},\\\frac{1}{\tilde{\Lambda}_1^2}=\left(\sum_{i=v_1+v_2+1}^{v}\frac{1}{{\Lambda}_{i}^2}\right)\left(1-\delta_{v,0}\right)+\frac{1}{\tilde{\Lambda}^2}\leq~\frac{v_0+1}{\Lambda^2}\leq \frac{\tau_2}{3}~,
\end{multline}
where again we used that $v_0\leq 3l-1$.
Hence, we have 
\begin{multline}
    c_{2,\delta'}+\frac{2(1+\delta')}{\Tilde{\Lambda}_1^2}=\tau_{2,\delta'}+\sum_{i=1}^{v_2}\frac{1+\delta'}{\Lambda_{i+v_1}^2}+\frac{2(1+\delta')}{\tilde{\Lambda}_1^2}\\
    \leq 2\tau_{2,\delta'}+\sum_{i=1}^{v_2-2}\frac{1+\delta'}{\Lambda_{i+v_1}^2}+\frac{1+\delta'}{\tilde{\Lambda}_1^2}~,
\end{multline}
which gives
\begin{equation}\label{cst}
   p_B\left(c_{2,\delta'}+\frac{2(1+\delta')}{\tilde{\Lambda}_1^2};{y_2},0\right)\leq~\sqrt{2}~ p_B\left(2\tau_{2,\delta'}+\sum_{i=1}^{v_2-2}\frac{1+\delta'}{\Lambda_{i+v_1}^2}+\frac{1+\delta'}{\tilde{\Lambda}_1^2};y_2,0\right).
\end{equation}
(\ref{rr++}) together with Lemma \ref{lemmSi} gives 
\begin{multline}\label{ruru}
    p_B\left(2\tau_{2,\delta'}+\sum_{i=1}^{v_2-2}\frac{1+\delta'}{\Lambda_{i+v_1}^2}+\frac{1+\delta'}{\tilde{\Lambda}_1^2};y_2,0\right)\\\leq 
    O(1)~\int_{z_1}\cdots\int_{z_{v_2+v_0-1}}~p_B\left(2\tau_{2,\delta'};z_1,{y_2}\right)\prod_{i=2}^{v_2-1}p_B\left(\frac{1+\delta'}{\Lambda_{v_1+i-1}^2};z_i,z_{i-1}\right)\\\times\prod_{i=0}^{v_0-1}p_B\left(\frac{1+\delta'}{~~~\tilde{\Lambda}_{v_1+v_2+i}^2};z_{v_2+i},z_{v_2+i-1}\right)~p_B\left(\frac{1+\delta'}{\tilde{\Lambda}^2};z_{v_2+v_0-1},0\right)~.
\end{multline}
The RHS of (\ref{ruru}) corresponds to the integrated surface weight factor of a surface tree $T^{1,0}_l$ which has an external vertex $y_2$ and $v_2+v_0-1$ internal vertices which all are of incidence number $2$. This tree belongs to the set of forests $\mathcal{W}_{1,l}(\pi_2)$ if and only if
\begin{equation}\label{v00}
v_0+v_2-1 \leq 3l-2+\frac{1}{2}~\iff v_0+v_2-1 \leq 3l-2~.
\end{equation}
Since the tree $T^{2,0}_l\left(\frac{\tau_1}{t^2},\frac{\tau_2}{{t'}^2};\frac{y_1}{t},\frac{y_2}{t'}\right)$ is in the forest $\mathcal{W}_{2,l}(\sigma_2)$, $v_0$, $v_1$ and $v_2$ necessarily verify $$v_0+v_2-1\leq 3l-2-v_1~,$$
which implies (\ref{v00}).
Hence, $T^{1,0}_l$ belongs to the set $\mathcal{W}_{1,l}(\pi_2)$. From (\ref{ruru}), we deduce 
\begin{equation}\label{k2}
    p_B\left(c_{2,\delta'}+\frac{1+\delta'}{\tilde{\Lambda}^2};{y_2},0\right)\leq O(1)~\mathcal{F}_{1,l;\delta'}^{\Lambda,0}\left(\tau_{2}\right).
\end{equation}
Proceeding similarly for $p_B\left(c_{1,\delta'};{y_1},0\right)$, we obtain 
\begin{equation}\label{k1}
    p_B\left(c_{1,\delta'};{y_1},0\right)\leq O(1)~\mathcal{F}_{1,l;\delta'}^{\Lambda,0}\left(\tau_{1}\right).
\end{equation}
(\ref{a56}) and (\ref{b27}) together with (\ref{k2}) give  
\begin{equation}\label{t2}
   \left(\frac{y_1}{\sqrt{\tau_1}}\right)^{\gamma_1}\left(\frac{y_2}{\sqrt{\tau_2}}\right)^{\gamma_2} \mathcal{I}(y_1,y_2)\leq t~t' \mathcal{Q}\left(\frac{\tau^{-\frac{1}{2}}}{\Lambda}\right)\mathcal{F}_{1,l;\delta'}^{\Lambda,0}\left(\tau_{1}\right)\mathcal{F}_{1,l;\delta'}^{\Lambda,0}\left(\tau_{2}\right),
\end{equation}
where all the constants are absorbed in the coefficients of the polynomial $\mathcal{Q}$.
\item The last term to bound is 
\begin{equation*}
   \mathcal{I}(y_2,+\infty):= \int_{y_2}^{\infty}dz~p_B\left(c_{1,\delta}(t);z,\frac{y_1}{t}\right)p_B\left(c_{2,\delta}(t');z,\frac{y_2}{t'}\right)p_B\left(\frac{1+\delta}{\tilde{\Lambda}_1^2};z,0\right).
\end{equation*}
For $z\geq y_2\geq y_1$ we have 
\begin{equation*}
    p_B\left(\frac{1+\delta}{\tilde{\Lambda}_1^2};z,0\right)\leq \frac{\tilde{\Lambda}_1}{\sqrt{2\pi}}\exp\left(-\frac{z^2\tilde{\Lambda}_1^2}{6(1+\delta)}\right)\exp\left(-\frac{y_1^2}{2c_{1,\delta}}\right)\exp\left(-\frac{y_2^2}{2c_{2,\delta}}\right),
\end{equation*}
where we used (\ref{a63}) and $c_i\geq \tau_i$. \\
Bounding respectively $ p_B\left(c_{1,\delta}(t);z,\frac{y_1}{t}\right)$ and $ p_B\left(c_{2,\delta}(t');z,\frac{y_2}{t'}\right)$ by $\frac{t}{\sqrt{2\pi c_{1,\delta}}}$ and $\frac{t'}{\sqrt{2\pi c_{2,\delta}}}$ we deduce that $\mathcal{I}(y_2,+\infty)$ is bounded by
\begin{equation*}
    O(1)~t~t'~p_B\left(c_{1,\delta};y_1,0\right)~p_B\left(c_{2,\delta};y_2,0\right).
\end{equation*}
Using the bound (\ref{in1}), we find
\begin{multline}
    \left(\frac{y_1}{\sqrt{\tau_1}}\right)^{\gamma_1}\left(\frac{y_2}{\sqrt{\tau_2}}\right)^{\gamma_2}~p_B\left(c_{2,\delta};{y_2},0\right)~p_B\left(c_{1,\delta};{y_1},0\right)\\
   \leq \mathcal{Q}\left(\frac{\tau_i^{-\frac{1}{2}}}{\Lambda}\right)~p_B\left(c_{2,\delta'};{y_2},0\right)~p_B\left(c_{1,\delta'};{y_1},0\right),
\end{multline}
where $0<\delta<\delta'<1$. Using (\ref{stot}), we deduce 
\begin{equation}\label{k1}
    p_B\left(c_{i,\delta'};{y_i},0\right)\leq O(1)~\mathcal{F}_{1,l;\delta'}^{\Lambda,0}\left(\tau_{i}\right)~,
\end{equation}
which implies 
\begin{equation}\label{t3}
   \left(\frac{y_1}{\sqrt{\tau_1}}\right)^{\gamma_1}\left(\frac{y_2}{\sqrt{\tau_2}}\right)^{\gamma_2} \mathcal{I}(y_2,+\infty)\leq t~t' \mathcal{Q}\left(\frac{\tau^{-\frac{1}{2}}}{\Lambda}\right)\mathcal{F}_{1,l}^{\Lambda,0}\left(\tau_{1,\delta'}\right)\mathcal{F}_{1,l}^{\Lambda,0}\left(\tau_{2,\delta'}\right).
\end{equation}
    \end{itemize}
Combining (\ref{t1}), (\ref{t2}) and (\ref{t3}) together with (\ref{borneweights2}) and (\ref{t0}), we obtain (\ref{t00}).
\item By definition, we have 
\begin{equation}
    \mathcal{F}_{1,l;\delta}^{\Lambda,0}(t,\tau_{1})=\sum_{T^{1,0}_l\in\mathcal{W}^1_l(\pi_1)}\mathcal{F}^0_{\delta}\left(\Lambda,\frac{\tau_{1}}{t^2};T^{1,0}_l;\frac{y_{1}}{t}\right).
\end{equation}
Applying Lemma \ref{lemma8} to the global surface weight factors $\mathcal{F}_{1,l;\delta}^{\Lambda,0}(t,\tau_{1})$ and $\mathcal{F}_{1,l;\delta}^{\Lambda,0}(t',\tau_{2})$ we obtain 
\begin{multline}
   \left(\frac{y_1}{\sqrt{\tau_1}}\right)^{\gamma_1} \sum_{T^{1,0}_l\in\mathcal{W}^1_l(\pi_1)}\mathcal{F}^0_{\delta}\left(\Lambda,\frac{\tau_{1}}{t^2};T^{1,0}_l;\frac{y_{1}}{t}\right) \\\times\left(\frac{y_2}{\sqrt{\tau_2}}\right)^{\gamma_2}\sum_{T^{1,0}_l\in\mathcal{W}^1_l(\pi_2)}\mathcal{F}^0_{\delta}\left(\Lambda,\frac{\tau_{2}}{{t'}^2};T^{1,0}_l;\frac{y_{2}}{t'}\right)\\
   \leq O(1)~t~t' \left(1+\frac{\tau_1^{-\frac{1}{2}}}{\Lambda}\right)^{\gamma_1}\left(1+\frac{\tau_2^{-\frac{1}{2}}}{\Lambda}\right)^{\gamma_2}\mathcal{F}_{1,l;\delta'}^{\Lambda,0}\left(\tau_{1}\right)\mathcal{F}_{1,l;\delta'}^{\Lambda,0}\left(\tau_{2}\right),
\end{multline}
and this together with (\ref{borneweights2}), (\ref{C288}) and  (\ref{t00}) conclude the proof of Lemma \ref{lemma9}.
\end{itemize}
\end{proof}

\section*{References}
\nocite{*}
\bibliographystyle{abbrv}
\bibliography{aipsamp}

\end{document}